\documentclass{article}
\usepackage{graphicx,amsfonts,amsmath,amssymb,amsthm,physics,xcolor,mathrsfs,mathtools,qcircuit} 
\usepackage{authblk}
\usepackage[colorlinks,linkcolor=blue,citecolor=blue,urlcolor=magenta,filecolor=cyan]{hyperref}
\usepackage{appendix}
\newtheorem{proposition}{Proposition}
\newtheorem{corollary}{Corollary}

\newtheorem{definition}{Definition}
\newtheorem{lemma}{Lemma}
\newtheorem{theorem}{Theorem}[section]

\usepackage{geometry}
\usepackage{bm}
\usepackage[capitalise]{cleveref}
\usepackage{algorithm}
\usepackage{algpseudocode}
\newcommand{\CO}{\mathcal{O} }

\title{Efficient Optimal Control of Open Quantum Systems}

\author[1,2]{Wenhao He\thanks{wenhaohe@mit.edu}}
\author[3]{Tongyang Li\thanks{tongyangli@pku.edu.cn}}
\author[4]{Xiantao Li\thanks{xiantao.li@psu.edu}}
\author[5]{Zecheng Li\thanks{zxl5523@psu.edu}}
\author[5]{Chunhao Wang\thanks{cwang@psu.edu}}
\author[4]{Ke Wang\thanks{wangke.math@psu.edu}}

\affil[1]{Center for Computational Science and Engineering, MIT}
\affil[2]{School of Physics, Peking University}
\affil[3]{Center on Frontiers of Computing Studies, School of Computer Science, Peking University}
\affil[4]{Department of Mathematics, Pennsylvania State University}
\affil[5]{Department of Computer Science and Engineering, Pennsylvania State University}

\date{}

\begin{document}

\maketitle

\begin{abstract}
  The optimal control problem for open quantum systems can be formulated as a time-dependent Lindbladian that is parameterized by a number of time-dependent control variables. Given an observable and an initial state, the goal is to tune the control variables so that the expected value of some observable with respect to the final state  is maximized. In this paper, we present algorithms for solving this optimal control problem efficiently, i.e., having a poly-logarithmic dependency on the system dimension, which is exponentially faster than best-known classical algorithms. Our algorithms are hybrid, consisting of both quantum and classical components. The quantum procedure simulates time-dependent Lindblad evolution that drives the initial state to the final state, and it also provides access to the gradients of the objective function via quantum gradient estimation. The classical procedure uses the gradient information to update the control variables.

At the technical level, we provide the first (to the best of our knowledge) simulation algorithm for time-dependent Lindbladians with an $\ell_1$-norm dependence. As an alternative, we also present a simulation algorithm in the interaction picture to improve the algorithm for the cases where the time-independent component of a Lindbladian dominates the time-dependent part. On the classical side, we heavily adapt the state-of-the-art classical optimization analysis to interface with the quantum part of our algorithms. Both the quantum simulation techniques and the classical optimization analyses might be of independent interest.
\end{abstract}

\newpage
\section{Introduction}
The ability to control the dynamics of a quantum system to maximize its property has been a persistent pursuit in quantum physics and chemistry \cite{d2021introduction}. This endeavor has recently gained momentum, spurred by the growing interest in designing quantum information processing devices.
One remarkable obstacle in controlling a quantum system's behavior stems from the reality that quantum systems typically evolve in the presence environmental noise. Consequently, the control strategy must take into account system/bath interactions.  In the Markovian regime, this problem can be formulated as an optimal control problem based on the Lindblad master equation \cite{lindblad1976generators,gorini1976completely} acting on $n$ qubits,
\begin{equation}
\label{eq:Lindblad}
    \frac{d}{dt}\rho =\mathcal{L}(t)(\rho):= -i\biggl[H_0+ \sum_{\beta=1}^{n_{\mathrm{c}}} u_\beta(t)\mu_\beta,\rho\biggr]+\sum_{j=1}^{m} \left(L_j\rho L_j^{\dag}-\frac12\{L_j^{\dag}L_j,\rho\}\right),
\end{equation}
in conjunction with a control functions $u_\beta(t)$ that enters the system Hamiltonian through the operator $\mu_\beta$, and we have $n_{\mathrm{c}}$ control functions. Here $\rho$ is a density operator  on $n$ qubits, and the second term in \cref{eq:Lindblad} is a result of system/bath interactions with $L_j$'s being the jump operators.  The quantum optimal control (QOC) is then formulated as an optimization problem following \cite{altafini2004coherent}:
\begin{equation}\label{qoc-opt}
    \max_{\bm u} f[\bm u(t)],  \quad  f[\bm u(t)]\coloneqq \text{tr}\big(\mathscr{O}\rho(T)\big)- \alpha \sum_{\beta=1}^{n_{\mathrm{c}}}\int_0^T\abs{u_\beta(t)}^2\mathrm{d}t.
\end{equation}
The Hermitian operator $\mathscr{O}$ represents the property to be maximized. The term $\bm u(t)$ embodies all the control variables $\{u_\beta\}$ and the last term in the objective function $f[\bm u(t)]$ is regarded as a regularization.
It is worthwhile to point out that there are other choices of the objective function \cite{brif2010control} in the formulation of the QOC problem. For example, one can guide the Lindblad dynamics \eqref{eq:Lindblad} toward a target state $\bar{\rho}(T)$. In this case, one can minimize the difference between $\bar{\rho}(T)$ and ${\rho}(T)$,
\begin{equation}\label{gate-optm}
    \min_{\bm u} f[\bm u(t)],  \quad  f[\bm u] := \norm{\rho(T)-\bar{\rho}(T)}^2+ \alpha \sum_\beta \int_0^T \abs{u_\beta(t)}^2\mathrm{d}t.
\end{equation}
Implicit in both optimization problems \cref{qoc-opt,gate-optm} is that $\rho(T)$ has to be obtained from the Lindblad equation \eqref{eq:Lindblad}. Thus the main computational challenge comes from the repeated computation of the solution of the Lindblad equation. In this paper, we mainly focus on the optimal control problem with the objective function \cref{qoc-opt}.

To be able to clearly illustrate the computational complexity,  we assume that the Hamiltonian $H(t)$ and the jump operators $L_j(t)$'s are sparse. Moreover, the sparsity structure for each operator does not change over time (i.e., the positions of nonzero entries do not change with time). For a sparse matrix $A$, we assume we have access to a procedure $\mathcal{P}_A$ that can apply the following oracles: 
\begin{align}
  \label{eq:oh}
  \mathcal{O}_{A, \mathrm{loc}}\ket{i, j} &= \ket{i, \nu_A(i, j)}, \quad \text{and} \\
  \label{eq:ohval}
  \mathcal{O}_{A, \mathrm{val}}\ket{t, i, j, z} &= \ket{t, i, j, z\oplus A_{i,j}(t)},
\end{align}
where $\nu_A(i, j)$ is the index of the $j$'s nonzero entry of column $i$ in $A$. 
Particularly for the optimal control problem, we assume we have access to $\mathcal{P}_{H_0}$, $\mathcal{P}_{\mu_\beta}$, and $\mathcal{P}_{L_j}$ for all $j \in [m]$, as well as $\mathcal{P}_{\mathscr{O}}$ for the observable $\mathscr{O}$.

\paragraph*{Main contributions}
We will present a hybrid quantum/classical algorithm for the QOC problem \eqref{eq:Lindblad} and \eqref{qoc-opt}. The overall algorithm consists of the following elements:
\begin{enumerate}
    \item A Lindblad simulation algorithm \cite{childs2016efficient,cleve2016efficient,li2022simulating} that prepares $\rho(T)$ in a purification form. The complexity of our algorithm exhibits a linear scaling with respect to $T$ with a scaling factor proportional to the $L_1$ norm of the Lindbladians instead of the $L_\text{max}$ norm. The dependence of the complexity on the precision $\epsilon$ is only poly-logarithmic. Alternatively, we can also simulate time-dependent Lindbladian using interaction picture~\cite{low2018hamiltonian}. This algorithm applies to important models in experimental physics. For instance, in an ion trap system, it is common to have a time-independent Hamiltonian with norm much larger than the rest of the Lindbladian terms, and thus our algorithm can make the simulation more efficient.
    
    \item The construction of a quantum phase oracle of the gradient of the function $f$. This is achieved by incorporating the quantum gradient computation algorithms in~\cite{gilyen2019optimizing}. This phase oracle will then be interfaced with a classical optimization algorithm.
    
    \item Having approximates of gradients $\nabla f( \bm  u(t))$, we use an accelerated gradient descent (AGD) method~\cite{jin2017accelerated} to solve the optimization problem. In particular, we analyze the influence of the statistical error from the gradient estimation and provide a precise complexity analysis for solving the optimization problem, which essentially characterizes the robustness of AGD for reaching second-order stationary points and may be of independent interest.

\end{enumerate}

In addition to the proposed algorithms, we provide rigorous analysis of the numerical error and precise overall complexity estimates for the hybrid algorithm. Formally, we establish the following result for optimal control of open quantum systems:

\begin{theorem}[main theorem] \label{thm:main}
   Assume there are $n_{\mathrm{c}}$ control functions  $u_\beta(t) \in \mathrm{C}^2([0,T])$. Further assume\footnote{More generally, if $\norm{H_0}, \norm{\mu} = \Theta(\Lambda)$, it is equivalent to enlarge the time duration $T$ by a factor $O(\Lambda)$.} that $\norm{H_0}, \norm{\mathscr{O}}, \norm{\mu_{\beta}}, \norm{L_j} \leq 1$, and $\alpha \geq 2/T$. 
There exists a quantum algorithm that, with probability at least $2/3$\footnote{Using standard techniques, the success probability can be boosted to a constant arbitrarily close to 1 while only introducing a logarithmic factor in the complexity.}, solves problem \eqref{qoc-opt} by:
\begin{itemize}
\item reaching a first-order stationary point $\|\nabla{f}\|<\epsilon$ with \eqref{eq:Lindblad} using
      $\widetilde{O}\left(\frac{n_{\mathrm{c}} \|\mathcal{L}\|_{\text {be }, 1}T}{\epsilon^{23/8}}\Delta_f\right)$
  queries to $\mathcal{P}_{H_0}$ and $\mathcal{P}_{\mu_\beta}$, $\beta=1,2,\ldots, n_{\mathrm{c}}$, and
    $\widetilde{O}\left(mn\frac{n_{\mathrm{c}} \|\mathcal{L}\|_{\text {be }, 1} T}{\epsilon^{23 / 8}} \Delta_f + n\frac{ T^{3/2}}{\epsilon^{9 / 4}} \Delta_f\right)$
  additional 1- and 2-qubit gates; or
  
\item reaching a second-order stationary point using
      $\widetilde{O}\left(\frac{n_{\mathrm{c}}  \|\mathcal{L}\|_{\text {be }, 1}T^{7/4}}{\epsilon^{5}}\Delta_f\right)$
  queries to $\mathcal{P}_{H_0}$ and $\mathcal{P}_{\mu_\beta}$, $\beta=1,2,\ldots, n_{\mathrm{c}}$ and
    $\widetilde{O}\left(mn\frac{n_{\mathrm{c}}  \|\mathcal{L}\|_{\text {be }, 1}T^{7/4}}{\epsilon^{5}} \Delta_f + n\frac{ T^{3/2}}{\epsilon^{9 / 4}} \Delta_f\right)$
  additional 1- and 2-qubit gates.
  \end{itemize} 
  Here $n_{\mathrm{c}}$ and $m$ are respectively the number of control variables and jump operators.
  \end{theorem}

\paragraph*{Techniques}
Our technical contributions are outlined as follows.
\begin{itemize}
\item In \Cref{sec:simulation}, we give efficient quantum algorithms for simulating time-dependent Lindbladians with a scaling factor in time proportional to the $L_1$-norm of the Lindbladians instead of the $L_\text{max}$-norm, as well as poly-logarithmic $\epsilon$ dependence. Our simulation algorithm is based on the higher-order series expansion from Duhamel's principle as sketched in~\cite{li2022simulating}. A notable difference from~\cite{li2022simulating} is that in their paper, Gaussian quadratures are used to approximate integrals; however, in our time-dependent case, Gaussian quadratures can no longer be used as, unless upper bounds on the higher-order derivatives of the operators are given in advance. The techniques for obtaining the $L_1$-norm dependence follow from the rescaling trick in~\cite{BCSWW20}, while generalized to Lindbladians. Our time-dependent Lindbladian simulation techniques might be of independent interest.

\item In \Cref{sec:interaction}, we show how to simulate time-dependent Lindbladian using interaction picture~\cite{low2018hamiltonian}. This technique is suited for simulating a Lindbladian $\mathcal{L}=\mathcal{L}_{1}+\mathcal{L}_{2}$ where $\mathcal{L}_1(\cdot ) = -i\left[H_1, \cdot\right]$ is a Hamiltonian with complexity linear in norm of $\mathcal{L}_{2}$ (up to poly-logarithmic factors) and similar number of simulations of the Hamiltonian $H_1$ . The simulation scheme is based on a mathematical treatment of the Lindblad equation as a differential equation, and the construction leverages the simulation algorithms shown in \cref{sec:simulation} without rescaling the evolution time. It turns out that using our simulation algorithm in the interaction picture, we obtain better gate complexity compared with directly using the simulation in \cref{sec:simulation} even with the $\ell_1$-norm dependence. To the best of our knowledge, this is the first Lindblad simulation algorithm in the interaction picture, which can also be of independent interest.

\item In \Cref{sec:optimization}, we adapt a nonconvex optimization algorithm that can reach first-order stationary points with $\tilde{O}(1/\epsilon^{7/4})$ \emph{noisy} gradient queries with $\ell_{2}$-norm error at most $O(\epsilon^{9/8})$, and reach second-order stationary points with $\tilde{O}(1/\epsilon^{7/4})$ \emph{noisy} gradient queries with $\ell_{2}$-norm error at most  $\tilde{O}(\epsilon^{3})$. Our setting is different from either gradient descent (GD) or stochastic gradient descent (SGD): Compared to GD we only have access to noisy gradients, while in standard SGD the noise can be adjusted and there is no Lipschitz condition for the noisy gradient. With this novel setting, we successfully design an optimization algorithm based on perturbed accelerated gradient descent (PAGD) \cite{jin2017accelerated}. We carefully analyze the error bound in different cases and it turns out that our algorithm reaches an optimal error scaling for PAGD (up to poly-logarithmic factors) in finding a first-order stationary point.
\end{itemize}

\paragraph*{Related work}
In addition to the large variety of conventional applications \cite{geppert2004laser}, quantum optimal control problems are crucial in near-term quantum computing, because in the architecture of quantum computers, the underlying physical operations such as microwave control pulses and the modulated laser beam can be abstracted as control pulse sequences (see the survey~\cite{shi2020resource} for more detailed discussions), and hence the are inherently quantum control problems. Quantum optimal control also plays a vital role  in quantum computing algorithms. For instance, Magann et al.~\cite{magann2021pulses} studied the relationship between variational quantum algorithms (VQAs) and quantum optimal control, and showed that the performance of VQAs can be informed by quantum optimal control theory. Banchi and Crooks~\cite{banchi2021measuring} demonstrated how gradients can be estimated in a hybrid quantum-classical optimization algorithm, and quantum control is used as one important application.
In Ref.~\cite{lloyd2014information} the authors showed that for a quantum many-body system, if it exists an efficient classical representation, then the optimal control problems on this quantum dynamics can be solved efficiently with finite precision.

There exist heuristic classical methods for solving the quantum optimal control problem, including the monotonically convergence algorithms~\cite{zhu1998rapidly}, the Krotov method~\cite{palao2003optimal}, the GRadient Ascent Pulse Engineering (GRAPE) algorithm~\cite{khaneja2005optimal,de2011second}, the Chopped RAndom-Basis (CRAB) algorithm~\cite{caneva2011chopped}, etc. Furthermore, such heuristics can be extended to quantum optimal control of open quantum systems~\cite{koch2016controlling,koch2022quantum,reich2015efficient,goerz2015optimizing}, including~\cite{abdelhafez2019gradient,li2009pseudospectral,boutin2017resonator}. However, these algorithms do not establish provable guarantees for the efficiency of solving the quantum optimal control problem. Meanwhile, the landscape of the quantum control problem has been analyzed in~\cite{chakrabarti2007quantum,de2013closer,ge2021optimization}, which suggests that for closed quantum systems, the landscape may not involve suboptimal optimizers. However, the implication to the computational complexity still remains open.

Quantum algorithms, due to their natural ability to simulate quantum dynamics, have been developed for quantum control problems~\cite{magann2021digital,li2017hybrid,castaldo2021quantum,li2023efficient}. Liu and Lin \cite{liu2023dense} 
developed an efficient algorithm to output the integral of the observable in \cref{qoc-opt}, which can potentially solve a more generalized optimal control problem.  
These approaches  employ hybrid quantum-classical algorithms that combine a quantum algorithm for the time-dependent Schr\"odinger equation with a classical optimization method. However, these
efforts have been focused on closed quantum systems, and quantum control algorithms for open quantum systems require separate techniques.

\paragraph{Open questions}
Our paper leaves several open questions for future investigations:
\begin{itemize}
\item Are there efficient quantum algorithms for the optimal control of other master equations beyond the Lindbladian equation?

\item How to extend the current framework to the control problems with a target density operator $ \bar{\rho}(T)$? The challenge in such a control problem \eqref{gate-optm} is the estimation of the Frobenius norm from the quantum circuit. 

\item Gaussian quadrature was used in the Lindblad simulation method \cite{li2022simulating}, which significantly suppressed the number of terms in a Dyson-series type of approach, and implies the implementation. The extension of Gaussian quadrature to the current framework with time-dependent Lindbladians would require derivative bounds for the evolution operator from both the drift and jump terms, which is not trivial.   
\end{itemize}

\section{Preliminaries}
\subsection{Notations}
For a positive integer $m$, we use $[m]$ to denote the set $\{1, \ldots, m\}$. In this paper, we use two types of notations to denote vectors. For a quantum state, we use the Dirac notation $\ket{\cdot}$ to denote the corresponding state vector. For vectors involved in classical information, e.g., the gradient vector, we use bold font, such as $\bm{u}$, to denote them. For such a vector $\bm{u} \in \mathbb{C}^d$, we use subscripts with a norm font to indicate its entries, i.e., $u_1, \ldots, u_d$ are the entries of $\bm{u}$. When we use subscripts with a bold font, such as, $\bm{u}_1, \ldots, \bm{u}_k$, they are a list of vectors. For a vector $\bm{v} \in \mathbb{C}^d $, we use $\norm{\bm{v}}$ to denote its \emph{Euclidean norm}. For a matrix $M \in \mathbb{C}^{d\times d}$, we use $\norm{M}$ to denote its \emph{spectral norm}, and use $\norm{M}_1$ to denote its \emph{trace norm}, i.e., $\norm{A}_1 = \Trace(\sqrt{M^{\dag}M})$. We also use $[\cdot, \cdot]$ to denote the operator \emph{commutator}, i.e., $[A, B] \coloneqq AB-BA$, and use $\{\cdot,\cdot\}$ to denote the \emph{anticommutator}, i.e., $\{A, B\} \coloneqq AB+BA$.

In addition, we use calligraphic fonts, such as $\mathcal{L}$, to denote  \emph{superoperators}, which is also referred to as \emph{quantum maps}. Superoperators are linear maps that take matrices to matrices. The \emph{induced trace norm} of a superoperator $\mathcal{M}$, denoted by $\norm{\mathcal{M}}_1$, is defined as
\begin{align}
  \norm{\mathcal{M}}_1 \coloneqq \max\{\norm{\mathcal{M}(A)}_1 : \norm{A}_1 \leq 1\}.
\end{align}
The \emph{diamond norm} of a superoperator $\mathcal{M}$, denoted by $\norm{\mathcal{M}}_{\diamond}$, is defined as
\begin{align}
  \norm{\mathcal{M}}_{\diamond} \coloneqq \norm{\mathcal{M}\otimes \mathcal{I}}_1,
\end{align}
where $\mathcal{I}$ acts on the space with the same size as the space $\mathcal{M}$ acts on.

We denote by $\mathrm{C}^2[0, T]$ the class of twice continuously differential functions in $[0, T]$.

\subsection{Algorithmic tools}
\subsubsection{Block-encoding and implementing completely-positive maps}
Although we assume that the input of the operators of the Lindbladian are given by sparse-access oracles, it is convenient to use a more general input model when presenting the simulation algorithm. For a matrix $A \in \mathbb{C}^{2^n\times 2^n}$, we say that a unitary, denoted by $U_A$, is an $(\alpha, b, \epsilon)$-block-encoding of $A$ if
\begin{align}
  \norm{A - \alpha(\bra{0}^{\otimes b}\otimes I)U_A(\ket{0}^{\otimes b}\otimes I)} \leq \epsilon,
\end{align}
where the identity operator $I$ is acting on $n$ qubits. Intuitively, this unitary $U_A$ is acting on $(n+b)$ qubits and $A$ appears in the upper-left block of it, i.e.,
\begin{align}
  U_A =
  \begin{pmatrix}
    A/\alpha & \cdot\\
    \cdot & \cdot
  \end{pmatrix}.
\end{align}
Here, we refer to $\alpha$ as the \emph{normalizing factor}.

Our simulation algorithm relies on the following technical tool from~\cite{Li_2023} for implementing completely positive maps given the block-encodings of its Kraus operators, which generalizes a similar tool in~\cite{cleve2016efficient} where the Kraus operators are given as linear combinations of unitaries.

\begin{lemma}[Implementing completely positive maps via block-encodings of Kraus operators~\cite{Li_2023}]
  \label{lemma:block-encoding-channel}
  Let $A_1, \ldots, A_{m} \in \mathbb{C}^{2^n}$ be the  Kraus operators of a completely positive map. Let $U_1, \ldots, U_m \in \mathbb{C}^{2^{n+n'}}$ be their corresponding $(s_j, n', \epsilon)$-block-encodings, i.e.,
  \begin{align}
    \norm{A_j - s_j (\bra{0}\otimes I)U_j\ket{0}\otimes I)} \leq \epsilon, \quad \text{ for all $1 \leq j \leq m$}.
  \end{align}
  Let $\ket{\mu}\coloneqq \frac{1}{\sqrt{\sum_{j=1}^{m}s_j^2}}\sum_{j=1}^{m}s_j\ket{j}$. Then $(\sum_{j=1}^m\ketbra{j}{j}\otimes U_j)\ket{\mu}\ket{0}\otimes I$ implements this completely positive map in the sense that
  \begin{align}\label{eq: M-apply}
    \norm{I\otimes \bra{0}\otimes I \left(\sum_{j=1}^m\ketbra{j}{j}\otimes U_j\right) \ket{\mu}\ket{0}\ket{\psi} - \frac{1}{\sqrt{\sum_{j=1}^ms_j^2}}\sum_{j=1}^m \ket{j}A_j\ket{\psi}} \leq \frac{m\epsilon}{\sqrt{\sum_{j=1}^ms_j^2}}
  \end{align}
  for all $\ket{\psi}$.
\end{lemma}

We also need the following lemma from~\cite{Li_2023} for obtaining a block-encoding of a linear combination of block-encodings.
\begin{lemma}[Block-encoding of a sum of block-encodings~\cite{Li_2023}]
  \label{lemma:sum-to-be}
  Suppose $A \coloneqq \sum_{j=1}^m y_j A_j \in \mathbb{C}^{2^n\times 2^n}$, where $A_j \in \mathbb{C}^{2^n \times 2^n}$ and $y_j > 0$ for all $j \in \{1, \ldots m\}$. Let $U_j$ be an $(\alpha_j, a, \epsilon)$-block-encoding of $A_j$, and $B$ be a unitary acting on $b$ qubits (with $m \leq 2^b-1$) such that $B\ket{0} = \sum_{j=0}^{2^b-1}\sqrt{\alpha_jy_j/s}\ket{j}$, where $s = \sum_{j=1}^my_j\alpha_j$. Then a $(\sum_{j}y_j\alpha_j, a+b, \sum_{j}y_j\alpha_j\epsilon)$-block-encoding of $\sum_{j=1}^my_jA_j$ can be implemented with a single use of $\sum_{j=0}^{m-1}\ketbra{j}{j}\otimes U_j + ((I - \sum_{j=0}^{m-1}\ketbra{j}{j})\otimes I _{\mathbb{C}^{2^a}}\otimes I_{\mathbb{C}^{2^{n}}})$ plus twice the cost for implementing $B$.
\end{lemma}

\subsubsection{Optimization}

For the current quantum-classical hybrid algorithm, we will couple a Lindblad simulation with a classical optimization algorithm. For this purpose, we work with the PAGD algorithm \cite{jin2017accelerated}, which is based on Nesterov’s accelerated gradient descent idea \cite{nesterov1983method},
\begin{equation}
\begin{aligned}
     \bm u_{k+1} &= \bm u_k  - \eta \nabla f(\bm u_k ) + (1 - \theta) \bm v_k, \quad 
     \bm v_{k+1} &= \bm u_{k+1} - \bm u_k.
\end{aligned}
\end{equation}
Here $\bm u_k$ is the $k$th iterate of the control variable. The idea in PAGD is to introduce a perturbation to the iterate when $\norm{\nabla f} > \epsilon$ for some iterations, along with a  negative curvature exploitation step. 

There are two common goals for solving (nonconvex) optimization problems:
\begin{itemize}
\item $\bm x$ is called an $\epsilon$-approximate first-order stationary point if $\|\nabla f(\bm x)\|\leq\epsilon$. 
\item $\bm x$ is called an $\epsilon$-approximate second-order stationary point if $\|\nabla f(\bm  x)\|\leq\epsilon,\quad\lambda_{\min}(\nabla^{2}f(\bm x))\geq-\sqrt{\varrho\epsilon}$. Here $f$ is a $\varrho$-Hessian-Lipschitz function, i.e., $\|\nabla^{2}f(\bm x)-\nabla^{2}f(\bm y)\|\leq\varrho\|\bm x-\bm y\|$ for any $\bm x$ and $\bm y$.
\end{itemize}

\subsubsection{Quantum gradient estimation}\label{sec:gradient-est}
With copies of $\rho(T)$, which will be obtained from Lindblad simulation algorithms,  and sparse access to $\mathscr{O}$, we can obtain an estimated gradient value of $\widetilde{J}_1(\boldsymbol{u})$. The high-level strategy is to construct a probability oracle first, then construct a phase oracle with the probability oracle, and finally obtain the gradient by the phase oracle. The probability oracle and the phase oracle are defined as follows.

The Lindblad simulation algorithm leads to a purification of $\rho(T)$, i.e., $\rho(T)= \text{tr}(\ketbra{\rho_T}).$  It is clear that the regularization term in \eqref{qoc-opt} is easy to compute. With the purification, we can express the first term as,
\begin{equation}
  \label{eq:j1}
  \widetilde{J}_1(\bm{u}) := \bra{\rho_T}\mathscr{O}\otimes I \ket{\rho_T}.
  \end{equation}
Suppose $U_{\mathscr{O}}$ denotes the block encoding of $\mathscr{O}$, i.e. $\left\langle 0\left|\left\langle\psi_N\left|U_{\mathscr{O}}\right| 0\right\rangle\right| \psi_N\right\rangle=\left\langle\psi_N|\mathscr{O}| \psi_N\right\rangle$. Let c-$U_{\mathscr{O}}$ be the controlled $U_{\mathscr{O}}$. Applying Hadamard test circuit $(H \otimes I)\left(\mathrm{c-}U_{\mathscr{O}}\right)(H \otimes I)$ acting on  $\ket{\rho_T}$ produces
\begin{align}
\sqrt{f(\boldsymbol{u})}\ket{1}\ket{\phi_1(u)}+\sqrt{1-f(\boldsymbol{u})}\ket{0}\ket{\phi_0(\boldsymbol{u})}
\end{align}
where $f(\boldsymbol{u}):=-\frac{1}{2}\left\langle\rho_T|\mathscr{O}| \rho_T\right\rangle+\frac{1}{2}= -\frac{1}{2}\widetilde{J}_1(\boldsymbol{u})+\frac{1}{2}$ . By Lemma 48 of \cite{Gily_n_2019}, we can efficiently construct a block encoding of $\mathscr{O}$ with sparse access to $\mathscr{O}$. The $1/2$ factor does not matter because the gradient will only be multiplied by a constant factor.

\begin{definition}[Probability oracle]
Consider a function $f: \mathbb{R}^d \rightarrow[0,1]$. The probability oracle for $f$, denoted by $U_f$, is a unitary defined as
$$
\begin{aligned}
U_f|\boldsymbol{x}\rangle|\mathbf{0}\rangle=|\boldsymbol{x}\rangle\left(\sqrt{f(\boldsymbol{x})}|1\rangle\left|\phi_1(\boldsymbol{x})\right\rangle+\sqrt{1-f(\boldsymbol{x})}|0\rangle\left|\phi_0(\boldsymbol{x})\right\rangle\right),
\end{aligned}
$$
where $\left|\phi_1(\boldsymbol{x})\right\rangle$ and $\left|\phi_0(\boldsymbol{x})\right\rangle$ are arbitrary states.
\end{definition}

\begin{definition}[Phase oracle]
Consider a function $f\colon \mathbb{R}^d \rightarrow \mathbb{R}$. The phase oracle for $f$, denoted by $\mathcal{O}_f$, is a unitary defined as
$$
\mathcal{O}_f|\boldsymbol{x}\rangle|\mathbf{0}\rangle=e^{i f(\boldsymbol{x})}|\boldsymbol{x}\rangle|\mathbf{0}\rangle
$$
\end{definition}

\begin{theorem}
[Constructing phase oracle with probability oracle, Theorem 14 of \cite{gilyen2019optimizing}]
Consider a function $f: \mathbb{R}^d \rightarrow[0,1]$. Let $U_f$ be the probability oracle for $f$. Then, for any $\epsilon \in(0,1 / 3)$, we can implement an $\epsilon$-approximate of the phase oracle $\mathcal{O}_f$ for $f$, denoted by $\widetilde{\mathcal{O}}_f$, such that $\| \widetilde{\mathcal{O}}_f|\psi\rangle|\boldsymbol{x}\rangle-\mathcal{O}_f|\psi\rangle|\boldsymbol{x}\rangle \| \leq \epsilon$, for all state $|\psi\rangle$. This implementation uses $O(\log (1 / \epsilon))$ invocations to $U_f$ and $U_f^{\dagger}$, and $O(\log \log (1 / \epsilon))$ additional qubits.
\end{theorem}

  In order to interface the Lindblad simulation algorithm with a classical optimization method, one needs to estimate the gradient of the objective function. Similar to the approach in \cite{li2023efficient},
we first represent the control variable as a piecewise linear function in time $u_\beta(t)  \approx \sum_{j=1}^N u_j B_j(t) $ with $B_j(t)$ being the standard shape function and $u_j$ being a nodal function. The total number of steps $N$ is proportional to the time duration $T.$
We will use the improved Jordan's algorithm \cite{Jordan_2005} using high order finite difference formulas \cite{gilyen2019optimizing}.  Basically, the gradient estimation in \cite{gilyen2019optimizing} produces an estimate $\bm g(\bm u),$ such that, $\norm{\nabla J_1 (\bm u)  - \bm g(\bm u) } < \epsilon$ with complexity $O(d/\epsilon)$, which is clearly better than a direct sampling approach. However, to achieve this complexity, the objective function needs to satisfy a derivative bound.
Toward this end, we first establish an a priori bound for the derivative.

    \begin{lemma}
  \label{lemma:partial-alpha}
  Let $\bm{\alpha}  = (\alpha_1, \ldots, \alpha_k) \in [N+1]^k$ be an index sequence\footnote{For a precise definition of an index sequence, see Definition 4 of~\cite{gilyen2019optimizing}.}. The derivatives of the control function $\widetilde{J}_1$ with respect to the control variables satisfy:
\begin{equation}
      \begin{aligned}
            \norm{ \frac{\partial^{\bm{\alpha}} \widetilde{J}_1}{\partial u_{\alpha_1} u_{\alpha_2} \cdots u_{\alpha_k} }   }
             \leq (k+1)! \left(\delta t \|\mu\| \right)^k.
      \end{aligned}
\end{equation}
\label{lemma: J1 drv}
\end{lemma}
This smoothness provides a basis for estimating the complexity of Jordan's algorithm \cite{gilyen2019optimizing},
\begin{lemma}[Rephrased from Theorem 23 of~\cite{gilyen2019optimizing}]
  \label{lemma:jordan}
  Suppose the access to $f\colon[-1, 1]^N\rightarrow\mathbb{R}$ is given via a phase oracle $O_f$. If $f$ is $(2m+1)$-times differentiable and for all $\bm{x}\in[-1, 1]^N$,
  \begin{align}
    |\partial_{\bm{r}}^{2m+1} f(\bm{x})| \leq B \quad \text{ for $\bm{r} = \bm{x}/\norm{\bm{x}}$},
  \end{align}
  then there exists a quantum algorithm that outputs an approximate gradient $\bm{g}$ such that $\norm{\bm{g}-\nabla f(\bm{0})}_{\infty} \leq \epsilon$ with probability at least $1-\rho$ using
  \begin{align}
    \widetilde{O}\left( \max \left\{\frac{N^{1/2}B^{1/(2m)}N^{1/(4m)}\log(N/\rho)}{\epsilon^{1+1/(2m)}}, \frac{m}{\epsilon} \right\} \right)
  \end{align}
  queries to $O_f$, and $\widetilde{O}(N)$ additional 1- and 2-qubit gates.

  In particular, when $f(\bm x)$ is a polynomial of degree no greater than $2m$, the query complexity to $O_f$ becomes, 
  \begin{align}\label{J-polynomial}
      \widetilde{O} \left( \frac{m}{\epsilon} \right).
  \end{align}
\end{lemma}

After adapting this algorithm to the objective function in \cref{eq:j1}, we find that,
\begin{lemma}
  \label{lemma:grad-est-J}
  Let $\widetilde{J}_1$ be defined as in \cref{eq:j1}. Suppose we are given access to the phase oracle $\mathcal{O}_{\widetilde{J}_1}$ for $\widetilde{J}_1$. Then, there exists a quantum algorithm that outputs an approximate gradient $\bm{g}$ such that
    $\norm{\bm{g}-\nabla \widetilde{J}_1} \leq \epsilon_g$
  with probability at least $1-\gamma$ using
    $\widetilde{\CO}\left(n_{\mathrm{c}}T\log(N/\gamma)/\epsilon_g\right)$
  queries to $\mathcal{O}_{\widetilde{J}_1}$, and $\widetilde{O}(N)$ additional 1- and 2-qubit gates.
\end{lemma}
\begin{proof}
    Although the derivative bound in \cref{lemma: J1 drv} does not fulfill the condition in~\cite{gilyen2019optimizing}, we can apply Theorem 23  in~\cite{gilyen2019optimizing}. By choosing the optimal value $m$, we arrive at the complexity bound.
\end{proof}

With the gradient estimated, we can now move to the optimization algorithm. The PAGD algorithm in \cite{jin2017accelerated} assumes the gradient- and Hessian-Lipschitz condition, which we will prove here for the control problem. In particular, the  
smoothness constant $\ell$ and the Hessian-Lipschitz constant $\varrho$ can be approximated by the same technique as \cref{lemma: J1 drv}.
\begin{lemma}\label{lem:upper-bound-gradient}
    Let $\bm{\alpha}  = (\alpha_1, \ldots, \alpha_k) \in [N+1]^k$ be an index sequence, then $\tilde{J_1}$ is $l$-smooth and $\rho$-Hessian Lipschitz continuous, i.e.
    \begin{equation}
    \|\grad \tilde{J}_1(\bm u)-\grad \tilde{J}_1(\bm v)\|\leq l\|\bm u-\bm v\|, \text{and} \quad   \|\grad^2 \tilde{J}_1(\bm u)-\grad^2 \tilde{J}_1(\bm v)\|\leq \varrho\|\bm u-\bm v\|.
\end{equation}
   The smoothness parameters are given by,
    \begin{align}
        l = 3!(N+1)\delta t^2\|\mu\|^2\|\mathscr{O}\|, \quad \varrho = 4!(N+1)\delta t^3 \|\mu\|^3\|\mathscr{O}\|.
    \end{align}
\end{lemma}
The proof of this lemma is deferred to \Cref{sec:upper-bound-gradient}.

\section{Simulating open quantum systems with time-dependent Lindbladian}\label{sec:simulation}

Ref.~\cite[Section 6]{li2022simulating} sketched a method for simulating open quantum systems with time-dependent Lindbladian. In this section, we present the details of this simulation algorithm.

Motivated by the time scaling idea in \cite{BCSWW20},
we define a change-of-variable function as
\begin{align}\label{eq: t-scale}
  \mathrm{var}(t) \coloneqq \int^t_0 \dd s\, \norm{\mathcal{L}(s)}_{\mathrm{be}}.
\end{align}
By simulating the Lindblad dynamics on the new time scale, the overall complexity exhibits a better dependence on the norm of the Lindbladians in time. To this end, we need the following oracle to perform the inverse change-of-variable:
\begin{align}
  \mathcal{O}_{\mathrm{var}}\ket{t}\ket{z} = \ket{t}\ket{z\oplus \mathrm{var}^{-1}(t)}.
\end{align}
In addition, we need the following oracle to obtain the normalizing constant $\alpha_0(t)$ for $H(t)$ and $\alpha_j(t)$ for $L_j(t)$: for all $j = [m]$,
\begin{align}
  \mathcal{O}_{H, \mathrm{norm}}\ket{t}\ket{z} = \ket{t}\ket{z\oplus \alpha_0(t)}, \quad \text{and} \quad
  \mathcal{O}_{L_j, \mathrm{norm}}\ket{t}\ket{z} = \ket{t}\ket{z\oplus \alpha_j(t)}.
\end{align}

As in~\cite{li2022simulating}, we define the \emph{block-encoding norm} for a Lindbladian $\mathcal{L}$, denoted by $\norm{\mathcal{L}}_{\mathrm{be}}$ for normalization purposed:
\begin{align}
  \label{eq:be-norm}
    \norm{\mathcal{L}}_{\mathrm{be}} \coloneqq \alpha_0 + \frac{1}{2}\sum_{j=1}^m\alpha_j^2.
\end{align}

The goal of this section is to prove the following theorem.
\begin{theorem}
  \label{thm:ldb-td}
  Suppose we are given an $(\alpha_0(t), a, \epsilon')$-block-encoding $U_{H(t)}$ of $H(t)$, and an $(\alpha_j(t), a, \epsilon')$-block-encoding $U_{L_j(t)}$ for each $L_j(t)$ for all $0 \leq t \leq T$. Let $\norm{\mathcal{L}}_{\mathrm{be},1}$ be defined as
  \begin{align}
    \norm{\mathcal{L}}_{\mathrm{be},1} \coloneqq \int_0^T  \dd \tau\,\norm{\mathcal{L}(\tau)}_{\mathrm{be}},
  \end{align}
  Suppose further that $\epsilon' \leq \epsilon/(2t(m+1))$.
  Then, there exists a quantum algorithm that outputs a purification of $\tilde{\rho}_T$ of $\tilde{\rho}(T)$ where $\norm{\tilde{\rho}(T) - \mathcal{T} e^{\int_0^T\dd\tau\,\mathcal{L}(\tau)}(\rho_0)}_1 \leq \epsilon$ using
  \begin{align}
    O\left(\norm{\mathcal{L}}_{\mathrm{be},1}\,\left(\frac{\log(\norm{\mathcal{L}}_{\mathrm{be},1}/\epsilon)}{\log\log(\norm{\mathcal{L}}_{\mathrm{be},1}/\epsilon)}\right)^2\right)
  \end{align}
  queries to $U_{H(t)}$, $U_{L_j(t)}$, $\mathcal{O}_{\mathrm{var}}$, $\mathcal{O}_{H, \mathrm{norm}}$, and $\mathcal{O}_{L_j, \mathrm{norm}}$, and
    $\widetilde{O}\left((m+m)\norm{\mathcal{L}}_{\mathrm{be},1}\right)$
  additional 1- and 2-qubit gates, where $n$ is the number of qubits the Lindbladian is acting on.
\end{theorem}

\subsection{High-level overview of the simulation algorithm}
Here we briefly outline the techniques that led to the stated complexity.
Let the Hamiltonian $H(t) = H_0 +\sum_{\beta}u_{\beta}(t)\mu_{\beta}$, we rewrite equation \eqref{eq:Lindblad} as follows
\begin{align}
\label{eq:Lindblad2}
    \frac{d}{dt}\rho &= \mathcal{L}(t)(\rho)\coloneqq -i[H(t),\rho]+\sum_{j=1}^m(L_j(t)\rho L_j^{\dag}(t)-\frac{1}{2}\{L_j(t)^{\dag}L_j(t),\rho\})    \\
\label{eq:Lindblad3}
    & = \mathcal{L}_D(t)(\rho)+\mathcal{L}_J(t)(\rho).
\end{align}
Here we have decomposed the Lindbladian into a drift term $\mathcal{L}_\mathrm{D}(t)$ and a jump term $\mathcal{L}_\mathrm{J}(t)$:
\begin{align}
  \mathcal{L}_\mathrm{D}(t)(\rho) &= -i[H(t),\rho]-\frac{1}{2}\sum_{j=1}^m\{L_j(t)^{\dag}L_j(t),\rho\} =J(t)\rho+\rho J(t)^{\dag},  \\
  \mathcal{L}_\mathrm{J}(t)(\rho) &= \sum_{j=1}^m L_j(t)\rho L_j(t)^{\dag},
\end{align}
where $J(t) \coloneqq -iH(t)-\frac{1}{2}\sum_{j=1}^mL_j(t)^{\dag}L_j(t).$

With the known initial value $\rho(0) = \rho_0$, the solution of \cref{eq:Lindblad3} can be written as the linear combination of the following equations.
\begin{align}
\label{eq: superposition}
    \begin{cases}
      \partial_t\rho &= \mathcal{L}_\mathrm{D}(t)(\rho) \\
        \rho(0) &= \rho_0
      \end{cases},\quad \text{ and } \quad
    \begin{cases}
      \partial_t\rho &= \mathcal{L}_\mathrm{D}(t)(\rho) +\mathcal{L}_\mathrm{J}(t)(\rho)\\
        \rho(0)& =0
    \end{cases}.
\end{align}
Specifically, for the first part of \cref{eq: superposition}, the density operator follows $\rho(t) = V(0,t)\rho_0V(0,t)^{\dag} = \mathcal{K}[V(0,t)](\rho_0)$, where $V(s,t) = \mathcal{T}e^{\int_s^tJ(\tau)\mathrm{d}\tau}$ is the time-ordered exponential of $J$. A brief introduction of time-ordered exponential can be found in \cref{app: timeordered}. For the second part of \cref{eq: superposition}, the density operator follows $\rho(t) = \int_0^tg(t,s)\mathrm{d}s$, where the function $g(t,s)$ satisfying
\begin{align}
  \partial_tg(t,s) &= \mathcal{L}_\mathrm{D}(t)(g(t,s)), \text{ and }
  \lim_{t\rightarrow s}g(t,s) = \mathcal{L}_\mathrm{J}(s)(\rho(s)).
\end{align}
By using time-ordered evolution operator and Duhamel's principle, the solution of \cref{eq:Lindblad2} can be expressed as
\begin{align}
\label{eq:Duhamel}
  \rho(t) = \mathcal{K}[V(0,t)](\rho_0)+\int_0^t\mathcal{K}[V(s,t)](\mathcal{L}_\mathrm{J}(s)(\rho(s)))\,\dd s.
\end{align}
The time-ordered exponential $V(0,t)$ can be approximated by the truncated Dyson series (see \cref{app: Dyson} for details),
\begin{align}
  \label{eq:v0t}
  V(0,t) =\mathcal{T}e^{\int_0^tJ(\tau)\mathrm{d}\tau} \approx \sum_{k=0}^K \frac{1}{k!} \mathcal{T} \int_0^t\, \dd \bm{\tau} J(\tau_k)\cdots J(\tau_1),
\end{align}
where $\mathcal{T}\int_0^t\,\dd \bm{\tau}(\cdot)$ denote an integration over a $k$-tuple of time variables $(\tau_1, \ldots, \tau_k)$ while keeping the time ordering: $\tau_1 \leq \tau_2 \leq \cdots \leq \tau_k$.
Thus,
\begin{align}
  \label{eq:vst-1}
    V(s,t) &= \mathcal{T}e^{\int_s^tJ(\tau)\mathrm{d}\tau} = \mathcal{T}e^{\int_0^{t-s}J(s+\tau)\mathrm{d}\tau}
    \\
    \label{eq:vst-2}
           &\approx \sum_{k=0}^K \frac{1}{k!} \int_0^{t-s} \,\dd \bm{\tau}\mathcal{T} [J(\tau_k)\cdots J(\tau_1)].
\end{align}

As in~\cite{KSB19}, we use the rectangle rule to approximate the integral in the truncated Dyson series. Note that more efficient quadratures could be potentially used as we use later approximation the integral in \cref{eq:Duhamel}, for instance, the scaled Gaussian quadrature; however such methods require upper bounds on higher-order derivatives of $J(t)$, which are not readily available.

By applying Duhamel's principle (see \cref{eq:Duhamel}) several times, we obtain the following approximation with notations  introduced in~\cite{cao2021structure}.
\begin{align}
  \label{eq:gk}
        \mathcal{G}_K(t) \coloneqq \mathcal{K}[V(0,t)]+\sum_{k=1}^K\int_{0\leq s_1\leq \cdots \leq s_k\leq t}\mathcal{F}_k(s_k,\ldots,s_1)\,\dd s_1\cdots \dd s_k,
\end{align}
where
\begin{align}
  \label{eq:fk}
  \mathcal{F}_k(s_k,\ldots,s_1) \coloneqq \mathcal{K}[V(s_k,t)]\mathcal{L}_\mathrm{J}(s_k)
  \cdots\mathcal{K}[V(s_1,s_2)]\mathcal{L}_\mathrm{J}(s_1)\mathcal{K}[V(0,s_1)].
\end{align}
This yields an approximation of the evolution superoperator $\rho(t) \approx \mathcal{G}_K(t)(\rho(0)).$ 
The key observation is that $ \mathcal{F}_k$ is a composition of CPTP maps. The second term of $\mathcal{G}_K(t)$ can be approximated by using truncated Dyson series.

\subsection{Detailed constructions}
In this subsection, we present the construction of the time-dependent Lindbladian simulation algorithm. For the sake of conciseness, we omit the convoluted details in treating the time-ordering of the truncated Dyson series. All these details are postponed to \cref{sec:simulaiton-details}.

\paragraph*{The scaled evolution time}
Recall that we introduced the rescaled time in \cref{eq: t-scale}, and let define $\hat{t}$ as
\begin{align}
  \label{eq:scaled-t}
  \hat{t}= \mathrm{var}(t) = \int^t_0 \dd s\, \norm{\mathcal{L}(s)}_{\mathrm{be}}.
\end{align}

Correspondingly, we follow the rescaled Lindblad equation, by defining $\hat{\rho}(\hat{t})= \rho((\mathrm{var}^{-1}(\hat{t})),$
which, from \cref{eq:Lindblad}, satisfies the equation
\begin{align}
  \label{eq:lindblad-scaled}
  \frac{\dd}{\dd\hat{t}} \hat{\rho}(\hat{t}) =  \widehat{\mathcal{L}}(\hat{t})  \hat{\rho}(\hat{t}),
\end{align}
where the rescaled Lindbladian is as,
\begin{align}
  \label{eq:scaled-l}
  \widehat{\mathcal{L}}(\hat{t}) = \frac{\mathcal{L}(\mathrm{var}^{-1}(\hat{t}))}{\norm{\mathcal{L}(\mathrm{var}^{-1}(\hat{t}))}_{\mathrm{be}}}.
\end{align}
This rescaling can be achieved by defining
\begin{align}
  \widehat{H}(\hat{t}) \coloneqq \frac{H(\mathrm{var}^{-1}(\hat{t}))}{\norm{\mathcal{L}(\mathrm{var}^{-1}(\hat{t}))}_{\mathrm{be}}}, \text{ and } 
  \widehat{L}(\hat{t}) \coloneqq \frac{L(\mathrm{var}^{-1}(\hat{t}))}{\sqrt{\norm{\mathcal{L}(\mathrm{var}^{-1}(\hat{t}))}_{\mathrm{be}}}}.
\end{align}
The \emph{scaled effective Hamiltonian} (not Hermitian), 
denoted by $\widehat{J}(\hat{t})$, is therefore defined as
\begin{align}
  \widehat{J}(\hat{t}) \coloneqq \frac{J(\mathrm{var}^{-1}(\hat{t}))}{\norm{\mathcal{L}(\mathrm{var}^{-1}(\hat{t}))}_{\mathrm{be}}}.
\end{align}
As a result, simulating $\widehat{\mathcal{L}}$ for time $\hat{t} = \mathrm{var}(t)$ is equivalent to simulating $\mathcal{L}$ for time $t$. Moreover, the block-encoding-norm of $\widehat{\mathcal{L}}$ is at most 1 because of \cref{eq:scaled-l}.

To simplify the notation, in the remainder of this section we assume the Lindbladian is already scaled so that we can drop the $\widehat{\cdot}$ notation for the scaled operators and evolution time.

\paragraph*{LCU construction}
Let $U_{J}(t)$ be an $(\alpha, a, \epsilon)$-block-encoding of $J(t)$. Given the oracles as in \cref{eq:oh,eq:ohval}, the unitary $\sum_t\ketbra{t}{t}\otimes U_{J}(t)$ for discretized times $t$ can be implemented. Using \cref{lemma:sum-to-be}, a block-encoding of $V(s, t)$ can also be implemented. More specifically, we use the rectangle rule as in~\cite{KSB19} to approximate the integrals in \cref{eq:vst-2}:
\begin{align}
  \label{eq:dyson-riemann}
  \widetilde{V}(s, t) = \sum_{k=0}^{K'}\frac{(t-s)^k}{M^kk!}\sum_{j_1,\ldots,j_k=0}^{M-1}\mathcal{T}J(t_{j_k})\cdots J(t_{j_1}).
\end{align}
Here the time-ordered term is defined as follows, for each tuple $t_k, t_{k-1}, \ldots, t_1$,
\[
\mathcal{T}J(t_k)\cdots J(t_1)
=
    J(t_{j_k})\cdots J(t_{j_1}), 
\]
where $t_{j_k}, \ldots, t_{j_1}$ is the permutation of $t_k, t_{k-1}, \ldots, t_1$ that is in ascending order. 

The error of the above approximation is bounded by
\begin{align}
  \norm{V(s, t) - \widetilde{V}(s, t)} \leq O\left(\frac{(t-s)^{K'+1}}{(K'+1)!} + \frac{(t-s)^2 \dot{J}_{\max}}{M}\right),
\end{align}
where $\dot{J}_{\max} \coloneqq \max_{\tau \in [0, t]} \norm{ \frac{\dd J(\tau)}{\dd \tau}}$.

Now, we need to approximate the integrals in \cref{eq:gk}. In~\cite{li2022simulating}, Gaussian quadratures were used to approximate similar integrals in the time-independent case, which yields a simpler LCU construction. Unfortunately, using  such efficient quadrature rules in the time-dependent case requires  bounding the norm of high-order derivatives of $V(s, t)$, which is not directly given. Instead, we use the simple Riemann sums for treating the integrals, where the LCU constructions follow closely from the ones in~\cite{KSB19}.

More specifically, we uniformly divide the evolution time $t$ into $q$ intervals, and let $t_j = tj/q$ for $j \in \{0, \ldots, M-1\}$. Assuming $V(s, t)$ is implemented perfectly, we consider the following superoperator,
\begin{align}
  \frac{t^k}{k!q^k}\sum_{j_1,\ldots,j_k=0}^q \mathcal{T} \mathcal{F}_k(t_{j_k}, \ldots, t_{j_1}),
\end{align}
which approximates the integrals in \cref{eq:gk}. To bound the quality of this approximation, we need to bound the derivative of $\mathcal{F}_k$. We begin by bounding $\norm{V(0, t)}$, which can be deduced from the stability of the differential equation $\frac{\dd}{\dd t}\bm y = J(t) \bm y $, which can be studied by examining the eigenvalues of the Hermitian part of $J(t)$  \cite[Lemma 1]{brauer1966perturbations}. Since the Hermitian part of $J(t)$ is semi-negative definite, one has $\norm{\bm y(t)} \leq  \norm{\bm y(0)}$, which implies that
\begin{align}
  \label{eq:norm-v0t}
  \norm{V(0, t)} \leq 1.
\end{align}

Since $\frac{\dd}{\dd t}V(0, t) = J(t)V(0, t)$, the derivative of $V(0, t)$ can be bounded by
\begin{align}
    \label{eq:dv0t}
    \norm{\frac{\dd}{\dd t}V(0, t)} \leq \dot{J}_{\max}.
\end{align}
We further consider $\frac{\dd}{\dd t}\mathcal{K}[V(0, t)]$. For any operator $A$ with $\norm{A}_1 = 1$, we have
\begin{align}
  \frac{\dd}{\dd t}\mathcal{K}[V(0, t)](A) = \left(\frac{\dd}{\dd t}V(0, t)\right) A V(0, t)^{\dag} + V(0, t)A\frac{\dd}{\dd t}V(0, t)^{\dag}.
\end{align}
We then have
\begin{align}
  \norm{\frac{\dd}{\dd t}\mathcal{K}[V(0, t)](A)}_1 \leq  2\dot{J}_{\max},
\end{align}
which follows from \cref{eq:dv0t,eq:norm-v0t} and the fact that $\norm{BAC}_1 \leq \norm{B}\norm{A}_1\norm{C}$ for matrices $A, B, C$. This bound easily extends to the diamond norm by tensoring the Kraus operator with an identity operator to extend it to a larger space. Hence, we have
\begin{align}
  \label{eq:diamondnorm-kv0t}
  \norm{\frac{\dd}{\dd t}\mathcal{K}[V(0, t)](A)}_{\diamond} \leq  2\dot{J}_{\max}.
\end{align}
Let $\dot{L}_{j,\max}$ be defined as $\dot{L}_{j,\max} \coloneqq \max_{\tau \in [0, t]} \norm{\frac{\dd}{\dd \tau} L_{j}(\tau)}.$
Then, using similar arguments, we can bound the derivative of $\mathcal{L}_{\mathrm{J}}(t)$ as
\begin{align}
  \label{eq:diamondnorm-ljt}
  \norm{\frac{\dd}{\dd t}\mathcal{L}_{\mathrm{J}}(t)}_{\diamond} \leq 2\sum_{j=1}^m\dot{L}_{j, \max},
\end{align}
where we have assumed that the Lindbladian is scaled as in \cref{eq:scaled-l}, i.e., $\norm{L_j} \leq 1$.
For the derivative of $\mathcal{F}_k$, we have
\begin{equation}
    \begin{aligned}
   &\quad \frac{\dd}{\dd t_j}\mathcal{F}_k  \\ 
  &= \mathcal{K}[V(t_k,t)]\mathcal{L}_\mathrm{J}(t_k)\cdots \frac{\dd}{\dd t_j}\left(\mathcal{K}[V(t_{j},t_{j+1})]\mathcal{L}_\mathrm{J}(t_{j})\mathcal{K}[V(t_{j-1},t_{j})]\right)\mathcal{L}_\mathrm{J}(t_{j-1})\cdots\mathcal{K}[V(0,t_1)] \\
                                   & = \mathcal{K}[V(t_k,t)]\mathcal{L}_\mathrm{J}(t_k)\cdots \frac{\dd}{\dd t_j}(\mathcal{K}[V(t_{j},t_{j+1})])\mathcal{L}_\mathrm{J}(t_{j})\mathcal{K}[V(t_{j-1},t_{j})]\mathcal{L}_\mathrm{J}(t_{j-1})\cdots\mathcal{K}[V(0,t_1)] \\
                                   & + \mathcal{K}[V(t_k,t)]\mathcal{L}_\mathrm{J}(t_k)\cdots \mathcal{K}[V(t_{j},t_{j+1})]\frac{\dd}{\dd t_j}(\mathcal{L}_\mathrm{J}(t_{j}))\mathcal{K}[V(t_{j-1},t_{j})]\mathcal{L}_\mathrm{J}(t_{j-1})\cdots\mathcal{K}[V(0,t_1)] \\
                                   & + \mathcal{K}[V(t_k,t)]\mathcal{L}_\mathrm{J}(t_k)\cdots \mathcal{K}[V(t_{j},t_{j+1})]\mathcal{L}_\mathrm{J}(t_{j})\frac{\dd}{\dd t_j}(\mathcal{K}[V(t_{j-1},t_{j})])\mathcal{L}_\mathrm{J}(t_{j-1})\cdots\mathcal{K}[V(0,t_1)].
\end{aligned}
\end{equation}

Again, assume the Lindbladian is scaled as in \cref{eq:scaled-l}, the above expression of $\frac{\dd}{\dd t_j}\mathcal{F}_k$ together with \cref{eq:diamondnorm-kv0t,eq:diamondnorm-ljt} implies that 
\begin{align}
  \norm{\frac{\dd}{\dd t_j}\mathcal{F}_k}_{\diamond} \leq 4\dot{J}_{\max}+2\sum_{j=1}^m\dot{L}_{j, \max}.
\end{align}
This implies that the error for using the Riemann sums can be bounded by
\begin{align}
  &\norm{\mathcal{G}_K - \mathcal{K}[V(0, t)] - \frac{t^k}{k!q^k}\sum_{k=1}^K\sum_{j_1,\ldots,j_k=0}^q \mathcal{T} \mathcal{F}_k(t_{j_k}, \ldots, t_{j_1})}_{\diamond} \nonumber\\
  &=
  \norm{\sum_{k=1}^K\int_{0\leq s_1\leq \cdots \leq s_k\leq t}\mathcal{F}_k(s_k,\ldots,s_1)\,\dd s_1\cdots \dd s_k - \frac{t^k}{k!q^k}\sum_{k=1}^K\sum_{j_1,\ldots,j_k=0}^q \mathcal{T} \mathcal{F}_k(t_{j_k}, \ldots, t_{j_1})}_{\diamond} \notag \\
  &\leq \sum_{k=1}^K\frac{t^2}{q}\cdot \left(4\dot{J}_{\max}+2\sum_{j=1}^m\dot{L}_{j, \max}\right) \\
  \label{eq:error-intermediate1}
  &= \frac{Kt^2}{q}\cdot \left(4\dot{J}_{\max}+2\sum_{j=1}^m\dot{L}_{j, \max}\right). 
\end{align}
In addition, it is easy to see that the error caused by using Duhamel's principle is
\begin{align}
  \label{eq:error-intermediate2}
  \norm{e^{\mathcal{T}\int_0^t\dd\tau\,\mathcal{L}(\tau)} - \mathcal{G}_K}_{\diamond} \leq \frac{(2t)^{K+1}}{(K+1)!}.
\end{align}
It follows from \cref{eq:error-intermediate1,eq:error-intermediate2} that
\begin{equation}
    \begin{aligned}
  \label{eq:error-intermediate3}
 & \norm{e^{\mathcal{T}\int_0^t\dd\tau\,\mathcal{L}(\tau)} - \mathcal{K}[V(0, t)] - \frac{t^k}{k!q^k}\sum_{j_1,\ldots,j_k=0}^q \mathcal{T} \mathcal{F}_k(t_{j_k}, \ldots, t_{j_1})}_{\diamond} \\ 
      \leq  & \frac{(2t)^{K+1}}{(K+1)!} + \frac{Kt^2}{q}\left(4\dot{J}_{\max}+2\sum_{j=1}^m\dot{L}_{j, \max}\right). 
\end{aligned}
\end{equation}

Finally, we have the following LCU form:
\begin{align}
  \label{eq:tildegk}
  \widetilde{\mathcal{G}}_K \coloneqq \mathcal{K}[\widetilde{V}(0, t)] + \sum_{k=1}^K\frac{t^k}{q^k}\sum_{j_1,\ldots,j_k=0}^q \widetilde{\mathcal{F}}_k(t_{j_k}, \ldots, t_{j_1}) ,
\end{align}
where $\widetilde{\mathcal{F}}_K$ is an approximation of \cref{eq:fk} by using $\widetilde{V}(s, t)$ instead of $V(s, t)$, i.e.,
\begin{align}
  \label{eq:tildefk}
\!\!  \widetilde{\mathcal{F}}_k(s_k,\ldots,s_1) \coloneqq \mathcal{K}[\widetilde{V}(s_k,t)]
  \mathcal{L}_J(s_k)\mathcal{K}[\widetilde{V}(s_{k-1},s_k)]
  \cdots\mathcal{K}[\widetilde{V}(s_1,s_2)]\mathcal{L}_J(s_1)\mathcal{K}[\widetilde{V}(0,s_1)].
\end{align}
We use the same compression scheme as in~\cite{KSB19} to deal with the time-ordering in \cref{eq:dyson-riemann,eq:tildegk}. Note that implementing the LCU requires additional $O(KK'm(\log M+\log q + n))$ 1- and 2-qubit gates.

\paragraph*{Complexity analysis}
We first analyze the normalizing constant for the LCU implementation. Recall that we are working with scaled operators, so the normalizing factors are at most 1. For the implement of $V(0, \hat{t})$, we can use, for example, the LCU construction involving quantum sort as in~\cite{KSB19} for implement \cref{eq:dyson-riemann}. If we further assume the implementation uses an infinite Dyson series, the normalizing constants of the block-encoding $\mathcal{K}[V(0, t)]$ is upper bounded by 
\begin{align}
  \label{eq:sos-kv0t}
  \sum_{k=0}^{\infty}\frac{t^{k}}{k!} = e^{t}.
\end{align}
As a result, the sum-of-squares of the normalizing constants of the Kraus operators of $\mathcal{F}_k(\hat{t}_k,\ldots, \hat{t}_1)$ can be bounded by
\begin{align}
  \sum_{j_1,\ldots,j_k=0}^me^{2(t-s_k)}e^{2(s_{k}-s_{k-1})}\ldots e^{2(s_{1}-0)} = e^{2t}.
\end{align}
Recall that the normalizing constant for $L_j$ is 1 since the Lindbladian is rescaled.
For the second term in \cref{eq:gk}, the sum-of-squares of the normalizing constants of the Kraus operators can be bounded by
\begin{align}
  \label{eq:sos-f}
  e^{2t}\frac{t^k}{k!q^k}q^k = e^{2t}\frac{t^k}{k!}.
\end{align}
By \cref{eq:sos-kv0t,eq:sos-f}, we have that the sum-of-squares of the normalizing constants of the Kraus operators of the LCU in \cref{eq:gk} can then be bounded by
\begin{align}
  e^{2t} + \sum_{k=1}^Ke^{2t}\frac{t^k}{k!}\leq e^{2t} + e^{3t}.
\end{align}

Therefore, it suffices to set
\begin{align}
  \label{eq:t-constant}
  t = \Omega(1)
\end{align}
to achieve constant success probability when using \cref{lemma:block-encoding-channel}. Then, we use the oblivious amplitude amplification for channels~\cite{cleve2016efficient} to boost the success probability to 1 with constant applications of \cref{lemma:block-encoding-channel}.
For the error bound in \cref{eq:error-intermediate3}, assume for now that the second error term is dominated by the first (by some choice of $q$ to be determined). It suffices to set $K = \frac{\log(1/\epsilon)}{\log\log(1/\epsilon)}$
to make the total error at most $\epsilon/2$, because of the choice of $t$ in \cref{eq:t-constant}. The choice of $q$ satisfies
\begin{align}
  q = \Theta\left(\frac{2K}{\epsilon}\left(4\dot{J}_{\max}+2\sum_{j=1}^m\dot{L}_{j, \max}\right)\right). 
\end{align}

Now, we deal with the error from truncated Dyson series and Riemann sum to implement $V(s, t)$. By \cref{eq:dyson-riemann}, we can choose $M$ large enough (determined later) so that the second error term is dominated by the first. Then, using \cite[Lemma 7]{li2022simulating}, we have
\begin{align}
  \norm{\mathcal{F}_k(s_k,\ldots,s_1) - \widetilde{\mathcal{F}}_k(s_k,\ldots,s_1)}_{\diamond} \leq \frac{8e^{t}}{(K'+1)!}2^{k+1}t^{K'+1}.
\end{align}
Further, using the analysis as in~\cite{li2022simulating}, we can bound the total approximation error (with appropriate choice of $M$ to be determined later) as
\begin{align}
  \label{eq:overall-error}
  \norm{\mathcal{T}e^{\int_0^t\mathcal{L}(\tau)\dd\bm{\tau}} - \widetilde{\mathcal{G}}_K}_{\diamond} \leq \frac{32e^{5t}t^{K'+2}}{(K'+1)!}.
\end{align}
With the choice of $t$ in \cref{eq:t-constant}, we can choose
\begin{align}
  K' = \frac{\log(1/\epsilon)}{\log\log(1/\epsilon)}
\end{align}
so that the total error is bounded by $\epsilon$. For the choice of $M$, we need to make sure the second error term in \cref{eq:dyson-riemann} is dominated by the first term. Hence we can choose
\begin{align}
  M = \Theta\left(\frac{\dot{J}_{\max}}{\epsilon}\right).
\end{align}

It remains to analyze the cost for the LCU implementation, which is the same as the analyses in~\cite{li2022simulating} and~\cite{KSB19}. Note that the dependence on $M$ is logarithmic if the compressed scheme is used in~\cite{KSB19} for implementing $\widetilde{V}(s, t)$. The total gate cost is now upper bounded by $O(KK'm(\log M + \log q + n))$. Further note that  the error $\epsilon'$ brought by the block-encoding can be eventually transferred to $\mathcal{L}$ causing an $(m+1)\epsilon'$ error on $\mathcal{L}$ in terms of the diamond norm, and the accumulated error for evolution time $t$ is then at most $t(m+1)\epsilon'$. As a result, choosing $\epsilon' \leq \epsilon/(2t(m+1))$ suffices to ensure the total error is at most $\epsilon$.

Recall that the above analysis is based on the scaled version of $\mathcal{L}$ defined in \cref{eq:scaled-l}, and the evolution time is scaled as in \cref{eq:scaled-t}. For arbitrary evolution time $\hat{t}$, we apply the above procedure $O(\hat{t})$ times with precision $\epsilon' = \epsilon/\hat{t}$. This gives the desired complexity in \cref{thm:ldb-td}. Lastly, it is important to note that the LCU circuit yields a purification of $\rho(t)$. This completes the proof of \cref{thm:ldb-td}.

Note that the above analysis easily extends to the simulation of the original Lindbladian without any scaling, where the complexity depends linearly on the product of evolution time and the maximum of the block-encoding norm of the Lindbladian. More specifically, we have the following corollary.
\begin{corollary}
  \label{cor:ldb-td-noscaling}
  Suppose we are given an $(\alpha_0(t), a, \epsilon')$-block-encoding $U_{H(t)}$ of $H(t)$, and an $(\alpha_j(t), a, \epsilon')$-block-encoding $U_{L_j(t)}$ for each $L_j(t)$ for all $t \geq 0$. Let $\norm{\mathcal{L}}_{\mathrm{be},\infty}$ be defined as
  \begin{align}
    \norm{\mathcal{L}}_{\mathrm{be},\infty} \coloneqq \max_{\tau \in [0, T]}  \norm{\mathcal{L}(\tau)}_{\mathrm{be}},
  \end{align}
  Suppose further that $\epsilon' \leq \epsilon/(2T(m+1))$.
  Then, there exists a quantum algorithm that outputs a purification of $\tilde{\rho}_T$ of $\tilde{\rho}(T)$ where $\norm{\tilde{\rho}(T) - e^{\mathcal{T}\int_0^T\dd\tau\,\mathcal{L}(\tau)}(\rho_0)}_1 \leq \epsilon$ using
  \begin{align}
    O\left(T\norm{\mathcal{L}}_{\mathrm{be},\infty}\,\left(\frac{\log(T\norm{\mathcal{L}}_{\mathrm{be},\infty}/\epsilon)}{\log\log(T\norm{\mathcal{L}}_{\mathrm{be},\infty}/\epsilon)}\right)^2\right)
  \end{align}
  queries to $U_{H(t)}$, $U_{L_j(t)}$, $\mathcal{O}_{\mathrm{var}}$, $\mathcal{O}_{H, \mathrm{norm}}$, and $\mathcal{O}_{L_j, \mathrm{norm}}$, and
  \begin{align}
    \widetilde{O}\left((m+n)T\norm{\mathcal{L}}_{\mathrm{be},\infty}\right)
  \end{align}
  additional 1- and 2-qubit gates, where $n$ is the number of qubits the Lindbladian is acting on.
\end{corollary}

\section{Simulations in the Interaction Picture}\label{sec:interaction}
Many control problems involve a system Hamiltonian that contains a time-independent Hamiltonian that dominates the spectral norm $H(t)$, and thus the overall computational complexity.   Motivated by the interaction picture approach for Hamiltonian simulations \cite{low2018hamiltonian}, we devise an approach to simulate the Lindblad dynamics. To formulate the problem, we assume that the Lindbladian admits the following decomposition:
\begin{align}
    \mathcal{L}(\cdot) = -i\left[H_1 + H_2(t), \cdot\right]+\sum_j L_j(\cdot) L_j^{\dagger}-\frac12 \left\{L_j^{\dagger} L_j, \cdot\right\} ,
    \label{eq:H1H2Lj_relation}
\end{align}
where $H_1$ is a time-independent free Hamiltonian, $H_2(t)$ is the coupling Hamiltonian which contains the control variables, and the dissipative terms still come from the interaction with the environment.

One such example is the control of an ion trap system~\cite{haffner2008quantum}, 
in which the model Hamiltonian consists of the following terms, 
\begin{align}
    H_1 =\ & \hbar \sum_{i=1}^N\left(\omega_{01}|1\rangle_i\left\langle 1\left|+\omega_{0 e}\right| e\right\rangle_i\langle e|\right)+\hbar \sum_k \omega_k a_k^{\dagger} a_k \\
H_2(t) =\ & \hbar \Omega_1 \cos \left(\boldsymbol{k}_1 \cdot \boldsymbol{r}_j-\omega_1 t-\varphi_1\right)\left(|0\rangle_j\langle e|+| e\rangle_j\langle 0|\right) \\
    & +\hbar \Omega_2 \cos \left(\boldsymbol{k}_2 \cdot \boldsymbol{r}_j-\omega_2 t-\varphi_2\right)\left(|1\rangle_j\langle e|+| e\rangle_j\langle 1|\right),
\end{align}
and $L_j$s includes $\lambda_{\text{heat}} a^{\dagger}_j$, $\lambda_{\text{damp}} a_j$ and $\lambda_{\text{dephase}}n_j$. The observation in \cite{haffner2008quantum} is that
\begin{align}
     \omega_{0e} \gg \left|\Omega_1\right|,\left|\Omega_2\right|  \gg \lambda_\text{{heat}}, \lambda_\text{{damp}}, \lambda_\text{{dephase}}.
\end{align}

Motivated by such applications, we assume that in \cref{eq:H1H2Lj_relation},
\begin{align}\label{eq:interaction-assum}
    \norm{H_1} \gg \norm{H_2(t)} \gg \norm{L_j}.
\end{align}

In the interaction approach, e.g., \cite{low2018hamiltonian}, one simulates the density operator in the interaction picture, where the large magnitude of $H_1$ is absorbed into the slow Hamiltonian $H_2(t)$ and the jump operators.  
In this section, we provide detailed quantum algorithms for simulating the Lindbladian  \cref{eq:H1H2Lj_relation} in the interaction picture.

\subsection{ Lindbladian simulation in interaction picture  }
In light of \cref{eq:H1H2Lj_relation}, we first write the Lindbladian into two parts
\begin{align}
\mathcal{L}(t) = \mathcal{L}_1 + \mathcal{L}_2(t)
\label{eq:int_liou}
\end{align}
where
 $\mathcal{L}_1$ contains a time-independent Hamiltonian term and $\mathcal{L}_2(t)$ can be a general Lindbladian term
 \begin{align}
     \mathcal{L}_1(\cdot ) &= -i\left[H_1, \cdot\right] \label{eq:L1} \\
     \mathcal{L}_2(t)(\cdot) &= -i\left[H_2(t), \cdot\right] + \sum_j L_j (\cdot) L_j^{\dagger}-\left\{L_j^{\dagger} L_j, \cdot\right\} \label{eq:L2}.
 \end{align}
Then the Lindblad master equation in \cref{eq:Lindblad} is equivalent to:
\begin{align}
    \frac{\dd}{\dd t} {V}_1\dag (t_0,t) \rho {V}_1(t_0,t) = {V}_1^\dag (t_0,t)  \mathcal{L}_2(t){V}_1(t_0,t) {V}^\dag_1(t_0,t) \rho {V}_1(t_0,t),
\end{align}
where ${V}_1(t_0,t) =  e^{-i H_1 (t-t_0)} $, and $t \ge t_0$. 

We can define $\rho_{\mathrm{I}} = {V}_1^\dag(t_0,t) \rho {V}_1(t_0,t)$ as the density operator in the interaction picture, and it satisfies the Lindblad equation, $ \frac{d}{d t} \rho_{\mathrm{I}}(t) = \mathcal{L}_{2,\mathrm{I}}(t) \rho_{\mathrm{I}}(t),$
where $\mathcal{L}_{2,\mathrm{I}}(t) \coloneqq{V}_1^\dag(t_0,t)  \mathcal{L}_2(t){V}_1(t_0,t).$
Effectively, this transforms $H_2$ and $L_j(t)$ in \cref{eq:L2} into an interaction picture as well.

By simulating the time evolution in the interaction picture and transforming it back to the original picture at last, we have 
\begin{align}
    \rho(t) = \left(e^{\mathcal{L}_1\left(t-t_0\right)}\right) \left(\mathcal{T}e^{\int_{t_0}^t \mathcal{L}_{2,\mathrm{I}}(s) \mathrm{d} s}\rho(t_0) \right) ,
\end{align}
where $\left(e^{\mathcal{L}_1\left(b-a\right)}\right) (\cdot) = V_1\left(a, b\right)(\cdot)V_1^{-1}\left(a, b\right) $. We can further decompose this evolution into $N$ Trotter steps (with $\tau = (t-t_0)/N$),
\begin{align}
    \rho (t) 
    &= \prod_{i=0}^{N-1} \left(e^{ \mathcal{L}_1 \tau}\mathcal{T} e^{\int_{t_0+i\tau}^{t_0+(i+1)\tau}  \mathcal{L}_{2,\mathrm{I}}(s) \mathrm{d} s}\right) \rho (t_0). \label{eq:int_iter}
\end{align}

At a high level, \cref{eq:int_iter} summarizes our simulation strategy in the interaction picture. The total time complexity is determined by the number of time steps $N$, and the time complexity in each step, which follows from  our Lindbladian simulation algorithm in \Cref{sec:simulation}. 

\begin{theorem}[Modified from Corollary~\ref{cor:ldb-td-noscaling}]\label{theo:lind_sim}
     Suppose we are given an $\left(\alpha_0, a, \epsilon^{\prime}\right)$-block-encoding $U_H$ of $H$, and an $\left(\alpha_j, a, \epsilon^{\prime}\right)$-block-encoding $U_{L_j}$ for each $L_j$. For all $\tau, \epsilon' \geq 0$ and $t\|\mathcal{L}(\tau)\|_{\mathrm{be},\infty}=\Theta(1)$, there exists a quantum algorithm for simulating $e^{\mathcal{L} \tau}$ using
    $$
    O\left( \frac{\log \left(1 / \epsilon'\right)}{\log \log \left(1 / \epsilon'\right)}\right)
    $$
    queries to $U_H$ and $U_{L_j}$ and
    $$
    O\left(m \left(\frac{\log \left(1 / \epsilon'\right)}{\log \log \left(1 / \epsilon'\right)}\right)^2\right)
    $$
    additional 1-and 2-qubit gates.
\end{theorem}
\begin{lemma}[Error accumulation]\label{lem:err_accum}
     
    Given that $A_j=W_j$ and $B_j=\mathcal{T}\left[e^{ \int_{t_{j-1}}^{t_j} \mathcal{L}_1(s)\mathrm{d} s}\right]$ are bounded $\left\|W_j\right\| \leq 1$, $\left\| B_j \right\| \leq 1$, and error in each segment is bounded by $\delta$
    \begin{align}
        \left\|A_j-B_j\right\| \le \delta.
    \end{align}
    Then the accumulated error is
    \begin{align}
        \left\|\prod_j^L W_j-\mathcal{T}\left[e^{ \int_0^t \mathcal{L}(s) \mathrm{d} s}\right]\right\| \leq L \delta .
    \end{align}
\end{lemma}
\begin{proof}
 The lemma holds by applying the triangle inequality
    \begin{align}
        \left\|\prod_{j=1}^L A_j-\prod_{j=1}^L B_j\right\| \leq \sum_{k=1}^L\left(\prod_{j=1}^{k-1}\left\|A_j\right\|\right)\left\|A_k-B_k\right\|\left(\prod_{j=k+1}^L\left\|B_j\right\|\right) .
    \end{align}
\end{proof}

These results imply the following result for Lindbladian simulation in the interaction picture:
\begin{theorem}[Query complexity of Lindbladian simulation in the interaction picture]
Let $\mathcal{L}(t)=\mathcal{L}_1(t)+\mathcal{L}_2(t)$, with $\mathcal{L}_1(t)$ and $\mathcal{L}_2(t)$ defined by \cref{eq:L1,eq:L2} respectively.  Assume the existence of a unitary oracle that implements the Hamiltonian and Lindbladian within the interaction picture, denoted $U_{H^I}$ and $U_{L_j^I}$ which implicitly depends on the time-step size $\tau \in \mathcal{O}\left(||\mathcal{L}_2||_{\text{be}}^{-1}\right)$ and number of quadrature points $q$, such that
\begin{align}
\left(\left\langle 0\right|_a \otimes \mathbf{1}_s\right) U_{H^I}\left(|0\rangle_a \otimes \mathbf{1}_s\right)&=\sum_{j_k = 1}^{q}|j_k\rangle\langle j_k| \otimes \frac{e^{i H_1 \tau \hat{x}_{(j_k)}} H_2 e^{-i H_1 \tau \hat{x}_{(j_k)}}}{\alpha_H} 
 \label{eq:H^I}\\
\left(\left\langle 0\right|_a \otimes \mathbf{1}_s\right) U_{L_j^I}\left(|0\rangle_a \otimes \mathbf{1}_s\right)&=\sum_{j_k = 1}^{q}|j_k\rangle\langle j_k| \otimes \frac{e^{i H_1 \tau \hat{x}_{(j_k)}} L_j e^{-i H_1 \tau \hat{x}_{(j_k)}}}{\alpha_{L_j}}, \label{eq:L_j^I}
\end{align}
For $t \ge ||\mathcal{L}_2(t)||_{\mathrm{be}} \tau$, the time-evolution operator $\mathcal{T} e^{\int_{0}^t \mathcal{L}_1(s) + \mathcal{L}_2(s) \mathrm{d} s}$ may be approximated to error $\epsilon$ with the following cost.

\begin{enumerate}
\item Simulations of $e^{-i H_1 \tau}: \mathcal{O}\left(t||\mathcal{L}_2(t)||_{\mathrm{be},\infty} \right)$,

\item Queries to $U_{H^I}$ and $U_{L_j^I}$: $\mathcal{O}\left(t||\mathcal{L}_2(t)||_{\text{be},\infty} \frac{\log \left(t||\mathcal{L}_2(t)||_{\text{be},\infty} / \epsilon\right)}{\log \log \left(t||\mathcal{L}_2(t)||_{\text{be},\infty} / \epsilon\right)}\right)$,

\item Primitive gates:  $\mathcal{O}\left(mt||\mathcal{L}_2(t)||_{\mathrm{be},\infty} (\frac{\log \left(t||\mathcal{L}_2(t)||_{\mathrm{be},\infty} / \epsilon\right)}{\log \log \left(t||\mathcal{L}_2(t)||_{\mathrm{be},\infty} / \epsilon\right)})^2\right)$.
\end{enumerate}
\end{theorem}

\begin{proof}
    Consider simulation strategy shown in \cref{eq:int_iter}, we uniformly divide the evolution time $[0,t]$ into $M=\lceil t\left\|\mathcal{L}_2(t)\right\|_{be,\infty} \rceil $, time step $\tau = t/M$. Then $\tau \left\|\mathcal{L}_2(t)\right\|_{b e, \infty} = \Theta (1)$, which satisfies the pre-condition of Theorem \ref{theo:lind_sim}. Therefore, using Theorem \ref{theo:lind_sim}, the time and gate complexity of each time interval is $O\left(\frac{\log \left(1 / \epsilon^{\prime}\right)}{\log \log \left(1 / \epsilon^{\prime}\right)}\right)$ and $O\left(m\left(\frac{\log \left(1 / \epsilon^{\prime}\right)}{\log \log \left(1 / \epsilon^{\prime}\right)}\right)^2\right)$, respectively. Furthermore, by the error accmulation in Lemma \ref{lem:err_accum}, we choose $\epsilon^{\prime} = \epsilon /t\left(\|\mathcal{L}\|_{\text{be}}\right)$ in order to bound the overall error by $\epsilon$.

    In addition, since we need to invoke $e^{-i H_1 \tau}$ once every step, the invoking number equals to $M$ and is hence bounded as claimed.
\end{proof}

\subsection{Comparison of the simulation complexity with and without interaction picture}

In this subsection, we compare the complexity with simulations of Lindblad dynamics with and without the interaction picture.
For the Lindbladian decomposition shown in Eq.~\eqref{eq:int_liou}, suppose we have access to the oracles $U_{H_1}$,$U_{H_{2}(t)}$, and $U_{L_j}$. According to \Cref{thm:ldb-td}, a direct simulation involves a time complexity
\begin{align} \label{eq:direct_simu}
    C_{\text{direct}} = O\left(t(C_1+C_2)(\alpha_{L_{1}}+\alpha_{L_{2}}) (\frac{\log \left(t(\alpha_{L_{1}}+\alpha_{L_{2}}) / \epsilon\right)}{\log \log \left(t(\alpha_{L_{1}}+\alpha_{L_{2}}) / \epsilon\right)})^2\right)
\end{align}
where $\alpha_{1} = \left\|\mathcal{L}_{1}(t)\right\|_{\mathrm{be}, 1}, \alpha_{L_{2}} = \left\|\mathcal{L}_{2}(t)\right\|_{\mathrm{be}, 1}$;  $C_1$ and $C_2$ representing the gate complexity of implement $U_{H_1}$ and the maximum gate complexity of implement $U_{H_2(t)},U_{L_j}$ respectively.

Meanwhile for the simulation algorithm in interaction picture, the time complexity is given by the following theorem.

\begin{theorem}[Gate complexity of Lindbladian simulation in the interaction picture]  \label{theo:inter_simu}
    Suppose we are given $U_{H_1}$,$U_{H_2(t)}$ and $U_{L_j}$ block encoding of $H_1, H_2(t)$ and $L_j$ respectively, such that $e^{-i H s}$ is approximated to error $\epsilon$ using $C_{e^{-i H_1 s}}[\epsilon] \in \mathcal{O}\left(|s| \log ^\gamma(s / \epsilon)\right)$ gates for some $\gamma>0$ and any $|s| \geq 0$.

For all $t>0$, the time-evolution \cref{eq:int_iter} may be approximated to error $\epsilon$ with gate complexity
\begin{align} \label{eq:inter_simu}
&C_{\mathrm{interact}} \notag \\
&= \mathcal{O}\left(\alpha_{L_{2}} t\left(C_2+C_{e^{-i A / \alpha_{L_{2}}}}\left[\frac{\epsilon}{\alpha_{L_{2}} t \log \left(\alpha_{L_{2}}\right)}\right] \log \left(\frac{t\left(\alpha_{L_{1}}+\alpha_{L_{2}}\right)}{\epsilon}\right)\right) \frac{\log \left(\alpha_{L_{2}} t / \epsilon\right)}{\log \log \left(\alpha_{L_{2}} t / \epsilon\right)}\right) \notag \\
& =\mathcal{O}\left(\alpha_{L_{2}} t\left(C_2+C_{e^{-i A / \alpha_{L_{2}}}}[\epsilon]\right) \operatorname{polylog}\left(t\left(\alpha_{L_{1}}+\alpha_{L_{2}}\right) / \epsilon\right)\right)
\end{align}
where  $\alpha_{1} = \left\|\mathcal{L}_{1}(t)\right\|_{\mathrm{be}, \infty}, \alpha_{L_{2}} = \left\|\mathcal{L}_{2}(t)\right\|_{\mathrm{be}, \infty} = \left\|\mathcal{L}_{2,I}(t)\right\|_{\mathrm{be}, \infty}.$
\end{theorem}

The proof follows by using $\alpha_{L_1},\alpha_{L_2,I}$ to substitute $\alpha_A$ and $\alpha_B$ in Theorem 7 in \cite{low2018hamiltonian}, respectively. 

We highlight that the assumption $C_{e^{-i H_1 s}}[\epsilon] = O\left(|s| \log ^\gamma(s / \epsilon)\right)$ imposes strong requirement on simulating the $H_1$ dynamics. With Hamiltonian simulation algorithm \cite{BCSWW20}, gate complexity should be
$C_{e^{-i H_1 s}}[\epsilon] = \tilde{O} (||H_1||s)$. But here the assumption removes the $||H_1||$ dependence. This implies that the simulation of $H_1$ is supposed to be easy, the dynamics can be fast-forwarded. Nevertheless this assumption is valid in some common settings~\cite{low2018hamiltonian}, for instance when $H_{1}$ is diagonal.
Another assumption is Eq.~\eqref{eq:interaction-assum}, which implies that $\alpha_2 \ll  \alpha_1$. By comparing Eq.~\eqref{eq:direct_simu} and Eq.~\eqref{eq:inter_simu} with this relation, we find the simulation strategy using the interacting picture  has a better gate complexity.
As long as these two assumptions hold, the simulation algorithm in the interaction picture  can serve as an alternative to reduce the simulation complexity.

\section{The Optimization Algorithm for Quantum Optimal Control}\label{sec:optimization}
In this section, we present our main results for finding first- and second-order stationary points of the optimization problem induced by the quantum optimal control problem \eqref{qoc-opt}, which in general is nonconvex. We consider the  accelerated gradient descent (AGD) method~\cite{jin2017accelerated}. A key departure from a direct implement of AGD is that the gradient has to be estimated using the quantum algorithm \cite{gilyen2019optimizing}, in which case, the gradient input is subject to noise. We believe that this result may be of general interest to the optimization community.
 \begin{theorem}  \label{theorem:opti}
Assume that the function $f(\cdot)$ is $\ell$-smooth and $\varrho$-Hessian Lipschitz. There exists an absolute constant $c_{\max }$ such that for any $\delta>0, \epsilon \leq \frac{\ell^2}{\varrho}, \Delta_f \geq f\left(\mathbf{x}_0\right)-f^{\star}$, if $\chi=\max \left\{1, \log \frac{d \ell \Delta_f}{\varrho \in \delta}\right\}$, $c \geq c_{\max }$ and such that if we run modified PAGD (Algorithm \ref{alg:cap}) with the choice of parameters in \cref{sec:hyper} using an approximate gradient $\hat{\nabla} f(x)$  with error bounded at every step: $\norm{\nabla f(x)-\hat{\nabla} f(x)}\le \epsilon_g$  with
\begin{align}\label{eq:9/8}
         \epsilon_g = \frac{\varrho^{1 / 8}}{\sqrt{2} \ell^{1 / 4} \chi^{3 / 2} c^{3 / 2}} \epsilon^{9 / 8},
     \end{align}
then with probability at least $1-\delta$,  one of the iterates $\mathbf{x}_t$ will be an $\epsilon$-first order stationary point in the following number of iterations:
\begin{align}
O\left(\frac{\ell^{1 / 2} \varrho^{1 / 4}\left(f\left(\mathbf{x}_0\right)-f^*\right)}{\epsilon^{7 / 4}} \log ^6\left(\frac{d \ell \Delta_f}{\varrho \epsilon \delta}\right)\right).
\end{align}
Furthermore, if the error bound of the gradient is chosen as,
     \begin{align}
         \epsilon_g = \frac{\delta\chi^{-11} c^{-16}}{64 \ell} \frac{\epsilon^3}{\sqrt{d}} \frac{1}{\Delta_f},
     \end{align}
     then with probability at least $1-\delta$, one of the iterates $\mathbf{x}_t$ will be an $\epsilon$-second order stationary point.
 \end{theorem}

The proof of this theorem is deferred to Appendix \ref{app:APGD_proof}. Note that the complexity $\tilde{O}(1/\epsilon^{7/4})$ in~\cite{jin2017accelerated} is the currently best-known result for finding first- and second-order stationary points using only gradient queries, and there is not much space to improve as~\cite{carmon2021lower} proved a lower bound $\Omega(1/\epsilon^{12/7})$ for deterministic algorithms with gradient queries when the function is gradient- and Hessian-Lipschitz. Our error bound $\tilde{O}(1/\epsilon^{9/8})$ in~\eqref{eq:9/8} is optimal (up to poly-logarithmic factors) for PAGD because up to a concentration inequality, it can give an algorithm for stochastic gradient descent with complexity $\tilde{O}(1/\epsilon^{7/4}\cdot(1/\epsilon^{9/8})^{2})=\tilde{O}(1/\epsilon^{4})$, which is optimal as there is a matching lower bound $\Omega(1/\epsilon^4)$~\cite{arjevani2023lower}. In other words, if the error $\epsilon^{9/8}$ can be further improved, it implies an algorithm for finding stationary points with better convergence than~\cite{jin2017accelerated}, the current state-of-the-art work on this.

The AGD algorithm relies on an estimate of the gradient. Toward this end, we first show that the objective function \eqref{eq:j1} from the quantum control problem is essentially a polynomial. The polynomial nature of the objective function allows us to use high-order finite difference methods to compute the gradient. In particular, a centered difference scheme with $2m+1$ points will produce an exact gradient for a polynomial of degree $2m.$ 
\begin{lemma}
    Assume that the control function is expressed as a linear combination of shape function $b_j(t)$: $u(t) = \sum_{j=0}^N u_j b_j(t)$ and let $\bm u=(u_0, u_1, \ldots, u_N)$. Then the expectation in \cref{eq:j1} from the Lindblad simulation algorithms from the previous section is a polynomial with degree $d ={O}\left( T\,\mathrm{polylog} \frac{1}{\epsilon} \right)$.
\end{lemma}
\begin{proof}
    We begin by examining the time-dependent unitary $V(0,t)$ in Duhamel's representation. Specifically, from \cref{eq:v0t}, we see that the Dyson series approximation yields a polynomial of degree at most $K.$ In addition, in the Kraus form approximation in \cref{eq:gk}, the operators $\mathcal{L}_\mathrm{J}(s)$ do not involve the control variable $\bm u.$ Overall, the approximation $\mathcal{G}_K(t)$  in \cref{eq:gk} constitutes a polynomial of degree at most $K^2$. Therefore, after applying $\mathcal{G}_K(\delta)$ for $T/\delta$ times to approximate the density operator at time $T$, we obtain a polynomial of degree at most $T\text{polylog} \frac{1}{\epsilon}$. Here we have used the fact that $ K = \frac{\log(1/\epsilon)}{\log\log(1/\epsilon)}.$  Furthermore, when the gradient estimation algorithm  in \cref{lemma:jordan} is applied, the query complexity \eqref{J-polynomial} becomes $\widetilde{O}(T/\epsilon).$ 
\end{proof}

\medskip

\section{Proof of Main Theorem}
Finally, we outline the proof of our main theorem (\Cref{thm:main}). We first summarize our quantum algorithm as follows,

\begin{algorithm}
    \caption{Quantum Algorithm for Open System Quantum Control}\label{alg:main}
    \begin{algorithmic}[1]
    \State Given $k_{\max }, \epsilon_g$ as in Theorem \ref{theorem:opti}; set $u(t)=0$
    \For{t = 0,1,...,$k_{\max}$}
        \State Use Theorem \ref{thm:ldb-td} and strategy in Section \ref{sec:gradient-est} to construct the phase oracle for $\widetilde{J}_1(\boldsymbol{u})$;
        \State Use Lemma \ref{lemma:grad-est-J} to estimate $\boldsymbol{g}^{(k)} \approx \nabla J\left(\boldsymbol{u}^{(k)}\right)$ with $||\boldsymbol{g}^{(k)} - J\left(\boldsymbol{u}^{(k)}\right)|| \le \epsilon_g$;
        \State Update control variable with one step of modified PAGD, as shown in Algorithm \ref{alg:cap}
    \EndFor

    \end{algorithmic}
\end{algorithm}

Now, we restated the main theorem and give its proof:

\begin{theorem}[main theorem, restated] \label{theo:main}
   Assume there are $n_{\mathrm{c}}$ control functions  $u_\beta(t) \in \mathrm{C}^2([0,T])$. Further assume\footnote{More generally, if $\norm{H_0}, \norm{\mu} = \Theta(\Lambda)$, it is equivalent to enlarge the time duration $T$ by a factor $O(\Lambda)$.} that $\norm{H_0}, \norm{\mathscr{O}}, \norm{\mu_{\beta}}, \norm{L_j} \leq 1$, and $\alpha \geq 2/T$. 
There exists a quantum algorithm that, with probability at least $2/3$\footnote{Using standard techniques, the success probability can be boosted to a constant arbitrarily close to 1 while only introducing a logarithmic factor in the complexity.}, solves problem \eqref{qoc-opt} by:
\begin{itemize}
\item reaching a first-order stationary point $\|\nabla{f}\|<\epsilon$ with \eqref{eq:Lindblad} using
      $\widetilde{O}\left(\frac{n_{\mathrm{c}} \|\mathcal{L}\|_{\text {be }, 1}T}{\epsilon^{23/8}}\Delta_f\right)$
  queries to $\mathcal{P}_{H_0}$ and $\mathcal{P}_{\mu_\beta}$, $\beta=1,2,\ldots, n_{\mathrm{c}}$, and
    $\widetilde{O}\left(mn\frac{n_{\mathrm{c}} \|\mathcal{L}\|_{\text {be }, 1} T}{\epsilon^{23 / 8}} \Delta_f + n\frac{ T^{3/2}}{\epsilon^{9 / 4}} \Delta_f\right)$
  additional 1- and 2-qubit gates; or
  
\item reaching a second-order stationary point using
      $\widetilde{O}\left(\frac{n_{\mathrm{c}}  \|\mathcal{L}\|_{\text {be }, 1}T^{7/4}}{\epsilon^{5}}\Delta_f\right)$
  queries to $\mathcal{P}_{H_0}$ and $\mathcal{P}_{\mu_\beta}$, $\beta=1,2,\ldots, n_{\mathrm{c}}$ and
    $\widetilde{O}\left(mn\frac{n_{\mathrm{c}}  \|\mathcal{L}\|_{\text {be }, 1}T^{7/4}}{\epsilon^{5}} \Delta_f + n\frac{ T^{3/2}}{\epsilon^{9 / 4}} \Delta_f\right)$
  additional 1- and 2-qubit gates.
  \end{itemize} 
  Here $n_{\mathrm{c}}$ and $m$ are respectively the number of control variables and jump operators.
  \end{theorem}

\begin{proof}
We denote gate complexity of control $U_{\mathscr{O}}$ by $ C_{\mathrm{c}-U_{\mathscr{O}}}$, gate complexity of $U_H,U_{L_j}$ by $C_{U_H,U_{L_j}}$, and gate complexity of quantum simulation by $C_{\varrho (t)}$.

The gate complexity of preparing a state after Lindblad evolution is given by Theorem \ref{thm:ldb-td},
\begin{align}
    C_{\varrho (t)} = O\left(\|\mathcal{L}\|_{\text{be},1} \frac{\log \left(\|\mathcal{L}\|_{\text{be},1} / \epsilon\right)}{\log \log \left(\|\mathcal{L}\|_{\text {be},1} / \epsilon\right)}\right) C_{U_H,U_{L_j}}+ \tilde{O}\left(m\|\mathcal{L}\|_{\text{be},1} n\right).
    \label{eq:sim_comp}
\end{align}

With copies of states $\varrho (t)$ and access to control $U_{\mathscr{O}}$ oracle, we can construct the gradient following  Section \ref{sec:gradient-est}. According to that section, we can construct the probability oracle with $\varrho (t)$, construct the phase oracle with probability oracle, and calculate the gradient with the phase oracle. The corresponding complexity is listed below:
\begin{align}
    C_{U_{J_1}} &= C_{\varrho(t)} + O(1)  + C_{c-U_O}, \\
    C_{\mathcal{O}_{J_1}} &= O(\log 1/\epsilon) C_{U_{J_1}}, \\
    C_{\nabla J} &= \tilde{O} (n_{\mathrm{c}} T \log (N / \gamma) / \epsilon) C_{\mathcal{O}_{J_1}} + \tilde{O} (N),
\end{align}
where $1-\gamma$ is the successful probability of obtaining a gradient, $n_{\mathrm{c}}$ is the number of parameters, and $N$ is the time steps $N = O\left(t^{3 / 2} / \epsilon^{1 / 2}\right)$ as in~\cite[Corollary 2.2]{li2023efficient}. 
Here we define $\gamma = \nu / k$, where $\nu$ is a small finite number and $k$ is the iteration steps, which we will give below. Combining them together, we have 
\begin{align}
    C_{\nabla J}
    =  \tilde{O} \left(n_{\mathrm{c}}\frac{\|\mathcal{L}\|_{\text {be,1}} T \log \frac{N}{\gamma}}{\epsilon}\right) C_{U_H,U_{L_j}} &+ \tilde{O} (n_{\mathrm{c}}\frac{T\log \frac{N}{\gamma}}{\epsilon}) C_{\mathrm{c-}U_O} \notag \\  + \tilde{O} (m n n_{\mathrm{c}}\frac{\|\mathcal{L}\|_{\text {be,1}} T \log \frac{N}{\gamma}}{\epsilon}+N).  
\end{align} 
Here we reassign the gradient noise $\epsilon$ with $\epsilon_g$ to distinguish from the other errors.
\begin{align}
    C_{\nabla J}
    =  \tilde{O} \left(n_{\mathrm{c}}\frac{\|\mathcal{L}\|_{\text {be,1}} T \log \frac{N}{\gamma}}{\epsilon_g}\right) C_{U_H,U_{L_j}} &+ \tilde{O} (n_{\mathrm{c}}\frac{T\log \frac{N}{\gamma}}{\epsilon_g}) C_{\mathrm{c-}U_O} \notag \\ 
    &+ \tilde{O} (m n n_{\mathrm{c}}\frac{\|\mathcal{L}\|_{\text {be,1}} T \log \frac{N}{\gamma}}{\epsilon_g}+N).  \label{eq:grad_comp}
\end{align}

With modified PAGD method(Algorithm \ref{alg:cap}), we can find a first or second order $\epsilon$-stationary point within 
\begin{align}
    k=\tilde{O}\left(\frac{\ell^{1 / 2} \varrho^{1 / 4}\left(f\left(\mathbf{x}_0\right)-f^*\right)}{\epsilon^{7 / 4}}\right) \label{eq:steps}
\end{align}
iterations by Theorem \ref{theorem:opti}. For first $\epsilon$-order stationary point, the gradient noise tolerance is
\begin{align}
    \epsilon_g=\frac{\varrho^{1 / 8}}{\sqrt{2} \ell^{1 / 4} \chi^{3 / 2} c^{3 / 2}} \epsilon^{9 / 8}.
\end{align}
For second order $\epsilon$-order stationary point, it is 
\begin{align}
    \epsilon_g=\frac{\delta  \chi^{-11} c^{-16}}{64 \ell} \frac{\epsilon^3}{\sqrt{d}} \frac{1}{\Delta_f}.
\end{align}

In each iteration, we need to calculate $\nabla J$ once and calculate $J$ once. Noticing that $C_J = O(C(\varrho (t)))$, we have
\begin{align}
    C_{\mathrm{total}} &= k\times (C_{\nabla J} + C_{\varrho (t)}). \label{eq:total_comp}
\end{align}

Substitute \cref{eq:sim_comp,eq:grad_comp},~\eqref{eq:steps} into \cref{eq:total_comp} we finish the proof. Notice that in optimization, dimension $d = N$, and here we regard $\ell$ and $\varrho$ as constants.
\end{proof}

\section*{Acknowledgments}
TL was supported by National Natural Science Foundation of China (Grant Numbers 62372006 and 92365117), and The Fundamental Research Funds for the Central Universities, Peking University. CW acknowledges support from from National Science Foundation grant CCF-2312456.
XL was supported by National Science Foundation grant CCF-2312456 and  DMS-2111221.

\bibliographystyle{plain}
\bibliography{bibtex,applications}

\appendix
\section{Time-ordered Exponential}
\label{app: timeordered}
The ordered exponential is a math concept defined in non-commutative algebras, which is a counterpart of exponential of integral in commutative algebras.
The ordered exponential of $J$ is denoted as
\begin{equation}
    \begin{aligned}
        OE[a](t)& = \mathcal{T}\{e^{\int_0^t J(t')\mathrm{d}t'}\} = \sum_{n=0}^{\infty}\frac{1}{n!}\int_0^t\cdots\int_0^t\mathcal{T}\{J(t_1')\cdots J(t_n')\}\mathrm{d}t_1'\cdots \mathrm{d}t_n'\\
        & = \sum_{n=0}^{\infty}\int_0^t\int_0^{t_n'}\cdots\int_0^{t_2'}J(t_n')\cdots J(t_1')\mathrm{d}t_1'\cdots\mathrm{d}t_n'.
    \end{aligned}
\end{equation}
From this definition we can observe that $t_1'\leq t_2'\leq \cdots \leq t_n'\leq t$ and it shows the ordered time property.

\subsection{Dyson Series}
\label{app: Dyson}
Consider a Schr\"odinger equation with time-dependent Hamiltonian
\begin{equation}
    \partial_t\psi = J(t)\psi,\quad\psi(0) = \psi_0
\end{equation}
Let $U(t_0,t)$ be the fundamental solution satisfying
\begin{equation}
    \partial_tU(t,t_0) = J(t)U(t,t_0),\quad U(t_0,t_0) = I
\end{equation}
and $\psi(t) = U(t,t_0)\psi_0$ by the definition of fundamental solution. Using fundamental theorem of calculus, we know that
\begin{equation}
    U(t,t_0) = I+\int_{t_0}^t J(t_1)U(t_1,t_0)\mathrm{d}t_1
\end{equation}
Applying the expression again for $U(t_1,t_0)$, we obtain
\begin{equation}
    \begin{aligned}
        U(t,t_0) &= I+\int_{t_0}^t J(t_1)\mathrm{d}t_1+\int_{t_0}^{t}\int_{t_0}^{t_1}J(t_1)J(t_2)U(t_2,t_0)\mathrm{d}t_2\mathrm{d}t_1\\
        &\vdots \\
        U(t,t_0)&=I+\sum_{n=1}^{\infty}\int_{t_0}^{t}\int_{t_0}^{t_1}\cdots\int_{t_0}^{t_{n-1}}J(t_1)\cdots J(t_n)\mathrm{d}t_n\cdots\mathrm{d}t_1\\
    \end{aligned}
\end{equation}
The time-order operator $\mathcal{T}$ can change the order of product such that
\begin{equation}
    \mathcal{T}(J(t_1)J(t_2))=
    \begin{cases}
    J(t_1)J(t_2)& \text{if } t_1\geq t_2\\
    J(t_2)J(t_1)& otherwise
    \end{cases}
     \\
\end{equation}
so that
\begin{equation}
\int_{t_0}^{t}\int_{t_0}^{t_1}J(t_1)J(t_2)\mathrm{d}t_2\mathrm{d}t_1 = \frac{1}{2}\int_{t_0}^{t}\int_{t_0}^{t}\mathcal{T}[J(t_1)J(t_0)]\mathrm{d}t_2\mathrm{d}t_1
\end{equation}
In general, the time-ordered operator is defined as
\begin{equation}
    \mathcal{T}[J(t_1)J(t_2)\cdots J(t_n)] = J(t_{i_1})J(t_{i_2})\cdots J(t_{i_n})
\end{equation}
where the subscript $i_1,\cdots,i_n$ is a reordering of $1,2,\cdots,n$ such that $t_{i_1}\geq t_{i_2}\geq \cdots \geq t_{i_n}$. Then the fundamental solution can be written as
\begin{equation}
\begin{aligned}
    U(t,t_0) &= I+\sum_{n=1}^{\infty}\frac{1}{n!}\int_{t_0}^t\int_{t_0}^t\cdots\int_{t_0}^t\mathcal{T}[J(t_1)\cdots J(t_n)]\mathrm{d}t_n\cdots\mathrm{d}t_1   
     = \mathcal{T}e^{\int_{t_0}^t J(s)\mathrm{d}s}.
\end{aligned}
\end{equation}

\section{Upper bounds for the gradients}\label{sec:upper-bound-gradient}
In this section, we prove \Cref{lem:upper-bound-gradient} that gives an upper bound on the gradient of $\widetilde{J}_1(\bm{u})$ in \Cref{eq:j1}. 

Consider the Lindblad master equation \ref{eq:Lindblad}. For simplicity, we consider a single control variable, expressed as a piecewise linear function in terms of the shape (or hat) functions:  $\sum_{\beta=0}^Nu_{\beta}b_{\beta}(t)$,
\begin{equation}
    \partial_t\varrho = -i[H_0+\sum_{\beta=0}^Nu_{\beta}b_{\beta}(t)\mu,\varrho]+\sum_jL_j\varrho L_j^{\dag}-\{L_j^{\dag}L_j,\varrho\} = \mathcal{L}\varrho.
\end{equation}
Let the first order derivative $\Gamma = \frac{\partial \varrho}{\partial u_{\alpha_1}} = \Gamma^{\alpha_1}$, it follows the equation,
\begin{equation}\label{eq: Gamma}
\begin{aligned}
    \partial_t\Gamma &= -i[H_0+\sum_{\beta=0}^Nu_{\beta}b_{\beta}(t)\mu,\Gamma] +\sum_jL_j\Gamma L_j^{\dag}-\{L_j^{\dag}L_j,\Gamma\}-i[b_{\alpha_1}\mu,\varrho]\\
    & = \mathcal{L} \Gamma-i[b_{\alpha_1}\mu,\varrho],\\
    \Gamma(0) &= 0.\\
\end{aligned}
\end{equation}

We first make the observation that the operator $\mathcal{L}$ is the generator of a Markovian dynamics, and $\norm{e^{t \mathcal{L}}  }$ is a completely positive map, satisfying the property that \cite{ruskai1994beyond},
\[  \norm{e^{t \mathcal{L}_I} A  } \leq \norm{A}, \quad \forall A \in \mathbb{C}^{N\times N}.  \]
Using the Duhamel's principle, this implies that,
\begin{equation}
    \|\Gamma(t)\|\leq\int_0^t \norm{ [b_{\alpha} \mu, \varrho ] }  \mathrm{d}\tau
 \leq \delta t \norm{\mu}.
\end{equation}

Similarly, one can differentiate \cref{eq: Gamma}, and let $\Gamma_2= \frac{\partial^2 \varrho}{\partial u_{\alpha_1}\partial u_{\alpha_2}} = \Gamma_2^{(\alpha_1,\alpha_2)}$, which satisfies the equation,
\begin{equation}\label{eq: Gamma2}
    \partial_t\Gamma_2 =\mathcal{L} \Gamma_2 -i[b_{\alpha_1}\mu,\Gamma^{\alpha_2} ]-i[b_{\alpha_2}\mu,\Gamma^{\alpha_1} ].
    \end{equation}
The same argument leads to the bound,
\begin{equation}
    \|\Gamma_2(t)\|\leq 2\int_0^t \norm{ [b_{\alpha} \mu, \Gamma(\tau) ] }  \mathrm{d}\tau
 \leq  2 \delta t^2 \norm{\mu}^2.
\end{equation}
More generally, for $\Gamma_k = \frac{\partial^k\varrho}{\partial u_{\alpha_1}\cdots\partial u_{\alpha_k}}$, it satisfies
\begin{equation}
    \partial\Gamma_k =  \mathcal{L} \Gamma_k-i\sum_{l=1}^k[b_{\alpha_l}\mu,\Gamma_{k-1}]
\end{equation}
Then by induction
\begin{equation}
    \|\Gamma_k\|\leq k\int_0^t\|[b_{\alpha_l}\mu,\Gamma_{k-1}] \|\mathrm{d}\tau\leq k!\delta t^k\|\mu\|^k.
\end{equation}

Using these bounds,
the derivative of the objective function $\Tilde{J}_1(\mathbf{u})$ can be bounded as follows
\begin{equation}
\begin{aligned}
    \norm{ \frac{\partial^{\bm{\alpha}} \widetilde{J}_1}{\partial u_{\alpha_1} u_{\alpha_2} \cdots u_{\alpha_k} }   } &\leq \sum_{l=0}^k
    {k\choose l}\norm{\Gamma_l}\|O\otimes I\|\norm{\Gamma_{k-l}}\\
    &\leq\sum_{l=0}^k {k\choose l}l!\delta t^l\norm{\mu}^l(k-l)!\delta t^{k-l}\norm{\mu}^{k-l} \|O\otimes I\|\\
    &= \sum_{l=0}^k k!\delta t^k\norm{\mu}^k\|O\| = (k+1)!\delta t^k\|\mu\|^k\|O\|.\\
\end{aligned}
\label{eq: J1 drv}
\end{equation}

Using similar techniques, we can show that $\tilde{J}_1(\bm u)$ is $\ell$-smooth and $\varrho$-Hessian Lipschitz such that
\begin{equation}
    \|\grad \tilde{J}_1(\bm u)-\grad \tilde{J}_1(\bm v)\|\leq \ell\|\bm u-\bm v\|, \quad  \|\grad^2 \tilde{J}_1(\bm u)-\grad^2 \tilde{J}_1(\bm v)\|\leq \varrho\|\bm u-\bm v\|,
\end{equation}
where $\|\cdot\|$ is $l_2$-norm for vectors and spectral norm for matrices. For smoothness, with the observation that there exist second order derivatives of $\tilde{J}_1(\bm u)$, it is sufficient to show that  $\nabla^2 \tilde{J}_1(u)$ is bounded, which can be obtained directly from \cref{eq: J1 drv}.
To show that $J_1(\bm u)$ is $\varrho$-Hessian Lipschitz, we utilize the property of induced matrix norm,
\begin{equation}
    \begin{aligned}
        &\|\grad^2 \tilde{J}_1(\bm u)-\grad^2 \tilde{J}_1(\bm v)\| = \sup_{\bm w\in\mathbb{R}^{N+1}}\frac{\|(\grad^2 \tilde{J}_1(\bm u)-\grad^2 \tilde{J}_1(\bm v))\bm w\|}{\|\bm w\|}\\
        =& \sup_{\bm w\in\mathbb{R}^{N+1}}\frac{\|(\int_0^1\partial_{\lambda}\grad^2 \tilde{J}_1(\bm u+\lambda(\bm u-\bm v))\mathrm{d}\lambda)\bm w\|}{\|\bm w\|}\\
        \leq & \sup_{\bm w\in\mathbb{R}^{N+1}}\frac{\int_0^1\|\partial_{\lambda}\grad^2 \tilde{J}_1(\bm u+\lambda(\bm u-\bm v))\|\mathrm{d}\lambda \|\bm w\|}{\|\bm w\|}\leq \max_{\lambda\in[0,1]} \|\partial_{\lambda}\grad^2 \tilde{J}_1(\bm u+\lambda(\bm u-\bm v))\|.
    \end{aligned}
\end{equation}
Consider the general form \cref{eq:j1}, the derivatives can be calculated as follows
\begin{equation}
    \begin{aligned}
        \partial_i \tilde J_1(\bm u) &= \bra{\varrho_N}O\otimes I\ket{\partial_i\varrho_N}+c.c.\\
        \partial_{i,j}^2\tilde J_1(\bm u) &= \bra{\varrho_N}O\otimes I\ket{\partial_{i,j}^2\varrho_N}+\bra{\partial_j\varrho_N}O\otimes I\ket{\partial_i\varrho_N}+c.c.\\
        \partial_{\lambda}\partial_{i,j}^2\tilde J_1(\bm u+\lambda(\bm u-\bm v)) &= \bra{\partial_{\lambda}\varrho_N(\bm u-\bm v)}O\otimes I\ket{\partial_{i,j}^2\varrho_N}+\bra{\partial_{\lambda}\partial_j\varrho_N(\bm u-\bm v)}O\otimes I\ket{\partial_i\varrho_N}\\
        &+\bra{\varrho_N}O\otimes I\ket{\partial_{\lambda}\partial_{i,j}^2\varrho_N(\bm u-\bm v)}+\bra{\partial_j\varrho_N}O\otimes I\ket{\partial_{\lambda}\partial_i\varrho_N(\bm u-\bm v)}+c.c.\\
        |\partial_{\lambda}\partial_{i,j}^2\tilde J_1(\bm u+\lambda(\bm u-\bm v))|&\leq 6\|\Gamma_1\|\|O\|\|\Gamma_2\|\|\bm u-\bm v\|+2\|\varrho_N\|\|O\|\|\Gamma_3\|\|\bm u-\bm v\| \\
        &\leq 4!\delta t^3\|\mu\|^3\|O\|\|\bm u-\bm v\|.
    \end{aligned}
\end{equation}
These bounds show that
\begin{align}
    \|\partial_{\lambda}\grad^2 \tilde{J}_1(\bm u+\lambda(\bm u-\bm v))\|&\leq \|\partial_{\lambda}\grad^2 \tilde{J}_1(\bm u+\lambda(\bm u-\bm v))\|_{\infty} =\max_i\sum_j|\partial_{\lambda}\partial_{i,j}^2\tilde J_1(\bm u+\lambda(\bm u-\bm v))|\nonumber \\
    &\leq (N+1)4!\delta t^3\|\mu\|^3\|O\|\|\bm u-\bm v\|   =: \varrho \|\bm u -\bm v\|,
\end{align}
which gives \cref{eq: J1 drv} when $k = 3$.
The Hessian Lipschitz coefficients $\varrho = 4!(N+1)\delta t^3 \|\mu\|^3\|O\|$ and smooth coefficients $\ell = 3!(N+1)\delta t^2\|\mu\|^2\|O\|$.

\section{Perturbed Accelerated Gradient Descent with Noisy Gradients} \label{appendix:PAGD}
In this section, we are going to prove that Perturbed Accelerated Gradient Descent (PAGD)~\cite{jin2017accelerated} can be used in our scheme, which implies that the step for optimization can be $\widetilde{O}\left(1 / \epsilon^{7 / 4}\right)$ instead of $\widetilde{O}\left(1 / \epsilon^{2}\right)$. Because gradient obtained by Jordan's algorithm \cite{Jordan_2005} is noisy, here we present a proof that as long as the error of gradient is bounded in every step, the asymptotic complexity $\widetilde{O}\left(1 / \epsilon^{7 / 4}\right)$ still holds. To distinguish from the real gradient $\nabla f(x)$, we denote the noisy gradient calculated by Jordan's algorithm \cite{gilyen2019optimizing}  as $\hat{\nabla} f(x)$, and denote the corresponding noise as $e(x)$:
\begin{align}
    \mathbf{e}\left(\mathbf{x}\right) = \nabla f(x) - \hat{\nabla} f(x).
\end{align}
We denote the upper bound of this noise as $\epsilon_g$,
\begin{align}
    \max_\tau ||\mathbf{e}(\mathbf{x_\tau})|| \le \epsilon_g,
\end{align}
where $\tau$ iterates all time steps in the optimization process.
Substituting $\nabla f(x)$ with $\hat{\nabla} f(x)$ to accommodate our scheme, we get a modified PAGD, as shown in Algorithm \ref{alg:cap}.

The major theorem in original paper~\cite{jin2017accelerated} is its Theorem 3, which bounds the time complexity by showing that when encountering no $\epsilon$-second-order stationary point, the Hamiltonian will decrease by $\mathscr{E}$ in at most $2 \mathscr{T}$ steps. In this Appendix, we prove Theorem \ref{theorem:opti} with a similar strategy.

 The proof is divided into three cases according to gradient, Hessian, and whether the negative curvature exploitation (NCE) step is invoked:
\begin{enumerate}
    \item Large gradient scenario: gradient is larger than $\epsilon$ in $\mathscr{T}$ steps.  We are going to prove that   $E_{\mathscr{T}}-E_0 \leq-\mathscr{E}$. Gradient noise is bounded by $\epsilon_g \le \tilde{O}(\epsilon^{1.125})$.
    \item Negative curvature scenario: at some step, we stuck in an saddle point, we use perturbation to escape from it. We are going to prove that we have high probability to escape from the saddle point. Gradient noise is bounded by $\epsilon_g \le \tilde{O}(\frac{ \sqrt{\varrho} }{ \ell} \frac{\epsilon^3}{\sqrt{d}} \Delta_f)$.
    \item NCE: we are going to prove that if NCE step is invoked,  Hamiltonian decreases by $\mathscr{E}$ in the next step. Gradient noise is bounded by $\epsilon_g \le \tilde{O}(\epsilon)$.
\end{enumerate}
The choices of these parameters are given in \Cref{sec:hyper}.
We first prove that improve or localize theorem still holds with $\epsilon_g = O(\epsilon)$ in Appendix \ref{app:stay_or_loc}. This theorem serves as a key ingredient in following parts. Then we show that gradient error accumulation does not affect the performance of our algorithm in all three cases cases listed above in Appendix \ref{app:large_grad}, \ref{app:neg_curv}, \ref{app:NCE} respectively. Our main technical contribution in this section is Appendix \ref{append:error-accumulation}, which gives a detailed analysis of accumulated gradient noise in PAGD. 
Finally in Appendix \ref{app:APGD_proof}, we put  together the final complexity.

\subsection{Choices of parameters}\label{sec:hyper}
Due to the complexity behind the analysis of the global convergence, numerous  hyper parameters have to be introduced, and they are summarized as follows:
\begin{itemize}
    \item $\epsilon$: condition of sufficiently small gradient $||\nabla f || < \epsilon$
    \item $\eta$: learning rate $x_{t+1} \leftarrow y_t - \eta \nabla f(y_t)$
    \item $1-\theta$: rate of momentum: $y_t \leftarrow x_t + (1-\theta)v_t$
    \item $\gamma$: condition of large non-convex: $f\left(\mathbf{x}_t\right) \leq f\left(\mathbf{y}_t\right)+\left\langle\hat{\nabla} f\left(\mathbf{y}_t\right), \mathbf{x}_t-\mathbf{y}_t\right\rangle-\frac{\gamma}{2}\left\|\mathbf{x}_t-\mathbf{y}_t\right\|^2$
    \item $\ell$: gradient Lipschitz: $\left\|\nabla f\left(\mathbf{x}_1\right)-\nabla f\left(\mathbf{x}_2\right)\right\| \leq \ell\left\|\mathbf{x}_1-\mathbf{x}_2\right\| \quad \forall \mathbf{x}_1, \mathbf{x}_2 .$
    \item $\varrho$: Hessian Lipschitz: $\left\|\nabla^2 f\left(\mathbf{x}_1\right)-\nabla^2 f\left(\mathbf{x}_2\right)\right\| \leq \varrho\left\|\mathbf{x}_1-\mathbf{x}_2\right\| \quad \forall \mathbf{x}_1, \mathbf{x}_2$
    \item $\kappa:=\ell / \sqrt{\varrho \epsilon}$
    \item $s$: NCE step length
    \item $c,\chi$: two large enough constants
    \item $\mathscr{T}=\sqrt{\kappa} \cdot \chi c$: a period
    \item $\mathscr{E}:=\sqrt{\frac{\epsilon^3}{\varrho}} \cdot \chi^{-5} c^{-7}$: scale of accumulated error in $\mathscr{T}$ steps
    \item $\mathscr{S}:=\sqrt{\frac{2 \eta \mathscr{T} \mathscr{E}}{\theta}}=\sqrt{\frac{2 \epsilon}{\varrho}} \cdot \chi^{-2} c^{-3}$
    \item $\mathscr{M}:=\frac{\epsilon \sqrt{\kappa}}{\ell} c^{-1}$: small momentum condition $\left\|\mathbf{v}_0\right\| \leq \mathscr{M}$
    \item $\Delta_f := f(y_0) - f(y^*)$: the difference between initial function value and the minimum function value.
    \item $r = \eta \epsilon \cdot \chi^{-5} c^{-8}$: the radius for perturbation in escaping the saddle point.
\end{itemize}
Settings of hyper parameters:
\begin{itemize}
    \item $\eta=\frac{1}{4 \ell}$
    \item $\theta=\frac{1}{4 \sqrt{\kappa}}$
    \item $\gamma=\frac{\theta^2}{\eta}$
    \item $s=\frac{\gamma}{4 \varrho}$
    \item $\mathcal{S}$: the subspace with eigenvalues in $\left(\theta^2 /\left[\eta(2-\theta)^2\right], \ell\right]$
    \item $\mathcal{S}^c$: the subspace with eigenvalues in $\left[-\ell , \theta^2 /\left[\eta(2-\theta)^2\right]\right]$
\end{itemize}
It is noticed that $\theta^2 /\left[\eta(2-\theta)^2\right]=\Theta(\sqrt{\varrho \epsilon})$.

\subsection{Improve or localize theorem}[Modified from lemma 4 in \cite{jin2017accelerated}; Hamiltonian decreases monotonically]\label{app:stay_or_loc}
\begin{lemma}\label{lem:improve_or_loc_1}
\label{lem:PAGD2}
Assume that the function $f(\cdot)$ is $\ell$-smooth and set the learning rate to be $\eta \leq \frac{1}{2 \ell}, \theta \in\left[2 \eta \gamma, \frac{1}{2}\right]$ in $A G D$ (Algorithm 1). Then, for every iteration $t$ where NCE condition $f\left(\mathbf{x}_t\right) \leq f\left(\mathbf{y}_t\right)+\left\langle\hat{\nabla} f\left(\mathbf{y}_t\right), \mathbf{x}_t-\mathbf{y}_t\right\rangle-\frac{\gamma}{2}\left\|\mathbf{x}_t-\mathbf{y}_t\right\|^2$ does not hold, we have:
    \begin{align}
    E_{t+1} &\leq E_t-\frac{\theta}{2 \eta}\left||\mathbf{v}_t\right||^2-\frac{\eta}{4}\left\|\nabla f\left(\mathbf{y}_t\right)+\mathbf{e}\left(\mathbf{y}_t\right)\right\|^2+\eta\left\|\mathbf{e}\left(\mathbf{y}_t\right)\right\|^2 \\
     E_{t+1} &\leq E_t-\frac{\theta}{2 \eta}\left||\mathbf{v}_t\right||^2-\frac{\eta}{8}\left\|\nabla f\left(\mathbf{y}_t\right)\right\|^2+\frac{5}{4}\eta\left\|\mathbf{e}\left(\mathbf{y}_t\right)\right\|^2
    \end{align}
\end{lemma}
\begin{proof}
    \begin{align}
        f\left(\mathbf{x}_{t+1}\right) \leq f\left(\mathbf{y}_t\right)-\eta \nabla f\left(\mathbf{y}_t\right)\cdot \hat{\nabla} f\left(\mathbf{y}_t\right)+\frac{\ell \eta^2}{2}\left\|\hat{\nabla} f\left(\mathbf{y}_t\right)\right\|^2
    \end{align}
    \begin{align}
        f\left(\mathbf{x}_t\right) \geq f\left(\mathbf{y}_t\right)+\left\langle\hat{\nabla} f\left(\mathbf{y}_t\right), \mathbf{x}_t-\mathbf{y}_t\right\rangle-\frac{\gamma}{2}\left\|\mathbf{y}_t-\mathbf{x}_t\right\|^2
    \end{align}
    \begin{align}
        \left\|\mathbf{x}_{t+1}-\mathbf{x}_t\right\|^2 & =\left\|\mathbf{y}_t-\mathbf{x}_t-\eta \hat{\nabla} f\left(\mathbf{y}_t\right)\right\|^2 \\
        & =\left[(1-\theta)^2\left\|\mathbf{x}_t-\mathbf{x}_{t-1}\right\|^2-2 \eta\left\langle\hat{\nabla} f\left(\mathbf{y}_t\right), \mathbf{y}_t-\mathbf{x}_t\right\rangle+\eta^2\left\|\hat{\nabla}f\left(\mathbf{y}_t\right)\right\|^2\right]
        \end{align}
        so we have
        \begin{align}
            &f\left(\mathbf{x}_{t+1}\right)+\frac{1}{2 \eta}\left\|\mathbf{x}_{t+1}-\mathbf{x}_t\right\|^2 \\
            \leq & f\left(\mathbf{x}_t\right)+\frac{1}{2 \eta}\left\|\mathbf{x}_t-\mathbf{x}_{t-1}\right\|^2-\frac{\theta}{2 \eta}\left\|\mathbf{v}_t\right\|^2-\frac{\eta}{4}\left\|\hat{\nabla} f\left(\mathbf{y}_t\right)\right\|^2 -\eta \mathbf{e}(\mathbf{y}_{t})\cdot \hat{\nabla}f(\mathbf{y}_{t}) \\
             = & f\left(\mathbf{x}_t\right)+\frac{1}{2 \eta}\left\|\mathbf{x}_t-\mathbf{x}_{t-1}\right\|^2-\frac{\theta}{2 \eta}\left\|\mathbf{v}_t\right\|^2-\frac{\eta}{4}\left\|\hat{\nabla} f\left(\mathbf{y}_t\right)+2\mathbf{e}(\mathbf{y}_{t})\right\|^2 + \eta ||\mathbf{e}(\mathbf{y}_{t})||^2\\
             = & f\left(\mathbf{x}_t\right)+\frac{1}{2 \eta}\left\|\mathbf{x}_t-\mathbf{x}_{t-1}\right\|^2-\frac{\theta}{2 \eta}\left\|\mathbf{v}_t\right\|^2-\frac{\eta}{4}\left\|{\nabla} f\left(\mathbf{y}_t\right)+\mathbf{e}(\mathbf{y}_{t})\right\|^2 + \eta ||\mathbf{e}(\mathbf{y}_{t})||^2.
        \end{align}
        We can further simplified the inequality with
        \begin{align}
            \left\|\nabla f\left(\mathbf{y}_t\right)+\mathbf{e}\left(\mathbf{y}_t\right)\right\|^2 \ge \frac{1}{2} \left\|\nabla f\left(\mathbf{y}_t\right)\right\|^2 - \left\|\mathbf{e}\left(\mathbf{y}_t\right)\right\|^2,
        \end{align}
        and we have
        \begin{align}
            &f\left(\mathbf{x}_{t+1}\right)+\frac{1}{2 \eta}\left\|\mathbf{x}_{t+1}-\mathbf{x}_t\right\|^2 \notag \\ 
            &\le f\left(\mathbf{x}_t\right)+\frac{1}{2 \eta}\left\|\mathbf{x}_t-\mathbf{x}_{t-1}\right\|^2-\frac{\theta}{2 \eta}\left\|\mathbf{v}_t\right\|^2-\frac{\eta}{8}\left\|\nabla f\left(\mathbf{y}_t\right)\right\|^2+\frac{5}{4}\eta\left\|\mathbf{e}\left(\mathbf{y}_t\right)\right\|^2.
        \end{align}
\end{proof}

\begin{algorithm}
    \caption{modified perturbed AGD}\label{alg:cap}
    \begin{algorithmic}[1]
    \State $\boldsymbol{v}_0 \gets 0$
    \For{t = 0,1,...}
        \If{$||\hat{\nabla} f(\mathbf{x}_t)|| \leq \epsilon $ and no perturbation in last $\mathscr{T}$steps }

         $\mathbf{x}_t \leftarrow \mathbf{x}_t+\xi_t \quad \xi_t \sim \operatorname{Unif}\left(\mathbb{B}_0(r)\right)$
        \EndIf
        \State $\mathbf{y}_t \leftarrow \mathbf{x}_t+(1-\theta) \mathbf{v}_t$
        \State $\mathbf{x}_{t+1} \leftarrow \mathbf{y}_t-\eta \hat{\nabla} f\left(\mathbf{y}_t\right)$
        \State $\mathbf{v}_{t+1} \leftarrow \mathbf{x}_{t+1}-\mathbf{x}_t$
        \If{$f\left(\mathbf{x}_t\right) \leq f\left(\mathbf{y}_t\right)+\left\langle\hat{\nabla} f\left(\mathbf{y}_t\right), \mathbf{x}_t-\mathbf{y}_t\right\rangle-\frac{\gamma}{2}\left\|\mathbf{x}_t-\mathbf{y}_t\right\|^2$}
        \State $\left(\mathbf{x}_{t+1}, \mathbf{v}_{t+1}\right) \leftarrow \text { Negative-Curvature-Exploitation }\left(\mathbf{x}_t, \mathbf{v}_t, s\right)$
        \EndIf
    \EndFor

    \end{algorithmic}
\end{algorithm}

\begin{theorem}[Modified from Corollary 6 of \cite{jin2017accelerated}]  Under the same setting as in Lemma \ref{lem:improve_or_loc_1},  we have:
   \label{theo:localize_or_stay}
\begin{align}
\sum_{\tau=t+1}^{t+T}||\mathbf{x}_\tau-\mathbf{x}_{\tau-1}||^2 \leq \frac{2 \eta}{\theta}\left(E_t-E_{t+T}\right) + \frac{2\eta^2}{\theta}\sum_{t+1}^{t+T}|| \mathbf{e}\left(\mathbf{y}_t\right) ||^2
\label{eq:localize_or_stay}
\end{align}
\end{theorem}
\begin{proof}
    \cref{eq:localize_or_stay} can be directly obtained by Lemma \ref{lem:PAGD2}
\end{proof}
\begin{proposition}
Suppose that $\max_\tau ||\mathbf{e}\left(\mathbf{y}_\tau\right)|| < \frac{1}{2\sqrt{2}} \epsilon$. Then the following inequality holds,
    \begin{align}
        \sum_{\tau=t+1}^{t+\mathscr{T}}||\mathbf{x}_\tau-\mathbf{x}_{\tau-1}||^2 \leq \frac{2 \eta}{\theta}\left(E_t-E_{t+\mathscr{T}}\right)+\frac{\eta}{\theta}\mathscr{E}.
        \end{align}
\end{proposition}
\begin{proof}
    Noticed that
    \begin{align}
        \frac{2 \eta^2}{\theta}\sum_{t+1}^{t+\mathscr{T}}\left\| \mathbf{e}\left(\mathbf{y}_t\right)\right\|^2 \le \frac{2 \eta^2}{\theta} \mathscr{T} \max _\tau\left\|\mathbf{e}\left(\mathbf{y}_\tau\right)\right\|^2 \le \frac{2 \eta^2}{\theta} \mathscr{T} \frac{\epsilon^2}{8} \le \frac{\eta}{\theta}\mathscr{E}
    \end{align}
    Substitute it into \cref{eq:localize_or_stay}
\end{proof}
In modified PAGD, we can define
\begin{align}
    \hat{\delta}_\tau = \delta_\tau + e\left(y_\tau\right)
\end{align}
    \begin{align}
        \hat{\nabla} f\left(\mathbf{y}_\tau\right)&=\nabla f(\mathbf{0})+\mathcal{H} \mathbf{y}_\tau+\delta_\tau + e\left(y_\tau\right) \\
        &=\nabla f(\mathbf{0})+\mathcal{H} \mathbf{y}_\tau+\hat{\delta}_\tau
    \end{align}

\subsection{Large gradient scenario} 
\label{app:large_grad}

In this subsection, we suppose that gradient is large,
\begin{align}
    \left\|\nabla f\left(\mathbf{x}_\tau\right)\right\| \geq \epsilon \text{  for  } \tau \in [0,\mathscr{T}]
\end{align}
\begin{lemma}[Modified from Lemma 15 in \cite{jin2017accelerated}]\label{aux:modi_lem_15}
      if $\left\|\mathbf{v}_t\right\| \geq \mathscr{M}$ or $\left\|\nabla f\left(\mathbf{x}_t\right)\right\| \geq 2 \ell \mathscr{M}$, and $e(\mathbf{y}_t) <\frac{1}{c^2}\epsilon$,and at time step $t$ only modified AGD is used without NCE or perturbation, then:
\begin{align}
E_{t+1}-E_t \leq-4 \mathscr{E} / \mathscr{T}.
\end{align}
\end{lemma}
\begin{proof}
    When $\left\|\mathbf{v}_t\right\| \geq \frac{\epsilon \sqrt{\kappa}}{10 \ell}$, by Lemma \ref{lem:PAGD2}, we have:
\begin{align}
E_{t+1}-E_t &\leq-\frac{\theta}{2 \eta}\left\|\mathbf{v}_t\right\|^2 +\frac{5}{4} \eta\left\|\mathbf{e}\left(\mathbf{y}_t\right)\right\|^2 \\
&\leq-\Omega\left(\frac{\ell}{\sqrt{\kappa}} \frac{\epsilon^2 \kappa}{\ell^2} c^{-2}\right)+\frac{5}{4}\eta (\frac{\epsilon}{c^2})^2=-\Omega\left(\frac{\epsilon^2 \sqrt{\kappa}}{2 \ell} c^{-2}\right)+O(c^{-4}\epsilon^2/\ell) \\
&\leq-\Omega\left(\frac{\mathscr{E}}{\mathscr{T}} c^6\right)+O(\frac{\mathscr{E}}{\mathscr{T}} c^4 / \sqrt{\kappa}) \\
&\leq-\frac{4 \mathscr{E}}{\mathscr{T}} .
\end{align}
In last step, we choose c to be large enough . When $\left\|\mathbf{v}_t\right\| \leq \mathscr{M}$ but $\left\|\nabla f\left(\mathbf{x}_t\right)\right\| \geq 2 \ell \mathscr{M}$, by the gradient Lipschitz assumption, we have:
$$
\left\|\nabla f\left(\mathbf{y}_t\right)\right\| \geq\left\|\nabla f\left(\mathbf{x}_t\right)\right\|-(1-\theta) \ell\left\|\mathbf{v}_t\right\| \geq \ell \mathscr{M} .
$$
Similarly, by Lemma \ref{lem:PAGD2}, we have:
$$
\begin{aligned}
  &  E_{t+1}-E_t \leq-\frac{\eta}{4}\left\|\nabla f\left(\mathbf{y}_t\right)\right\|^2 +\frac{5}{4} \eta\left\|\mathbf{e}\left(\mathbf{y}_t\right)\right\|^2  \\ & \leq-\Omega\left(\frac{\epsilon^2 \kappa}{\ell} c^{-2}\right)+O\left(c^{-4} \epsilon^2 \ell^{-1}\right) \leq-\Omega\left(\frac{\mathscr{E}}{\mathscr{T}} c^6\right)+O\left(\frac{\mathscr{E}}{\mathscr{T}} c^4  \kappa^{-1}\right) \leq-\frac{4 \mathscr{E}}{\mathscr{T}}.
\end{aligned}
$$
Again we choose c to be large enough, and thus finish the proof.
\end{proof}
\begin{lemma}[Modified from Lemma 16 in \cite{jin2017accelerated}]\label{aux:modi_lem_16}
Under the setting of Theorem $\mathbf{S}$, if $\left\|\mathcal{P}_{\mathcal{S}^c} \nabla f\left(\mathbf{x}_0\right)\right\| \geq \frac{\epsilon}{2}$, $\left\|\mathbf{v}_0\right\| \leq \mathscr{M}, \mathbf{v}_0^{\top}\left[\mathcal{P}_{\mathcal{S}}^{\top} \nabla^2 f\left(\mathbf{x}_0\right) \mathcal{P}_{\mathcal{S}}\right] \mathbf{v}_0 \leq 2 \sqrt{\varrho \epsilon} \mathscr{M}^2$, $\max _\tau\left\|\mathbf{e}\left(\mathbf{y}_\tau\right)\right\|<\epsilon / \chi c \sqrt{2 \mathscr{T}} = \frac{\varrho^{1/8}}{\sqrt{2}\ell^{1/4}(\chi c)^{3/2}}\epsilon^{9/8}$, and for $t \in[0, \mathscr{T} / 4]$ only modified AGD steps are used without NCE or perturbation, then:
\begin{align}   
E_{\mathscr{T} / 4}-E_0 \leq-\mathscr{E} .
\end{align}
\end{lemma}
\begin{proof}
    The proof is almost the same as in \cite{jin2017accelerated}, the only difference is that every $\delta_\tau$ should be replaced by $\hat{\delta}_\tau=\delta_\tau+e\left(y_\tau\right)$, i.e. the last equation in page 22 and the first two equations in page 23 should be re-evaluate:
    \begin{align}
        \begin{aligned}
\left|\hat{\tilde{\delta}}^{(j)}\right| & =\left|\sum_{\tau=0}^{t-1} p_\tau^{(j)} \hat{\delta}_\tau^{(j)}\right| \\
& \leq \left|\sum_{\tau=0}^{t-1} p_\tau^{(j)} \delta_\tau^{(j)}\right| +|\sum_{\tau=0}^{t-1} p_\tau^{(j)}e^{(j)}\left(y_\tau\right) |\\
&\leq \sum_{\tau=0}^{t-1} p_\tau^{(j)}\left(\left|\delta_0^{(j)}\right|+\left|\delta_\tau^{(j)}-\hat{\delta}_0^{(j)}\right|\right)+|\sum_{\tau=0}^{t-1} p_\tau^{(j)}e^{(j)}\left(y_\tau\right) | \\
&\le \sum_{\tau=0}^{t-1} p_\tau^{(j)}\left(\left|\hat{\delta}_0^{(j)}\right|+\left|\hat{\delta}_\tau^{(j)}-\hat{\delta}_0^{(j)}\right|\right) + |\sum_{\tau=0}^{t-1} p_\tau^{(j)}e^{(j)}\left(y_\tau\right) |\\
& \leq\left[\sum_{\tau=0}^{t-1} p_\tau^{(j)}\right]\left(\left|\delta_0^{(j)}\right|+\sum_{\tau=1}^{t-1}\left|\delta_\tau^{(j)}-\delta_{\tau-1}^{(j)}\right|\right) + |\sum_{\tau=0}^{t-1} p_\tau^{(j)}e^{(j)}\left(y_\tau\right) | \\
&\leq \left|\delta_0^{(j)}\right|+\sum_{\tau=1}^{t-1}\left|\delta_\tau^{(j)}-\delta_{\tau-1}^{(j)}\right|  + |\sum_{\tau=0}^{t-1} p_\tau^{(j)}e^{(j)}\left(y_\tau\right) |
\end{aligned}
    \end{align}
    Similar to Lemma 16 in \cite{jin2017accelerated}, we know
    \begin{align}
        &\left\|\mathcal{P}_{\mathcal{S}^c} \tilde{\delta}\right\|^2 \\
        =&\sum_{j \in \mathcal{S}^c}\left|\hat{\tilde{\delta}}^{(j)}\right|^2 \\
        \le & \sum_{j \in \mathcal{S}^c}\left(\left|\delta_0^{(j)}\right|+\sum_{\tau=1}^{t-1}\left|\delta_\tau^{(j)}-\delta_{\tau-1}^{(j)}\right|+\left|\sum_{\tau=0}^{t-1} p_\tau^{(j)} e^{(j)}\left(y_\tau\right)\right|\right)^2 \\
        \leq & 2\sum_{j \in \mathcal{S}^c}\left(\left|\delta_0^{(j)}\right|+\sum_{\tau=1}^{t-1}\left|\delta_\tau^{(j)}-\delta_{\tau-1}^{(j)}\right|\right)^2+\left|\sum_{\tau=0}^{t-1} p_\tau^{(j)} e^{(j)}\left(y_\tau\right)\right|^2 \\
        \leq & 4\left\|\delta_0\right\|^2+4 t \sum_{\tau=1}^{t-1}\left\|\delta_\tau-\delta_{\tau-1}\right\|^2 + 2\sum_j\left|\sum_{\tau=0}^{t-1} p_\tau^{(j)} e^{(j)}\left(y_\tau\right)\right|^2 \label{eq:PAGD_c11}
    \end{align}
    and
    \begin{align}
        \sqrt{4\left\|\delta_0\right\|^2+4 t \sum_{\tau=1}^{t-1}\left\|\delta_\tau-\delta_{\tau-1}\right\|^2} \le O\left(\varrho \mathscr{S}^2\right) \leq O\left(\epsilon \cdot c^{-6}\right) \leq \epsilon / 10
    \end{align}
    By Lemma \ref{lem:aux1}, the third term is also bounded by $\epsilon^2$, so we have
    \begin{align}
        \left\|\mathcal{P}_{\mathcal{S}^c} \tilde{\delta}\right\|^2 \le \epsilon^2.
    \end{align}

\end{proof}
\begin{lemma}[Modified from Lemma 17 in \cite{jin2017accelerated}]\label{aux:modi_lem_17}
Under the setting of Theorem G, suppose $\left\|\mathbf{v}_0\right\| \leq \mathscr{M}$ and $\left\|\nabla f\left(\mathbf{x}_0\right)\right\| \leq 2 \ell \mathscr{M}, E_{\mathscr{T} / 2}-E_0 \geq-\mathscr{E}$, $\max _\tau\left\|\mathbf{e}\left(\mathbf{y}_\tau\right)\right\|<\epsilon /  c \sqrt{2 \mathscr{T}}=\frac{\rho^{1 / 8}}{\sqrt{2} \ell^{1 / 4}\chi^{1/2} c^{3 / 2}} \epsilon^{9 / 8}$, and for $t \in[0, \mathscr{T} / 2]$ only modified AGD steps are used, without NCE or perturbation. Then $\forall t \in[\mathscr{T} / 4, \mathscr{T} / 2]$:
$$
\left\|\mathcal{P}_{\mathcal{S}} \nabla f\left(\mathbf{x}_t\right)\right\| \leq \frac{\epsilon}{2} \quad \text { and } \quad \mathbf{v}_t^{\top}\left[\mathcal{P}_{\mathcal{S}}^{\top} \nabla^2 f\left(\mathbf{x}_0\right) \mathcal{P}_{\mathcal{S}}\right] \mathbf{v}_t \leq \sqrt{\varrho \epsilon} \mathscr{M}^2.
$$
\end{lemma}
\begin{proof}
    In our case, the expression of $\nabla f\left(\mathbf{x}_t\right)$ and $\mathbf{v}_t$ should each have one more term representing the effect of gradient error.
    \begin{align}
        \nabla f\left(\mathbf{x}_t\right) &= \boldsymbol{g}_1+\boldsymbol{g}_2+\boldsymbol{g}_3+\boldsymbol{g}_4+\boldsymbol{g}_5, \\
        \boldsymbol{g}_5 &= -\eta \mathcal{H}\left(\begin{array}{ll}
\mathbf{I} & 0
\end{array}\right) \sum_{\tau=0}^{t-1} \mathbf{A}^{t-1-\tau}\left(\begin{array}{c}
 \mathbf{e}\left(\mathbf{y}_\tau\right) \\
0
\end{array}\right),
    \end{align}
where $\boldsymbol{g}_1$,$\boldsymbol{g}_2$,$\boldsymbol{g}_3$, and $\boldsymbol{g}_4$ are defined as in Lemma 17 in \cite{jin2017accelerated}.
    \begin{align}
    \mathbf{v}_t &= \boldsymbol{m}_1 + \boldsymbol{m}_2 + \boldsymbol{m}_3 + \boldsymbol{m}_4, \\
    \boldsymbol{m}_4 &= -\eta\left(\begin{array}{ll}
    1 & -1
    \end{array}\right) \sum_{\tau=0}^{t-1} \mathbf{A}^{t-1-\tau}\left(\begin{array}{c}
    \mathbf{e}\left(\mathbf{y}_\tau\right) \\
    0
    \end{array}\right).
    \end{align}
where $\boldsymbol{m}_1$,$ \boldsymbol{m}_2$, $ \boldsymbol{m}_3$, and $\boldsymbol{m}_3$ are defined as in Lemma 17 in \cite{jin2017accelerated}.

We will prove that
\begin{align}
    \left\|\mathcal{P}_{\mathcal{S}} \boldsymbol{g}_5\right\|^2 < \epsilon^2/64,
\end{align}
and 
\begin{align}
    \left\|\left[\mathcal{P}_{\mathcal{S}}^{\top} \nabla^2 f\left(\mathbf{x}_0\right) \mathcal{P}_{\mathcal{S}}\right]^{\frac{1}{2}} \boldsymbol{m}_4\right\|^2 < \frac{1}{3} \sqrt{\varrho \epsilon} \mathscr{M}^2 = \frac{1}{3} \frac{\epsilon^2 \ell}{c^2},
\end{align}
so that the proof still holds.
The proof is in Lemma \ref{lem:aux2} and \ref{lem:aux3}
\end{proof}
\subsection{Negative curvature scenario} \label{app:neg_curv}

\begin{lemma}[Modified from Lemma 18 of \cite{jin2017accelerated}]\label{aux:modi_lem_18}
Under the same setting as Theorem \$, suppose $\|\nabla f(\tilde{\mathbf{x}})\| \leq \epsilon$ , $\lambda_{\min }\left(\nabla^2 f(\tilde{\mathbf{x}})\right) \leq-\sqrt{\varrho \epsilon}$ and $\max_\tau \left\|\mathbf{e}\left(\mathbf{y}_\tau\right)\right\| \leq \frac{\theta }{8 \eta \mathscr{T}}\frac{\delta \mathscr{E}}{2 \Delta_f}  \frac{r}{\sqrt{d}} = \frac{\delta \sqrt{\rho}\chi^{-11}c^{-16}}{64l}\frac{\epsilon^3}{\sqrt{d}}\Delta_fx$. Let $\mathbf{x}_0$ and $\mathbf{x}_0^{\prime}$ be at distance at most $r$ from $\tilde{\mathbf{x}}$. Let $\mathbf{x}_0-\mathbf{x}_0^{\prime}=r_0 \cdot \mathbf{e}_1$ and let $\mathbf{v}_0=\mathbf{v}_0^{\prime}=\tilde{\mathbf{v}}$ where $\mathbf{e}_1$ is the minimum eigen-direction of $\nabla^2 f(\tilde{\mathbf{x}})$. Let $r_0 \geq \frac{\delta \mathscr{E}}{2 \Delta_f} \cdot \frac{r}{\sqrt{d}}$. Then, running PAGD starting at $\left(\mathbf{x}_0, \mathbf{v}_0\right)$ and $\left(\mathbf{x}_0^{\prime}, \mathbf{v}_0^{\prime}\right)$ respectively, we have:
$$
\min \left\{E_{\mathscr{T}}-\widetilde{E}, E_{\mathscr{T}}^{\prime}-\widetilde{E}\right\} \leq-\mathscr{E},
$$
where $\widetilde{E}, E_{\mathscr{T}}$ and $E_{\mathscr{T}}^{\prime}$ are the Hamiltonians at $(\tilde{\mathbf{x}}, \tilde{\mathbf{v}}),\left(\mathbf{x}_{\mathscr{T}}, \mathbf{v}_{\mathscr{T}}\right)$ and $\left(\mathbf{x}_{\mathscr{T}}^{\prime}, \mathbf{v}_{\mathscr{T}}^{\prime}\right)$ respectively.
\end{lemma}

\begin{proof}
    The only difference from original proof is the treatment of
    \begin{align}
        \min \left\{E_{\mathscr{T}}-E_0, E_{\mathscr{T}}^{\prime}-E_0^{\prime}\right\},
    \end{align}
    in our case, because we have a modified improve or localize theorem \ref{theo:localize_or_stay},
    We can show that with $e_{\nabla f}< \epsilon$ and $A_1$ smaller enough, we have
\begin{align}
    \max \left\{\left\|\mathbf{x}_t-\tilde{\mathbf{x}}\right\|,\left\|\mathbf{x}_t^{\prime}-\tilde{\mathbf{x}}\right\|\right\} \leq r+\max \left\{\left\|\mathbf{x}_t-\mathbf{x}_0\right\|,\left\|\mathbf{x}_t^{\prime}-\mathbf{x}_0^{\prime}\right\|\right\} \leq r+\sqrt{9 \eta \mathscr{T} \mathscr{E} / \theta} \leq 3 \mathscr{S}.
\end{align}
And  in our case, the expression of $\mathbf{w}_t$ should each have one more term representing the effect of gradient error
\begin{align*}
    \mathbf{w}_t=\left(\begin{array}{ll}
\mathbf{I} & 0
\end{array}\right) \mathbf{A}^t\left(\begin{array}{c}
\mathbf{w}_0 \\
\mathbf{w}_0
\end{array}\right)-\eta\left(\begin{array}{ll}
\mathbf{I} & 0
\end{array}\right) \sum_{\tau=0}^{t-1} \mathbf{A}^{t-1-\tau}\left(\begin{array}{c}
\delta_\tau \\
0
\end{array}\right) \\
-\eta\left(\begin{array}{ll}
\mathbf{I} & 0
\end{array}\right) \sum_{\tau=0}^{t-1} \mathbf{A}^{t-1-\tau}\left(\begin{array}{c}
\mathbf{e}\left(\mathbf{y}_\tau\right) - \mathbf{e}\left(\mathbf{y}_{\tau-1}\right) \\
0
\end{array}\right).
\end{align*}
 This holds if we can prove that
\begin{align}
\left\|\left(\begin{array}{ll}
\mathbf{I} & 0
\end{array}\right) \mathbf{A}^t\left(\begin{array}{l}
\mathbf{w}_0 \\
\mathbf{w}_0
\end{array}\right)\right\|
-
    \left\|\eta\left(\begin{array}{ll}
\mathbf{I} & 0
\end{array}\right) \sum_{\tau=0}^{t-1} \mathbf{A}^{t-1-\tau}\left(\begin{array}{c}
\mathbf{e}\left(\mathbf{y}_\tau\right)-\mathbf{e}\left(\mathbf{y}_{\tau}'\right) \\
0
\end{array}\right)\right\| \ge \Omega (\frac{\theta}{4}(1+\Omega(\theta))^t r_0).
\end{align}

    Notice that $\mathbf{w}_0 = r_0 \cdot \mathbf{e}_1$, where subscript 1 represents the minimum eigenvalue direction, we have
    \begin{align}
        \left\|\left(\begin{array}{ll}
\mathbf{I} & 0
\end{array}\right) \mathbf{A}^t\left(\begin{array}{l}
\mathbf{w}_0 \\
\mathbf{w}_0
\end{array}\right)\right\| =  \left(a_\tau^{(1)}-b_\tau^{(1)}\right)r_0.
    \end{align}
For the second term, we use Cauchy–Schwarz inequality,
\begin{align*}
    \left\|\eta\left(\begin{array}{ll}
\mathbf{I} & 0
\end{array}\right) \sum_{\tau=0}^{t-1} \mathbf{A}^{t-1-\tau}\left(\begin{array}{c}
\mathbf{e}\left(\mathbf{y}_\tau\right)-\mathbf{e}\left(\mathbf{y}_\tau^{\prime}\right) \\
0
\end{array}\right)\right\|
&\le
\eta \sum_{\tau=0}^{t-1}\left\|\left(\begin{array}{ll}
\mathbf{I} & 0
\end{array}\right)  \mathbf{A}^{t-1-\tau}\left(\begin{array}{c}
\mathbf{e}\left(\mathbf{y}_\tau\right)-\mathbf{e}\left(\mathbf{y}_\tau^{\prime}\right) \\
0
\end{array}\right)\right\| \\
& \le \eta \sum_{\tau=0}^{t-1}\left\|\left(\begin{array}{ll}
\mathbf{I} & 0
\end{array}\right)  \mathbf{A}^{t-1-\tau}\left(\begin{array}{c}
\mathbf{I} \\
0
\end{array}\right)\right\| \left\| \mathbf{e}\left(\mathbf{y}_\tau\right)-\mathbf{e}\left(\mathbf{y}_\tau^{\prime}\right) \right\|.
\end{align*}
According to Lemma 18 and Lemma 32 in \cite{jin2017accelerated}, $(\mathbf{I}, 0) \mathbf{A}^{t-\tau}\left(\begin{array}{l}
\mathbf{I} \\
0
\end{array}\right)$ achieves its spectral norm along the first coordinate which corresponds to the eigenvalue $\lambda_{\min }(\mathcal{H})$, so
\begin{align*}
    \left\|\eta\left(\begin{array}{ll}
\mathbf{I} & 0
\end{array}\right) \sum_{\tau=0}^{t-1} \mathbf{A}^{t-1-\tau}\left(\begin{array}{c}
\mathbf{e}\left(\mathbf{y}_\tau\right)-\mathbf{e}\left(\mathbf{y}_\tau^{\prime}\right) \\
0
\end{array}\right)\right\|
&\le \eta \sum_{\tau=0}^{t-1}a_{t-1-\tau}^{(1)}\left\|\mathbf{e}\left(\mathbf{y}_\tau\right)-\mathbf{e}\left(\mathbf{y}_\tau^{\prime}\right)\right\| \\
&\le 2\eta \max_\tau \left\|\mathbf{e}\left(\mathbf{y}_\tau\right)\right\| \sum_{\tau=0}^{t-1} a_{\tau}^{(1)} \\
&\le \frac{\theta r_0}{4\mathscr{T}} \sum_{\tau=0}^{t-1} a_\tau^{(1)} \\
& \le \frac{\theta r_0}{4}a_\mathscr{T}^{(1)}.
\end{align*}

In last equality, we use $a_\tau^{(1)} \le a_{\tau+1}^{(1)}$. Combining these results together, we have
\begin{align}
    &\left\|\left(\begin{array}{ll}
\mathbf{I} & 0
\end{array}\right) \mathbf{A}^t\left(\begin{array}{l}
\mathbf{w}_0 \\
\mathbf{w}_0
\end{array}\right)\right\|-\left\|\eta\left(\begin{array}{ll}
\mathbf{I} & 0
\end{array}\right) \sum_{\tau=0}^{t-1} \mathbf{A}^{t-1-\tau}\left(\begin{array}{c}
\mathbf{e}\left(\mathbf{y}_\tau\right)-\mathbf{e}\left(\mathbf{y}_\tau^{\prime}\right) \\
0
\end{array}\right)\right\| \\
\ge &
\left((1-\frac{\theta}{4})a_\mathscr{T}^{(1)}-b_\mathscr{T}^{(1)}\right) r_0.
\end{align}
Following Lemma 18 and Lemma 32 in \cite{jin2017accelerated}, we have finished the proof.
\end{proof}

\subsection{Negative curvature exploitation}\label{app:NCE}
\begin{lemma}[Modified from Lemma 5 in \cite{jin2017accelerated}]\label{aux:modi_lem_5}
     Assume that $f(\cdot)$ is $\ell$-smooth and $\varrho$-Hessian Lipschitz. For every iteration $t$ of Algorithm where NCE condition $f\left(\mathbf{x}_t\right) \leq f\left(\mathbf{y}_t\right)+\left\langle\hat{\nabla} f\left(\mathbf{y}_t\right), \mathbf{x}_t-\mathbf{y}_t\right\rangle-\frac{\gamma}{2}\left\|\mathbf{x}_t-\mathbf{y}_t\right\|^2$ holds (thus running NCE), we have:
$$
E_{t+1} \leq E_t-\min \left\{\frac{s^2}{2 \eta}, \frac{1}{2}(\gamma-2 \varrho s) s^2-s ||\mathbf{e}\left(\mathbf{y}_t\right)||\right\}.
$$
\end{lemma}
\begin{proof}
    When $\left\|\mathbf{v}_t\right\| \geq s$, the same; when $\left\|\mathbf{v}_t\right\| < s$, we have
    \begin{align}
        \left\langle\mathbf{e}\left(\mathbf{y}_t\right), \mathbf{x}_t-\mathbf{y}_t\right\rangle+\frac{1}{2}\left(\mathbf{x}_t-\mathbf{y}_t\right)^{\top} \nabla^2 f\left(\zeta_t\right)\left(\mathbf{x}_t-\mathbf{y}_t\right) \leq-\frac{\gamma}{2}\left\|\mathbf{x}_t-\mathbf{y}_t\right\|^2
    \end{align}
    \begin{align}
        E_{t+1} &\le E_t - \frac{1}{2}(\gamma-2 \varrho s) s^2 - \left\langle\mathbf{e}\left(\mathbf{y}_t\right), \mathbf{x}_t-\mathbf{y}_t\right\rangle \\
        & \le E_t - \frac{1}{2}(\gamma-2 \varrho s) s^2 + s ||\mathbf{e}\left(\mathbf{y}_t\right)||
    \end{align}

    Then we evaluate the $\min$ function value. Note that,
    \begin{equation}
\frac{1}{2}(\gamma-2 \varrho s) s^2
= \frac{1}{2}\theta^2
\frac{s^2}{2 \eta}.
\end{equation}
\end{proof}
Thus we have
\begin{equation}
    \min \left\{\frac{s^2}{2 \eta}, \frac{1}{2}(\gamma-2 \varrho s) s^2-s\left\|\mathbf{e}\left(\mathbf{y}_t\right)\right\|\right\} = \frac{1}{2}(\gamma-2 \varrho s) s^2-s\left\|\mathbf{e}\left(\mathbf{y}_t\right)\right\| = O(\sqrt{\frac{\epsilon^3}{\varrho}} - \sqrt{\frac{\epsilon}{\varrho}}\left\|\mathbf{e}\left(\mathbf{y}_t\right)\right\|)
\end{equation}
so we require $\left\|\mathbf{e}\left(\mathbf{y}_t\right)\right\| = A^{-1}\epsilon$ where $A$ is large enough. Then the following arguments in \cite{jin2017accelerated} still holds.

\subsection{Detailed analysis of accumulated gradient noise}\label{append:error-accumulation}

In this subsection, we provide the proof of some lemmas we used above.
\begin{lemma} \label{lem:aux1}
    Suppose that $\max_\tau \left\|\mathbf{e}\left(\mathbf{y}_\tau\right)\right\| < \epsilon /\chi c \sqrt{2 \mathscr{T}}$, $1<\theta \le 1/4$, we have
    \begin{align}
        \sum_{j\in \mathcal{S}^c}\left|\sum_{\tau=0}^{\mathscr{T}-1} p_\tau^{(j)} e^{(j)}\left(y_\tau\right)\right|^2 \le \epsilon^2,
    \end{align}
    where index $j$ sums over all  eigen-direction of $\mathcal{H}=\nabla^2 f(\mathbf{0})$ satisfying $\lambda_j \in \left[-\ell, \theta^2 /\left[\eta(2-\theta)^2\right]\right]$.
\end{lemma}
\begin{proof}
    \begin{align}
        \sum_{j\in\mathcal{S}^c}\left|\sum_{\tau=0}^{\mathscr{T}-1} p_\tau^{(j)} e^{(j)}\left(y_\tau\right)\right|^2 &\le \mathscr{T} \sum_{j\in\mathcal{S}^c} \sum_{\tau=0}^{\mathscr{T}-1}\left(p_\tau^{(j)} e^{(j)}\left(y_\tau\right)\right)^2 \\
        &\le \mathscr{T}\sum_{\tau=0}^{\mathscr{T}-1} \sum_{j\in\mathcal{S}^c} (\max_j p_\tau^{(j)})^2(e^{(j)}\left(y_\tau\right))^2 \\
        &= \mathscr{T}\sum_{\tau=0}^{\mathscr{T}-1} (\max_j p_\tau^{(j)})^2\sum_{j\in\mathcal{S}^c} (e^{(j)}\left(y_\tau\right))^2  \\
        &\le \mathscr{T}\left(\sum_{\tau=0}^{\mathscr{T}-1} (\max_j p_\tau^{(j)})^2 \right) \max_\tau \left\|\mathbf{e}\left(\mathbf{y}_\tau\right)\right\|^2 \\
        &\le \frac{2\epsilon^2 }{\chi^2 c^2}\left(\sum_{\tau=0}^{\mathscr{T}-1} (\max_j p_\tau^{(j)})^2 \right).
    \end{align}
Now we try to bound the summation $\sum_{\tau=0}^{\mathscr{T}-1} (\max_j p_\tau^{(j)})^2 $. By definition, we have
\begin{align}
    p_\tau^{(j)} &= \frac{a_{\mathscr{T}-1-\tau}^{(j)}}{\sum_{\tau=0}^{\mathscr{T}-1} a_\tau^{(j)}} \label{eq:p_tau_j},\\
    a_\tau^{(j)} &= \left(\begin{array}{ll}
1 & 0
\end{array}\right) \mathbf{A}_j^\tau \left(\begin{array}{l}
1 \\
0
\end{array}\right), \\
&= \frac{\mu_1^{\tau+1}-\mu_2^{\tau+1}}{\mu_1-\mu_2}. \label{eq:a_t_j}
\end{align}
where $\mu_1$ and $\mu_2$($
    \mu_1 \ge \mu_2$) are two roots satisfying
\begin{align}
   & \mu^2-(2-\theta)(1-x) \mu+(1-\theta)(1-x)=0, \\
   & x = \eta \lambda_j.
\end{align}
It is noticed that $\mu_1$ and $\mu_2$ are dependent of eigen-direction $j$, and they are real numbers when  $j\in \mathcal{S}^c$.

Combining \cref{eq:p_tau_j} and \cref{eq:a_t_j}, we have
\begin{align}
    p_\tau^{(j)}&=\frac{\mu_1^{\mathscr{T}-\tau} - \mu_2^{\mathscr{T}-\tau}}{\sum_{\tau_0=1}^{\mathscr{T}} \mu_1^{\tau_0} - \mu_2^{\tau_0}} \\
    &\le \frac{\mu_1^{\mathscr{T}-\tau} - \mu_2^{\mathscr{T}-\tau}}{\sum_{\tau_0=\mathscr{T}-\tau}^{\mathscr{T}} \mu_1^{\tau_0} - \mu_2^{\tau_0}} \\
    & \le \frac{\mu_1^{\mathscr{T}-\tau} }{\sum_{\tau_0=\mathscr{T}-\tau}^{\mathscr{T}} \mu_1^{\tau_0} } \\
    & = \frac{1}{\sum_{\tau_0=0}^{\tau} \mu_1^{\tau_0}}\label{eq:p_tau_j_2}.
\end{align}
Noticed that the reciprocal sum of $\mu_1$ and $\mu_2$ is independent of the eigen-direction $j$,
\begin{align}
    \frac{1}{\mu_1}+\frac{1}{\mu_2} = \frac{2-\theta}{1-\theta},
\end{align}
 and $\mu_1 \ge \mu_2$, we have
\begin{align}
    \mu_1 \ge 1-\frac{\theta}{1-\theta} \ge 1-2\theta. \label{eq:mu1_bound}
\end{align}
Combining \cref{eq:p_tau_j_2} and \cref{eq:mu1_bound}, we have
\begin{align}
    p_\tau^{(j)} &\le \frac{1}{\sum_{\tau_0=0}^{\tau}(1-2\theta)^{\tau_0}} \\
    &= \frac{2\theta}{1-(1-2\theta)^{\tau+1}} \\
    &\le \frac{2\theta}{1-\frac{1}{1+2(\tau+1)\theta}} \\
    &= \frac{1 + 2(\tau+1)\theta}{\tau+1 } \\
    &\le \frac{1 + 2(\mathscr{T}+1)\theta}{\tau+1 } \\
    &\le \frac{\chi c}{\tau + 1}.
\end{align}

Therefore, we have
\begin{align}
    \left(\sum_{\tau=0}^{\mathscr{T}-1}\left(\max _j p_\tau^{(j)}\right)^2\right) \le (\chi c)^2 \sum_\tau \frac{1}{(1+\tau)^2} \le 2 (\chi c)^2.
\end{align}
Then we get $\sum_{j \in \mathcal{S}^c}\left|\sum_{\tau=0}^{\mathscr{T}-1} p_\tau^{(j)} e^{(j)}\left(y_\tau\right)\right|^2 \leq \epsilon^2$.
\end{proof}
\begin{lemma}\label{lem:aux2}
Suppose that $\max_\tau \left\|\mathbf{e}\left(\mathbf{y}_\tau\right)\right\|<\epsilon / \sqrt{2 \mathscr{T}}$, we have
\begin{align}
    \sum_{j\in \mathcal{S}}\left|\eta \lambda_j \sum_{\tau=0}^{t-1} a_\tau^{(j)} e^{(j)}\left(y_{t-1-\tau}\right)\right|^2<\epsilon^2 / 64,
\end{align}
where index $j$ sums over all eigen-direction of $\mathcal{H}=\nabla^2 f(\mathbf{0})$ satisfying $\lambda_j \in\left( \theta^2 /\left[\eta(2-\theta)^2\right],\ell\right]$.
\end{lemma}
\begin{proof}
\begin{align}
\sum_{j \in \mathcal{S}}\left|\eta \lambda_j\sum_{\tau=0}^{\mathscr{T}-1}  a_\tau^{(j)} e^{(j)}\left(y_{t-1-\tau}\right)\right|^2 & \leq \mathscr{T} \sum_{j \in \mathcal{S}} \sum_{\tau=0}^{\mathscr{T}-1}\left(x a_\tau^{(j)} e^{(j)}\left(y_{t-1-\tau}\right)\right)^2 \\
& \leq \mathscr{T} \sum_{\tau=0}^{\mathscr{T}-1} \sum_{j \in \mathcal{S}}\left(\max _j x a_\tau^{(j)}\right)^2\left(e^{(j)}\left(y_{t-1-\tau}\right)\right)^2 \\
& =\mathscr{T} \sum_{\tau=0}^{\mathscr{T}-1}\left(\max _j x a_\tau^{(j)}\right)^2 \sum_{j \in \mathcal{S}}\left(e^{(j)}\left(y_{t-1-\tau}\right)\right)^2 \\
& \leq \mathscr{T}\left(\sum_{\tau=0}^{\mathscr{T}-1}\left(\max _j x a_\tau^{(j)}\right)^2\right) \max_\tau \left\|\mathbf{e}\left(\mathbf{y}_\tau\right)\right\|^2 \\
& \leq \epsilon^2\left(\sum_{\tau=0}^{\mathscr{T}-1}\left(\max _j x a_\tau^{(j)}\right)^2\right).
\end{align}
Thus we have
\begin{align}
    (x a_\tau^{(j)})^2 = (x\frac{r^{\tau} \sin [(\tau+1) \phi]}{ \sin [\phi]})^2,
\end{align}
where
\begin{align}
    r=\sqrt{(1-\theta)(1-x)}, \quad \sin \phi=\sqrt{\left((2-\theta)^2 x-\theta^2\right)(1-x)} / 2 r.
\end{align}

To proceed, we  split $x \in\left(\frac{\theta^2}{(2-\theta)^2}, \frac{1}{4}\right]$ into two cases

Case 1:$x \in\left(\frac{\theta^2}{(2-\theta)^2}, \frac{2 \theta^2}{(2-\theta)^2}\right]$.Using $|\frac{\sin [(\tau+1) \phi]}{\sin [\phi]}|\le \tau$
\begin{align}
    \left(x a_\tau^{(j)}\right)^2 = (O(x(1-\theta)^{\tau/2} \tau))^2 = O(x^2\tau^2(1-\theta)^\tau) = O(\theta^4\tau^2(1-\theta)^\tau).
\end{align}

Case 2: $x \in\left(\frac{2 \theta^2}{(2-\theta)^2}, \frac{1}{4}\right]$. We have $|\sin [(\tau+1) \phi]|\le 1$ and $x=\Theta\left(\sin ^2 \phi\right)$, then we have
\begin{align}
    \left(x a_\tau^{(j)}\right)^2=O(\left(x \frac{r^\tau \sin [(\tau+1) \phi]}{\sin [\phi]}\right)^2) = O((\sqrt{x}r^\tau)^2)=O((1-\theta)^\tau x(1-x)^\tau)<O((1-\theta)^\tau / \tau),
\end{align}
where we make use of $x(1-x)^\tau < O(1/\tau)$.

From these cases we deduce that
\begin{align}
    \left(x a_\tau^{(j)}\right)^2 \le \max \{O\left(\theta^4 \tau^2(1-\theta)^\tau\right),O\left((1-\theta)^\tau / \tau\right)\} \le O\left(\theta^4 \tau^2(1-\theta)^\tau\right)+O\left((1-\theta)^\tau / \tau\right),
\end{align}
implying that
\begin{align}
    \sum_{\tau=0}^{\mathscr{T}-1}\left(\max _j x a_\tau^{(j)}\right)^2 \le \sum_{\tau=0}^{\infty}\left(\max _j x a_\tau^{(j)}\right)^2 \le \sum_{\tau=0}^{\infty} O\left(\theta^4 \tau^2(1-\theta)^\tau\right)+O\left((1-\theta)^\tau / \tau\right) = O(\theta) < O(1),
\end{align}
where we make use of $\sum_{\tau = 1}^{\infty}(1-\theta)^\tau / \tau = \log (1-\theta)$, and $\sum_{\tau = 1}^{\infty}\tau^2(1-\theta)^\tau = O(1/\theta^3)$.

Gathering these bounds, we find that, 
\begin{align}
    \sum_{j \in \mathcal{S}}\left|\eta \lambda_j \sum_{\tau=0}^{\mathscr{T}-1} x a_\tau^{(j)} e^{(j)}\left(y_{t-1-\tau}\right)\right|^2 < O(\epsilon^2),
\end{align}
\end{proof}
\begin{lemma} \label{lem:aux3}
Suppose that $\max_\tau \left\|\mathbf{e}\left(\mathbf{y}_\tau\right)\right\|<\epsilon /c \sqrt{ \mathscr{T}}$, we have
    \begin{align}
        \sum_{j \in \mathcal{S}}\left|\eta \lambda_j^{\frac{1}{2}} \sum_{\tau=0}^{t-1}\left(a_\tau-a_{\tau-1}\right) e^{(j)}\left(y_{t-1-\tau}\right)\right|^2<\frac{1}{3} \frac{\epsilon^2 \ell}{c^2},
    \end{align}
    where index $j$ sums over all eigen-direction of $\mathcal{H}=\nabla^2 f(\mathbf{0})$ satisfying $\lambda_j \in\left(\theta^2 /\left[\eta(2-\theta)^2\right], \ell\right]$.
\end{lemma}
\begin{proof}
\begin{align*}
\sum_{j \in \mathcal{S}}\left|\eta \lambda_j^{\frac{1}{2}} \sum_{\tau=0}^{t-1}\left(a_\tau-a_{\tau-1}\right) e^{(j)}\left(y_{t-1-\tau}\right)\right|^2 & \leq \mathscr{T} \sum_{j \in \mathcal{S}} \sum_{\tau=0}^{\mathscr{T}-1}\left(\eta \lambda_j^{\frac{1}{2}}\left(a_\tau-a_{\tau-1}\right) e^{(j)}\left(y_{t-1-\tau}\right)\right)^2 \\
& \leq \mathscr{T} \sum_{\tau=0}^{\mathscr{T}-1} \sum_{j \in \mathcal{S}}\left(\max _j \eta \lambda_j^{\frac{1}{2}}\left(a_\tau-a_{\tau-1}\right)\right)^2\left(e^{(j)}\left(y_{t-1-\tau}\right)\right)^2 \\
& =\mathscr{T} \sum_{\tau=0}^{\mathscr{T}-1}\left(\max _j \eta \lambda_j^{\frac{1}{2}}\left(a_\tau-a_{\tau-1}\right)\right)^2 \sum_{j \in \mathcal{S}}\left(e^{(j)}\left(y_{t-1-\tau}\right)\right)^2 \\
& \leq \mathscr{T}\left(\sum_{\tau=0}^{\mathscr{T}-1}\left(\max _j \eta \lambda_j^{\frac{1}{2}}\left(a_\tau-a_{\tau-1}\right)\right)^2\right) \max_\tau \left\|\mathbf{e}\left(\mathbf{y}_\tau\right)\right\|^2 \\
& \leq \eta\epsilon^2/c^2\left(\sum_{\tau=0}^{\mathscr{T}-1}\left(\max _j \sqrt{x}\left(a_\tau-a_{\tau-1}\right)\right)^2\right).
\end{align*}
We notice that
\begin{align}
    \sqrt{x}\left(a_\tau-a_{\tau-1}\right) = \sqrt{x}(\frac{r \cos \phi-1}{r \sin \phi} \cdot r^\tau \sin [\tau \phi]+r^\tau \cos [\tau \phi]),
\end{align}
and the second term can be bounded as follows
\begin{align}
    |\sqrt{x}r^\tau \cos [\tau \phi]| \le \sqrt{x}r^{\tau} \le \sqrt{(1-\theta)^\tau}\sqrt{x(1-x)^\tau} \le \sqrt{\frac{(1-\theta)^\tau}{1+\tau}},
\end{align}
where we make use of $x(1-x)^\tau \le \frac{1}{1+\tau}(1-\frac{1}{1+\tau})^\tau \le \frac{1}{1+\tau}$.

For the first term, we also split $x \in\left(\frac{\theta^2}{(2-\theta)^2}, \frac{1}{4}\right]$ into two cases

Case 1: $x \in\left(\frac{\theta^2}{(2-\theta)^2}, \frac{2 \theta^2}{(2-\theta)^2}\right]$.
\begin{align}
    |\sqrt{x}\frac{r \cos \phi-1}{r \sin \phi} \cdot r^\tau \sin [\tau \phi]| \le O(\sqrt{x}(\theta+x)  r^{\tau -1}\tau) \le O(\theta^2 \tau (1-\theta)^{\tau-1}),
\end{align}
where we use $|\sin [\tau \phi]/\sin \phi| \le \tau$,$r \le O(1-\theta)$ and $|r \cos \phi-1|\le O(\theta+x)$.

Case 2: $x \in\left(\frac{2 \theta^2}{(2-\theta)^2}, \frac{1}{4}\right]$.
\begin{align}
    \left|\sqrt{x} \frac{r \cos \phi-1}{r \sin \phi} \cdot r^\tau \sin [\tau \phi]\right| &\le O((\theta+x)r^{\tau-1}) \le O((1-\theta)^{\frac{\tau-1}{2}}(\theta+x)(1-x)^{\frac{\tau-1}{2}}) \\
    &\le O((1-\theta)^{\frac{\tau-1}{2}}\theta) + O(\frac{(1-\theta)^{\tau/2}}{1+\tau}),
\end{align}
where we used $|\sin [\tau \phi]|\le 1$, $x=\Theta\left(\sin ^2 \phi\right)$ , $|r \cos \phi-1| \leq O(\theta+x)$, $(1-x)^{\frac{\tau-1}{2}}<1$ and $x(1-x)^{\frac{\tau-1}{2}} < O(\frac{1}{\tau + 1})$.

Thus we have the following bound,
\begin{align*}
    &\left(\max _j \sqrt{x}\left(a_\tau-a_{\tau-1}\right)\right)^2 \\
    &\le 2 (\left|\sqrt{x} r^\tau \cos [\tau \phi]\right|^2 + \mid \sqrt{x} \frac{r \cos \phi-1}{r \sin \phi} \cdot r^\tau \sin [\tau \phi]|^2) \\
    & \le 2(\frac{(1-\theta)^\tau}{1+\tau} + \max\{|O\left(\theta^2 \tau(1-\theta)^{\tau-1}\right)|^2,|O\left((1-\theta)^{\frac{\tau-1}{2}} \theta\right)+O\left(\frac{(1-\theta)^{\tau / 2}}{1+\tau}\right)|^2\}) \\
    &\le 2(\frac{(1-\theta)^\tau}{1+\tau} + \left|O\left(\theta^2 \tau(1-\theta)^{\tau-1}\right)\right|^2 + 2(|O\left((1-\theta)^{\frac{\tau-1}{2}} \theta\right)|^2+|O\left(\frac{(1-\theta)^{\tau / 2}}{1+\tau}\right)|^2)).
\end{align*}

Similar to the proof of Lemma \ref{lem:aux2}, we obtain
\begin{align}
    \left(\sum_{\tau=0}^{\mathscr{T}-1}\left(\max _j \sqrt{x}\left(a_\tau-a_{\tau-1}\right)\right)^2\right) < O(\theta) < O(1).
\end{align}

\end{proof}

\subsection{Proof of Theorem \ref{theorem:opti}} \label{app:APGD_proof}
\begin{proof}
    Combining \cref{aux:modi_lem_15}, \cref{aux:modi_lem_16}, \cref{aux:modi_lem_17}, \cref{aux:modi_lem_18}, and \cref{aux:modi_lem_5}, we know that when gradient noise satisfies
    \begin{align}
        \max_\tau \left\|\mathbf{e}\left(\mathbf{y}_\tau\right)\right\| &\le \min \{ \frac{1}{c^2} \epsilon,\frac{\rho^{1 / 8}}{\sqrt{2} \ell^{1 / 4} \chi^{1 / 2} c^{3 / 2}} \epsilon^{9 / 8},\frac{\rho^{1 / 8}}{\sqrt{2} \ell^{1 / 4} \chi^{3 / 2} c^{3 / 2}} \epsilon^{9 / 8},\frac{\delta \chi^{-11} c^{-16}}{64 \ell} \frac{\epsilon^3}{\sqrt{d}} \frac{1}{\Delta_f} ,\epsilon \} \\
        &= \frac{\delta  \chi^{-11} c^{-16}}{64 \ell} \frac{\epsilon^3}{\sqrt{d}} \Delta_f,
    \end{align}
Lemma 15, Lemma 16, Lemma 17, Lemma 18, Lemma 5 in \cite{jin2017accelerated} still hold. And thus the proof of theorem 3 in \cite{jin2017accelerated} applies here. 

    If we want to search an $\epsilon$-first order stationary point instead, only large gradient scenario needs to be taken into consideration, then
    \begin{align}
        \max _\tau\left\|\mathbf{e}\left(\mathbf{y}_\tau\right)\right\| \le
        \min \{ \frac{1}{c^2} \epsilon,\frac{\rho^{1 / 8}}{\sqrt{2} \ell^{1 / 4} \chi^{1 / 2} c^{3 / 2}} \epsilon^{9 / 8},\frac{\rho^{1 / 8}}{\sqrt{2} \ell^{1 / 4} \chi^{3 / 2} c^{3 / 2}} \epsilon^{9 / 8} \} = \frac{\rho^{1 / 8}}{\sqrt{2} \ell^{1 / 4} \chi^{3 / 2} c^{3 / 2}} \epsilon^{9 / 8}.
    \end{align}
\end{proof}

\section{Detailed constructions of time-dependent Lindbladian simulation}
\label{sec:simulaiton-details}
In this section, we present the details of the simulation algorithm for simulating time-dependent Lindblad evolution. Recall in \cref{sec:simulation} that we aim to implement the superoperator defined by \cref{eq:gk}. If we use time ordering operator for this, it becomes
\begin{align}
    \mathcal{G}_{\mathcal{K}}(t) = \mathcal{K}[V(0,t)] + \sum_{k=1}^K \frac{1}{k!}\int_{0\leq s_1\leq t} \int_{0\leq s_2\leq t} \cdots\int_{0\leq s_k\leq t}\mathcal{T}[\mathcal{F}_k(s_k,\dots,s_1)]\dd s_1 \cdots \dd s_k,
\end{align}
where $\mathcal{T}[\mathcal{F}_k(s_k,\dots,s_1)]:= \mathcal{F}_k(\tau_k,\dots,\tau_1)$, such that $\tau_k\leq \dots\leq \tau_1$.
Then applying the Riemann sum, we have
\begin{align} 
\label{eq:Gk2}
\mathcal{G}_{\mathcal{K}}(t) &= \mathcal{K}[V(0,t)] + \sum_{k=1}^K \frac{t^k}{k!q^k}\sum_{j_1,\dots,j_k=0}^q\mathcal{T}\mathcal{F}_k(t_{j_k},\dots,t_{j_1}) \\
    \label{eq:Gk3}
                             &= \mathcal{K}[V(0,t)] + \sum_{k=1}^K \frac{t^k}{q^k}\sum_{0\leq j_1\leq ,\dots,\leq j_k\leq q}\mathcal{F}_k(t_{j_k},\dots,t_{j_1}).
\end{align}
Recall \cref{eq:fk}. The superoperator $\mathcal{F}_k$ is defined as
\begin{align}
    \mathcal{F}_k(s_k,\dots. s_1) = \mathcal{K}[V(s_k,t)] \mathcal{L}_J(s_k)\mathcal{K}[V(s_{k-1},s_k)] \mathcal{L}_J(s_{k-1}) \cdots \mathcal{K}[V(s_1,s_2)]\mathcal{L}_J(s_1)\mathcal{K}[V(0,s_1)]. 
\end{align}
Applying $\mathcal{F}_k$ to $\rho$ yields 
\begin{align}
    \mathcal{F}_k(s_k,\dots. s_1)(\rho) = \sum_{\ell_1,\dots,\ell_k=1}^m [V(s_k,t)L_{\ell_k}(s_k)V(s_{k-1},s_k)L_{\ell_k}(s_{k-1}) \cdots L_{\ell_1}(s_1)V(0,s_1)] \rho \\ [V(s_k,t)L_{\ell_k}(s_k)V(s_{k-1},s_k)L_{\ell_k}(s_{k-1}) \cdots L_{\ell_1}(s_1)V(0,s_1)]^{\dagger}.
\end{align}
To simplify the notation, we define 
\begin{align}
  \label{eq:kraus_A}
    A_{\ell_k,\dots, \ell_1}(s_k,\dots,s_1):= V(s_k,t)L_{\ell_k}(s_k)V(s_{k-1},s_k)L_{\ell_k}(s_{k-1}) \cdots L_{\ell_1}(s_1)V(0,s_1).
\end{align} 
Then $\mathcal{F}_k$ can be rewritten as
\begin{align}
    \mathcal{F}_k(s_k,\dots. s_1)(\rho) = \sum_{\ell_1,\dots,\ell_k=1}^m A_{\ell_k,\dots, \ell_1}(s_k,\dots,s_1) \rho A^{\dagger}_{\ell_k,\dots, \ell_1}(s_k,\dots,s_1).
\end{align}

Now, rewriting \cref{eq:Gk3} using \cref{eq:kraus_A}, we have
\begin{equation}
    \begin{aligned}
        \label{eq:gk1}\mathcal{G}_{\mathcal{K}}(t)(\rho) = V(0,t)\rho V^{\dagger}(0,t) + \sum_{k=1}^K \sum_{0\leq j_1\leq \cdots \leq j_k\leq q} \sum_{\ell_1,\ell_2, \dots, \ell_k=1}^m \frac{t^k}{q^k} A_{\ell_k,\dots, \ell_1}(t_{j_k},\dots,t_{j_1}) \\ \rho   A^{\dagger}_{\ell_k,\dots, \ell_1}(t_{j_k},\dots,t_{j_1}).
    \end{aligned}
\end{equation}
We define the index sets $\mathscr{J}$ as
\begin{align}
    \mathscr{J}:= \{k,\ell_1,\dots,\ell_k,j_1, \dots,j_k: k\in[K], \ell_1,\dots,\ell_k \in [m],j_1,\dots,j_k \in[q]\}.
\end{align}
With $A_{\hat{0}}:= V(0,t)$, and for all $\hat{j} \in \mathscr{J}$, define
\begin{align}
    A_{\hat{j}} := A_{\ell_k,\dots,\ell_1}(j_k t/q,\ldots,j_1 t/q).
\end{align}
Then, \cref{eq:gk1} can be further rewritten as
\begin{align}
    \mathcal{G}_{\mathcal{K}}(t)(\rho) = A_0(t)\rho A_0^{\dagger}(t) + \sum_{\hat{j}\in \mathscr{J}}\frac{t^k}{q^k} A_{\hat{j}}(t)\rho A_{\hat{j}}^{\dagger}(t).
\end{align}
To implement the above completely positive map, we need to first construct the block-encoding $U_{A_{\hat{j}}}$ for each $A_{\hat{j}}$'s, namely,
\begin{align}
    \sum_{\hat{j}}|\hat{j}\rangle
    \langle \hat{j}| \otimes U_{A_{\hat{j}}},
\end{align}
where 
\begin{align}
    U_{A_{\hat{j}}} = \begin{bmatrix}
       A_{\hat{j}}/s_{\hat{j}} & \cdot  \;         \\[0.3em]
       \cdot           & \cdot \;
     \end{bmatrix}.
\end{align}
To apply \cref{lemma:block-encoding-channel}, we also need the following state.
\begin{align}
  \label{eq:mu}
    |\mu\rangle = \sum_{\hat{j}} s_{\hat{j}}|\hat{j}\rangle = \frac{1}{\mathcal{N}}\sum_{k=1}^K \sum_{\ell_1,\dots,\ell_k = 1}^m \sum_{0\leq j_1<\cdots<j_k\leq q} s_{\hat{j}}\sqrt{\frac{t^k}{q^k} }|k\rangle |j_1,\dots,j_k\rangle |\ell_1,\dots, \ell_k\rangle,
\end{align}
where $s_{\ell_k},\ldots,s_{\ell_1}$ are normalization constant for block-encodings of $L_{\ell_k}, \dots, L_{\ell_1}$. It is straightforward that the block-encoding $U_{A_{\hat{j}}}$ can be implemented by the block-encodings of $A_{\ell_k,\dots,\ell_1}\left( j_k t/q,\ldots,j_1 t/q \right)$ what was defined in \cref{eq:kraus_A}.

By the time-dependent Hamiltonian simulation algorithm~\cite{KSB19}, we can simulate $V(s,t)$ by truncated Dyson series. Note that
\begin{align}
    V(s,t):= \mathcal{T} e^{\int_s^t J(\tau)\dd \tau},
\end{align}
where
\begin{align}
    J(t) := -iH(t) - \frac{1}{2} \sum_{j=1}^m L_j^{\dagger}(t) L_j(t).
\end{align}
Assume that we are given an $(\alpha_0(t), a, \epsilon')$-block-encoding $U_{H(t)}$ of $H(t)$, and an $(\alpha_j(t), a, \epsilon')$-block-encoding $U_{L_j(t)}$ for each $L_j(t)$ for all $t \geq 0$. 
That is, we have
\begin{align}
    U_{H(t)} = \sum_t |t\rangle\langle t| \otimes \begin{bmatrix}
       H(t)/\alpha_0(t) & \cdot  \;         \\[0.3em]
       \cdot           & \cdot \;
     \end{bmatrix},
\end{align}
and 
\begin{align}
     U_{L_j(t)} = \sum_t |t\rangle\langle t| \otimes \begin{bmatrix}
       L_j(t)/\alpha_j(t) & \cdot  \;         \\[0.3em]
       \cdot           & \cdot \;
     \end{bmatrix}.
\end{align}
Then by \cref{lemma:sum-to-be}, we can get a $(\|\mathcal{L}(t)\|_{\mathrm{be}},a,2\|\mathcal{L}(t)\|_{\mathrm{be}}\epsilon')$-block-encoding of $J(t)$. Where $\|\mathcal{L}(t)\|_{\mathrm{be}}$ is defined in \cref{eq:be-norm}. We denote this block-encoding by $U_{J(t)}$.

More specifically, we first multiply $U_{L_j^{\dagger}}$ with $U_{L_j}$ and get an $(\alpha_j^2(t),a,2\epsilon')$-block-encoding of each $L_j(t)^{\dagger}L_j(t)$, 
this requires $O(1)$ implementations to each $U_{L_j}(t)$'s and additional $O(m)$ elementary gates. Then we apply \cref{lemma:sum-to-be} with each $L_j(t)^{\dagger}L_j(t)$'s and $H(t)$ to get an $(\alpha_0(t) + \frac{1}{2} \sum_{j=1}^m \alpha_j^2(t),a,(\alpha_0 + \sum_j \alpha_j^2)\epsilon')$-block-encoding of $J(t)$. This requires $O(1)$ implementations of $U_{H(t)}$ and additional $O(1)$ elementary gates. So the total gate complexity for implementing $U_{J(t)}$ is $O(m)$. For any input state $|\psi\rangle$, the effect of $U_{J(t)}$ is
\begin{align}
    U_{J(t)}|0\rangle |t\rangle |\psi\rangle = |0\rangle |t\rangle \frac{J(t)}{\alpha_0(t) + \frac{1}{2} \sum_{j=1}^m \alpha_j^2(t)} |\psi\rangle + |0^{\perp}\rangle,  
\end{align}
or equivalently, we can write it as
\begin{align}
    U_{J(t)} = \sum_t |t\rangle\langle t| \otimes \begin{bmatrix}
       \frac{J(t)}{\alpha_0(t) + \frac{1}{2} \sum_{j=1}^m \alpha_j^2(t)} & \cdot  \;         \\[0.3em]
       \cdot           & \cdot \;
     \end{bmatrix}.
\end{align}
Regarding the normalizing constant and the approximation error, it is a $(\|\mathcal{L}(t)\|_{\mathrm{be}},a,2\|\mathcal{L}(t)\|_{\mathrm{be}}\epsilon')$-block-encoding of $J(t)$.

Noticing that the evolution time has been rescaled, we can assume our $\|\mathcal{L}(t)\|_{\mathrm{be}}$ at each time $t$ to be
\begin{align}
    \|\mathcal{L}(t)\|_{\mathrm{be}} = 1,
\end{align}
and
\begin{align}
    \alpha_0(t) + \frac{1}{2} \sum_{j=1}^m \alpha_j^2(t) \leq 1.
\end{align}
For each $V(s, t)$, consider its Dyson series expansion: 
\begin{align}
\label{operator:V(s,t)}
    \Tilde{V}(s,t) = \sum_{k=1}^{K'} \frac{(t-s)^k}{M^k k!} \sum_{j_1,\dots,j_k=0}^{M-1} \mathcal{T} J(t_{j_k})\cdots J(t_{j_1}).
\end{align}
Then we use the LCU method as in \cite{KSB19} to simulate the above operator, by $O(K')$ invocations of $U_{J(t)}$,
where each invocation requires $O(m)$ gates.
We also need additional $O(K'(\log M + \log K'))$ gates for LCU state preparation. So our total gate complexity for preparing each $\Tilde{V}(s,t)$ is 
\begin{align}
    O(K'(\log M + \log K' + m)).
\end{align}
So the gate complexity for implementing each Kraus operator in \cref{eq:gk1} is 
\begin{align}
    O(KK'(\log M + \log K' + m)).
\end{align}
The normalizing constant for $\Tilde{V}(s,t)$ is then
\begin{align}
    \left(\sum_{k=0}^{K'} \frac{(t-s)^k \|\mathcal{L}(t)\|^k_{\mathrm{be}}}{k!}\right)^2
     \leq  e^{2(t-s)\|\mathcal{L}(t)\|_{\mathrm{be}}} \leq e^{2(t-s)}.
\end{align}

Next, by multiplying the block-encodings of each operator in \cref{eq:kraus_A}, we can get a block-encoding of $A_{\ell_k,\dots, \ell_1}(s_k,\dots,s_1)$ with normalizing constant $s_{\hat{j}}$, which can be bounded by
\begin{align}
    s_{\hat{j}} \leq e^{2(t-s_k)}\alpha_{\ell_k}(t)e^{2(s_k-s_{k-1})}\alpha_{\ell_{k-1}}(t)\cdots \alpha_{\ell_1}(t)e^{2 s_1} = e^{2t} \alpha_{\ell_k}(t)\alpha_{\ell_{k-1}}(t)\cdots \alpha_{\ell_1}(t).
\end{align}
For $t=O(1)$, we have
\begin{align}
    e^{2t} = O(1).
\end{align}
Then $s_{\hat{j}}$ can be bounded by
\begin{align}
    s_{\hat{j}}\leq c_0 \alpha_{\ell_k}(t)\alpha_{\ell_{k-1}}(t)\cdots \alpha_{\ell_1}(t),
\end{align}
where $c_0 = O(1)$ is some constant.

For the third register of \cref{eq:mu}, we need to prepare
\begin{align}
\label{3rd register}
    \frac{1}{\mathcal{N}_1(t)}\sum_{\ell_1,\dots,\ell_k = 1}^m \alpha_{\ell_k}(t)\alpha_{\ell_{k-1}}(t)\cdots \alpha_{\ell_1}(t)|\ell_1,\dots,\ell_k\rangle,
\end{align}
corresponding to
\begin{align}
    \sum_{\ell_1,\dots,\ell_k = 1}^m s_{\hat{j}}|\ell_1,\dots,\ell_k\rangle.
\end{align}
Basically, we prepare each subregister $\alpha_{\ell}|\ell\rangle$
separately, controlled by a unary register $|k\rangle$ (we will discuss how to prepare $|k\rangle$ later). That is we assign $K$ registers for $|\ell_k\rangle$, then we do control state preparation. If the control register is $|k\rangle$, we prepare $\alpha_{\ell_1}|\ell_1\rangle, \alpha_{\ell_2}|\ell_2\rangle,\dots,\alpha_{\ell_k}|\ell_k\rangle$ for only the first $k$ registers, while keeping the rest of them unchanged. 
The normalizing constant $\mathcal{N}_1$ is then
\begin{align}
    \mathcal{N}_1(t) = \left(\sum_{\ell=1}^m \alpha_{\ell}^2(t)\right)^{k/2} \leq (2\|\mathcal{L}(t)\|_{\mathrm{be}})^{k/2}\leq 2^{k/2},
\end{align}
where the last inequality follows from $\|\mathcal{L}(t)\|_{\mathrm{be}} = O(1)$.
The gate complexity for preparing each of $\alpha_{\ell_1}|\ell_1\rangle, \alpha_{\ell_2}|\ell_2\rangle,\dots,\alpha_{\ell_k}|\ell_k\rangle$ is $O(m)$, since $k$ can be as large as $K$, the total time complexity for preparing the third register is $O(Km)$. 

For the second register of \cref{eq:mu}, we aim to prepare the state
\begin{align}
\label{state:second_part}
    \frac{1}{\mathcal{N}_2(t)}\sum_{k=1}^K \sum_{0\leq j_1<\cdots<j_k \leq q} \sqrt{\frac{(2\|\mathcal{L}(t)\|_{\mathrm{be}})^k t^k}{q^k}}|k\rangle |j_1,\dots,j_k\rangle.
\end{align}
We add $(2\|\mathcal{L}(t)\|_{\mathrm{be}})^{k}$ term to absorb the normalizing constant $\mathcal{N}_1(t)$ in our previous step. The normalizing constant $\mathcal{N}_2(t)$ is 
\begin{align}
    \mathcal{N}_2(t) = \sqrt{\sum_{k=1}^K \sum_{0\leq j_1<\cdots<j_k \leq q} \frac{(2\|\mathcal{L}(t)\|_{\mathrm{be}})^k t^k}{q^k}} \leq \sqrt{\sum_{k=1}^K \frac{q^k}{k!}\frac{(2\|\mathcal{L}(t)\|_{\mathrm{be}})^k t^k}{q^k}} \leq e^{t\|\mathcal{L}(t)\|_{\mathrm{be}}}\leq e^t.
\end{align}
For $t = O(1)$, the normalizing constant will be a constant, so the number of steps in the amplitude amplification is $O(1)$. As in \cite{KSB19}, there are two methods to prepare \cref{state:second_part}. In the next subsection, we present the method based on a compressed encoding scheme, and in \cref{sec:qsort}, we present the method based on quantum sorting.

\subsection{State preparation by compressed encoding scheme}
We first initialize $q$ qubits to $|0\rangle ^{\otimes q}$, then we rotate each qubits by a small angle, to give
\begin{equation}
\label{state:compress}
    \begin{aligned}
    \left(\frac{|0\rangle + \sqrt{\zeta}|1\rangle}{\sqrt{1+\zeta}}\right)^{\otimes q} &= (1+\zeta)^{-q/2}\sum_x \zeta^{|x|/2}|x\rangle\\
    &=(1+\zeta)^{-q/2}\sum_{x,|x|\leq K} \zeta^{|x|/2} |x\rangle + (1+\zeta)^{-q/2}\sum_{x,|x|> K} \zeta^{|x|/2} |x\rangle\\   
    &= \sqrt{1-\mu^2}|\mathtt{time}\rangle + \nu |\nu\rangle,
    \end{aligned}
\end{equation}
where $\zeta := \frac{2\|\mathcal{L}\|_{\mathrm{be}}t}{q}$, and
\begin{align}
    \label{state:time}|\mathtt{time}\rangle = \frac{1}{\sqrt{S}} \sum_{|x|\leq K} \zeta^{|x|/2}|x\rangle.
\end{align}
The amplitude $\nu$ satisfies
\begin{align}
    \nu^2 = (1+\zeta)^{-q} \sum_{x,|x|>K}\zeta^{|x|} = O\left(\frac{(2\|\mathcal{L}(t)\|_{\mathrm{be}}t)^{K+1}}{(K+1)!}\right).
\end{align}
For $t\|\mathcal{L}(t)\|_{\mathrm{be}}=O(1)$, by setting $K = O\left(\frac{\log(1/\epsilon)}{\log\log(1/\epsilon)}\right)$, we can get $\nu^2 = O(\epsilon)$. Then we compress each string $x$. For a string $x = 0^{s_1}1 0^{s_2}10^{s_3} \cdots 0^{s_k}10^t$, we represent it with $s_1 s_2\dots s_k$. We define the encoding operator $C_q^K$ as
\begin{align}
    C_q^K|x\rangle :=|s_1,\dots,s_k,q,\dots,q\rangle.
\end{align}
By applying encoding operator $C_q^K$ to [\ref{state:compress}], we get
\begin{align}
    |\Xi_q^K\rangle = \sqrt{1-\mu^2} C_q^K |\mathtt{time}\rangle + \nu |\nu'\rangle,
\end{align}
where $|\nu '\rangle$ is defined as
\begin{align}
    |\nu '\rangle = C_q^K |\nu\rangle.
\end{align}
But this only gives the interval between different time. We still need to compute the absolute time and the Hamming weight for each string $x$. We first count the number times $q$ appears to determine $k$, which has complexity $O(K\log q)$. Then we increment registers $2$ to $k$ to give $|s_1,s_2 +1,\dots,s_k + 1,q,\dots,q\rangle |k\rangle$. Note that, here our $k$ takes unary form. Then we add register $1$ to register $2$, register $2$ to register $3$, and so on to give $|j_1,\dots,j_k,q,\dots,q\rangle |k\rangle$.

However, in practice, we can't construct encoding operator $C_q^K$ by really "counting" $0$'s in state (\ref{state:time}). The detailed construction follows from \cite{kieferova2019simulating}. The idea is that we first prepare a state 
\begin{align}
    |\phi_p\rangle = \sum_{s=0}^{p-1} \beta \alpha^s |s\rangle + \alpha^p|p\rangle,
\end{align}
where $\alpha = 1/\sqrt{1+\zeta}$ and $\beta = \sqrt{\zeta} \alpha$. We choose $p$ to be 
\begin{align}
    \log p = \Theta(\log M + \log\log (1/\delta)).
\end{align}
Taking a tensor product of $K+1$ of the state $|\phi_p\rangle$ gives a state similar to $|\Xi_q^K\rangle$.

According to \cite{KSB19}, to prepare the above state within trace distance $O(\epsilon)$, the total gate complexity for preparing $|\mu\rangle$ is 
\begin{align}
    O(K(m + \log q + \log \log (1/\epsilon))).
\end{align}
Note that the total gate complexity for preparing 
\begin{align}
    \sum_{\hat{j}}|\hat{j}\rangle \langle \hat{j}| \otimes U_{A_{\hat{j}}}
\end{align}
is
\begin{align}
    O(KK'(\log M + \log K' + n)).
\end{align}
So the total gate complexity for simulating (\ref{eq:gk1}) is 
\begin{align}
    O(KK'(\log M + \log q + m + n + \log \log (1/\epsilon))).
\end{align}
And the total query complexity is
\begin{align}
    O(KK').
\end{align}
To make our error to be within $\epsilon$, also by our analysis in the main text, it suffices to take
\begin{align}
    K = \frac{\log(1/\epsilon)}{\log\log (1/\epsilon)},
\end{align}
\begin{align}
    q = \Theta\left(\frac{2K}{\epsilon}\left(4J_{\max} + 2\sum_{j=1}^m \Dot{L}_{j,\max}\right)\right),
\end{align}
\begin{align}
    K' = \frac{\log(1/\epsilon)}{\log\log(1/\epsilon)},
\end{align}
\begin{align}
    M = \Theta\left(\frac{\Dot{J}_{\max}}{\epsilon}\right).
\end{align}
By substituting these back, we can rewrite the total gate complexity as
\begin{align}
    O\left(\left(\frac{\log(1/\epsilon)}{\log\log(1/\epsilon)}\right)^2\cdot \left(m+n+\log(1/\epsilon)\right)\right).
\end{align}
\subsection{State Preparation by Quantum Sort}\label{sec:qsort}
We use the scheme in \cite{KSB19} to sort the time index. For quantum sort scheme, we use \cref{eq:Gk2} to simulate our target operator. Now we define $A_j$ as
\begin{align}
    A_j = \sqrt{\frac{t^k}{k!q^k}} A_{l_k,\dots,l_1}(\frac{j_k t}{q},\dots, \frac{j_1 t}{q}).
\end{align}
Then we use quantum sort scheme to prepare $|\mu\rangle$, such that
\begin{align}
    |\mu\rangle = \sum_{k=1}^K \sum_{\ell_1,\dots,\ell_k = 1}^m \sum_{j_1,\dots,j_k = 0}^q \sqrt{\frac{t^k}{q^k} }s_{\hat{j}}|k\rangle |\ell_1,\dots, \ell_k\rangle \mathrm{SORT}|j_1,\dots,j_k\rangle .
\end{align}
Here $\mathrm{SORT}$ operator is defined as
\begin{align}
    \mathrm{SORT}|\tau_1, \dots,\tau_k\rangle = |s_1, \dots,s_k\rangle,
\end{align}
such that $s_1\leq s_2\leq \cdots \leq s_k$. 
To discuss the details for preparing $|\mu\rangle$, first we need to prepare
\begin{align}
    \frac{1}{\mathcal{N}}\sum_{\ell_1,\dots,\ell_k = 1}^m \alpha_{\ell_k} \alpha_{\ell_{k-1}}\cdots \alpha_{\ell_1} |\ell_1,\dots,\ell_k\rangle.
\end{align}
The state preparation process is the same as in Compression Scheme. Its gate complexity is $O(Km)$. For the rest of them, we first prepare a superposition over $k$ in unary form
\begin{align}
    |0\rangle^{\otimes K} \mapsto \frac{1}{\sqrt{s}} \sum_{k=0}^K \sqrt{\frac{t^k}{k!}} |1^k 0^{K-k}\rangle.
\end{align}
Then we do control rotation to $K$ subregisters, that is we do the following
\begin{align}
    \frac{1}{\sqrt{s}} \sum_{k=0}^K \sqrt{\frac{t^k}{k!}} |1^k 0^{K-k}\rangle |0\rangle \mapsto \frac{1}{\sqrt{s}} \sum_{k=0}^K \sqrt{\frac{t^k}{q^kk!}} |1^k 0^{K-k}\rangle \sum_{j_1=0}^{q-1}\sum_{j_2=0}^{q-1}\cdots \sum_{j_k=0}^{q-1}|j_1,j_2,\dots,j_k\rangle.
\end{align}
The gate complexity for the above operation is $O(K\log q)$. Then we apply $\mathrm{SORT}$ to $|j_1,j_2,\dots,j_k\rangle$ as in \cite{KSB19}. We use a sequence of comparison operation $\mathrm{COMPARE}$ which acts on two multi-qubit registers storing the values $q_1$ and $q_2$ and an ancillary qubit initialized to $|0\rangle$
\begin{align}
    \mathrm{COMPARE}|q_1\rangle |q_2\rangle |0\rangle = |q_1\rangle |q_2\rangle |\theta(q_1 - q_2)\rangle,
\end{align}
where $\theta$ is the Heaviside step function ($\theta(0) = 0$). Then we apply a $\mathrm{SWAP}$ gate to $|q_1\rangle$ and $|q_2\rangle$ controlled by our ancillary qubit. Then the overall effect is, if $q_1\leq q_2$, we do nothing; otherwise, we swap the first two registers. Then we apply such sorting network to all registers $|j_1\rangle,\dots,|j_k\rangle$ and by certain sorting algorithm, we make $j_1\leq j_2\leq \cdots \leq j_k$. We also swap qubits in register $|k\rangle$ to keep their correspondence. Readers may refer to \cite{kieferova2019simulating} for an example of a sorting circuit. 

The gate complexity for quantum sort network is $O(K\log K)$. So the total gate complexity for preparing $\mu$ is $O(K(\log q+ \log K + m))$. Since we take $K = O\left(\frac{\log(1/\epsilon)}{\log\log(1/\epsilon)}\right)$, the gate complexity is the same as the compressed encoding scheme.

\end{document}